\documentclass[10pt,a4paper,parskip=half,numbers=noenddot,abstract=true]{scrartcl}

\usepackage[margin=1in]{geometry}

\usepackage{amsmath,amssymb,amsthm}
\usepackage{mathtools}

\usepackage[T1]{fontenc}
\usepackage{libertinus}
\usepackage[scaled=0.92]{FiraSans}
\usepackage{microtype}
\usepackage{xstring}

\addtokomafont{disposition}{\sffamily}
\addtokomafont{title}{\sffamily}
\addtokomafont{subtitle}{\sffamily}
\addtokomafont{author}{\sffamily}
\addtokomafont{date}{\sffamily}
\setkomafont{section}{\sffamily\Large\bfseries}
\setkomafont{subsection}{\sffamily\large\bfseries}
\setkomafont{subsubsection}{\sffamily\normalsize\bfseries}

\usepackage{graphicx}
\usepackage{booktabs}
\usepackage{makecell}
\usepackage{float}
\usepackage[section]{placeins}
\usepackage{xcolor}

\usepackage[numbers,sort&compress]{natbib}
\usepackage{hyperref}
\usepackage{xurl}
\hypersetup{colorlinks=true,linkcolor=blue!50!black,citecolor=blue!50!black,urlcolor=blue!50!black,
  pdftitle={FRESH: Information-Geometric Calibration of Patient-Level Models to Aggregate Evidence},
  pdfauthor={F. Fuller, D. Bertolini, S. Liang, J. Christopher, A. Smith}}

\usepackage{algorithm}
\usepackage{algpseudocode}

\usepackage[affil-it]{authblk}

\newcommand{\ebHr}{0.867}
\newcommand{\ebHrLo}{0.709}
\newcommand{\ebHrHi}{1.077}

\newcommand{\ebDrOne}{22.1}
\newcommand{\ebDrOneLo}{-4.0}
\newcommand{\ebDrOneHi}{43.5}
\newcommand{\ebDrOnePos}{94.9}
\newcommand{\ebDrTwo}{34.4}
\newcommand{\ebDrTwoLo}{-10.6}
\newcommand{\ebDrTwoHi}{79.5}
\newcommand{\ebDrTwoPos}{92.9}

\newcommand{\ebStwlvLo}{0.301}
\newcommand{\ebStwlvHi}{0.430}

\newcommand{\ebStwfrLo}{0.060}
\newcommand{\ebStwfrHi}{0.134}

\newcommand{\ebMedOsRwDaysLo}{248}
\newcommand{\ebMedOsRwDaysHi}{309}

\newcommand{\ebNativeGapDays}{79}
\newcommand{\ebMpactPredDaysDec}{257.1}
\newcommand{\ebProdigePredDaysDec}{336.7}
\newcommand{\ebMedOsRwDaysDec}{284.3}
\newcommand{\ebShiftDaysDec}{27.2}
\newcommand{\ebGapDaysDec}{79.6}
\newcommand{\ebTreatDaysDec}{52.4}
\newcommand{\ebHrBelowOne}{90.1}
\newcommand{\ebMedOsRwBootDays}{269}
\newcommand{\ebMedOsRwBootMos}{8.8}
\newcommand{\ebStwlvBoot}{0.358}
\newcommand{\ebStwfrBoot}{0.094}
\newcommand{\ebIpdRoundTrip}{0.009}
\newcommand{\crcNumFolds}{5}
\newcommand{\crcNumReps}{100}
\newcommand{\crcNumCells}{500}

\newcommand{\crcNumChains}{1{,}500}
\newcommand{\crcOsMedPub}{15.1}
\newcommand{\crcOsMedPubLo}{13.7}
\newcommand{\crcOsMedPubHi}{17.7}
\newcommand{\crcOsMedPred}{14.49}
\newcommand{\crcOsMedLo}{13.00}
\newcommand{\crcOsMedHi}{15.65}
\newcommand{\crcOsMedSigFold}{0.768}
\newcommand{\crcOsMedSigSrc}{0.387}

\newcommand{\crcOsMedShareFold}{76}
\newcommand{\crcOsMedShareSrc}{24}
\newcommand{\crcOsMedPredInPub}{yes}
\newcommand{\crcOsMedPubInPred}{yes}

\newcommand{\crcOsMedFoldMin}{13.28}
\newcommand{\crcOsMedFoldMax}{15.19}

\newcommand{\crcOsSixPub}{0.874}
\newcommand{\crcOsSixPubLo}{0.828}
\newcommand{\crcOsSixPubHi}{0.913}
\newcommand{\crcOsSixPred}{0.866}
\newcommand{\crcOsSixLo}{0.818}
\newcommand{\crcOsSixHi}{0.901}
\newcommand{\crcOsSixSigFold}{0.024}
\newcommand{\crcOsSixSigSrc}{0.009}

\newcommand{\crcOsSixShareFold}{85}
\newcommand{\crcOsSixShareSrc}{15}
\newcommand{\crcOsSixPredInPub}{yes}
\newcommand{\crcOsSixPubInPred}{yes}
\newcommand{\crcOsSixExcess}{-0.8}

\newcommand{\crcOsTwelvePub}{0.660}
\newcommand{\crcOsTwelvePubLo}{0.595}
\newcommand{\crcOsTwelvePubHi}{0.718}
\newcommand{\crcOsTwelvePred}{0.601}
\newcommand{\crcOsTwelveLo}{0.556}
\newcommand{\crcOsTwelveHi}{0.642}
\newcommand{\crcOsTwelveSigFold}{0.022}
\newcommand{\crcOsTwelveSigSrc}{0.015}

\newcommand{\crcOsTwelveShareFold}{63}
\newcommand{\crcOsTwelveShareSrc}{37}
\newcommand{\crcOsTwelvePredInPub}{yes}
\newcommand{\crcOsTwelvePubInPred}{no}
\newcommand{\crcOsTwelveExcess}{-5.9}

\newcommand{\crcOsEighteenPub}{0.422}
\newcommand{\crcOsEighteenPubLo}{0.350}
\newcommand{\crcOsEighteenPubHi}{0.487}
\newcommand{\crcOsEighteenPred}{0.373}
\newcommand{\crcOsEighteenLo}{0.326}
\newcommand{\crcOsEighteenHi}{0.409}
\newcommand{\crcOsEighteenSigFold}{0.020}
\newcommand{\crcOsEighteenSigSrc}{0.014}

\newcommand{\crcOsEighteenShareFold}{63}
\newcommand{\crcOsEighteenShareSrc}{37}
\newcommand{\crcOsEighteenPredInPub}{yes}
\newcommand{\crcOsEighteenPubInPred}{no}
\newcommand{\crcOsEighteenExcess}{-4.9}

\newcommand{\crcOsTwentyFourPub}{0.290}
\newcommand{\crcOsTwentyFourPubLo}{0.221}
\newcommand{\crcOsTwentyFourPubHi}{0.363}
\newcommand{\crcOsTwentyFourPred}{0.224}
\newcommand{\crcOsTwentyFourLo}{0.180}
\newcommand{\crcOsTwentyFourHi}{0.262}
\newcommand{\crcOsTwentyFourSigFold}{0.020}
\newcommand{\crcOsTwentyFourSigSrc}{0.012}

\newcommand{\crcOsTwentyFourShareFold}{70}
\newcommand{\crcOsTwentyFourShareSrc}{30}
\newcommand{\crcOsTwentyFourPredInPub}{yes}
\newcommand{\crcOsTwentyFourPubInPred}{no}
\newcommand{\crcOsTwentyFourExcess}{-6.6}

\newcommand{\crcOsThirtyPub}{0.196}
\newcommand{\crcOsThirtyPubLo}{0.128}
\newcommand{\crcOsThirtyPubHi}{0.273}
\newcommand{\crcOsThirtyPred}{0.129}
\newcommand{\crcOsThirtyLo}{0.100}
\newcommand{\crcOsThirtyHi}{0.163}
\newcommand{\crcOsThirtySigFold}{0.015}
\newcommand{\crcOsThirtySigSrc}{0.009}

\newcommand{\crcOsThirtyShareFold}{69}
\newcommand{\crcOsThirtyShareSrc}{31}
\newcommand{\crcOsThirtyPredInPub}{yes}
\newcommand{\crcOsThirtyPubInPred}{no}
\newcommand{\crcOsThirtyExcess}{-6.7}

\newcommand{\crcPfsMedPub}{7.1}
\newcommand{\crcPfsMedPubLo}{6.8}
\newcommand{\crcPfsMedPubHi}{8.5}
\newcommand{\crcPfsMedPred}{7.72}
\newcommand{\crcPfsMedLo}{7.26}
\newcommand{\crcPfsMedHi}{8.15}
\newcommand{\crcPfsMedSigFold}{0.182}
\newcommand{\crcPfsMedSigSrc}{0.171}

\newcommand{\crcPfsMedShareFold}{48}
\newcommand{\crcPfsMedShareSrc}{52}
\newcommand{\crcPfsMedPredInPub}{yes}
\newcommand{\crcPfsMedPubInPred}{no}

\newcommand{\crcPfsMedLoMinusPub}{0.2}
\newcommand{\crcPfsSixPub}{0.610}
\newcommand{\crcPfsSixPubLo}{0.537}
\newcommand{\crcPfsSixPubHi}{0.676}
\newcommand{\crcPfsSixPred}{0.640}
\newcommand{\crcPfsSixLo}{0.607}
\newcommand{\crcPfsSixHi}{0.670}
\newcommand{\crcPfsSixSigFold}{0.008}
\newcommand{\crcPfsSixSigSrc}{0.015}

\newcommand{\crcPfsSixShareFold}{20}
\newcommand{\crcPfsSixShareSrc}{80}
\newcommand{\crcPfsSixPredInPub}{yes}
\newcommand{\crcPfsSixPubInPred}{yes}
\newcommand{\crcPfsSixExcess}{+3.0}

\newcommand{\crcPfsTwelvePub}{0.292}
\newcommand{\crcPfsTwelvePubLo}{0.218}
\newcommand{\crcPfsTwelvePubHi}{0.365}
\newcommand{\crcPfsTwelvePred}{0.217}
\newcommand{\crcPfsTwelveLo}{0.183}
\newcommand{\crcPfsTwelveHi}{0.248}
\newcommand{\crcPfsTwelveSigFold}{0.015}
\newcommand{\crcPfsTwelveSigSrc}{0.011}

\newcommand{\crcPfsTwelveShareFold}{57}
\newcommand{\crcPfsTwelveShareSrc}{43}
\newcommand{\crcPfsTwelvePredInPub}{no}
\newcommand{\crcPfsTwelvePubInPred}{no}
\newcommand{\crcPfsTwelveExcess}{-7.5}

\newcommand{\crcPfsEighteenPub}{0.191}
\newcommand{\crcPfsEighteenPubLo}{0.127}
\newcommand{\crcPfsEighteenPubHi}{0.266}
\newcommand{\crcPfsEighteenPred}{0.085}
\newcommand{\crcPfsEighteenLo}{0.063}
\newcommand{\crcPfsEighteenHi}{0.108}
\newcommand{\crcPfsEighteenSigFold}{0.011}
\newcommand{\crcPfsEighteenSigSrc}{0.008}

\newcommand{\crcPfsEighteenShareFold}{62}
\newcommand{\crcPfsEighteenShareSrc}{38}
\newcommand{\crcPfsEighteenPredInPub}{no}
\newcommand{\crcPfsEighteenPubInPred}{no}
\newcommand{\crcPfsEighteenExcess}{-10.6}

\newcommand{\crcPfsTwentyFourPub}{0.139}
\newcommand{\crcPfsTwentyFourPubLo}{0.069}
\newcommand{\crcPfsTwentyFourPubHi}{0.234}
\newcommand{\crcPfsTwentyFourPred}{0.052}
\newcommand{\crcPfsTwentyFourLo}{0.039}
\newcommand{\crcPfsTwentyFourHi}{0.070}
\newcommand{\crcPfsTwentyFourSigFold}{0.007}
\newcommand{\crcPfsTwentyFourSigSrc}{0.006}

\newcommand{\crcPfsTwentyFourShareFold}{49}
\newcommand{\crcPfsTwentyFourShareSrc}{51}
\newcommand{\crcPfsTwentyFourPredInPub}{no}
\newcommand{\crcPfsTwentyFourPubInPred}{no}

\newcommand{\crcNumBenchmarks}{11}
\newcommand{\crcNumPredInPub}{8}
\newcommand{\crcNumPubInPred}{3}
\newcommand{\crcOsShareFoldMin}{63}
\newcommand{\crcOsShareFoldMax}{85}
\newcommand{\crcOsNumBenchmarks}{6}
\newcommand{\crcOsNumPredInPub}{6}
\newcommand{\crcOsNumPubInPred}{2}
\newcommand{\crcPfsShareFoldMin}{20}
\newcommand{\crcPfsShareFoldMax}{62}
\newcommand{\crcPfsNumBenchmarks}{5}
\newcommand{\crcPfsNumPredInPub}{2}
\newcommand{\crcPfsNumPubInPred}{1}

\newcommand{\crcTribeFcOs}{16.8}

\newcommand{\crcNoShare}{73}

\newcommand{\crcXeloxLmDev}{0.008}

\newcommand{\crcXeloxOsPub}{19.8}
\newcommand{\crcXeloxOsCal}{19.7}
\newcommand{\crcXeloxOsCalLo}{18.7}
\newcommand{\crcXeloxOsCalHi}{20.6}

\newcommand{\crcXeloxOsPipd}{19.8}
\newcommand{\crcXeloxOsPipdLo}{18.8}
\newcommand{\crcXeloxOsPipdHi}{20.7}

\newcommand{\crcXeloxOsPipdSd}{0.51}
\newcommand{\crcXeloxOsWidthRatio}{1.0}

\newcommand{\crcXeloxPfsPub}{8.0}
\newcommand{\crcXeloxPfsCal}{8.0}
\newcommand{\crcXeloxPfsCalLo}{7.6}
\newcommand{\crcXeloxPfsCalHi}{8.5}

\newcommand{\crcXeloxPfsPipd}{8.0}
\newcommand{\crcXeloxPfsPipdLo}{7.6}
\newcommand{\crcXeloxPfsPipdHi}{8.5}

\newcommand{\crcXeloxPfsWidthRatio}{1.0}

\newcommand{\crcXeloxPfsMedDev}{-0.3}
\newcommand{\crcXeloxOsMedDev}{-1.8}

\newcommand{\crcFolfoxLmDev}{0.008}

\newcommand{\crcFolfoxOsPub}{19.5}
\newcommand{\crcFolfoxOsCal}{19.5}
\newcommand{\crcFolfoxOsCalLo}{18.5}
\newcommand{\crcFolfoxOsCalHi}{20.6}

\newcommand{\crcFolfoxOsPipd}{19.6}
\newcommand{\crcFolfoxOsPipdLo}{18.5}
\newcommand{\crcFolfoxOsPipdHi}{20.5}

\newcommand{\crcFolfoxOsPipdSd}{0.54}
\newcommand{\crcFolfoxOsWidthRatio}{1.1}

\newcommand{\crcFolfoxPfsPub}{8.5}
\newcommand{\crcFolfoxPfsCal}{8.4}
\newcommand{\crcFolfoxPfsCalLo}{8.0}
\newcommand{\crcFolfoxPfsCalHi}{8.8}

\newcommand{\crcFolfoxPfsPipd}{8.5}
\newcommand{\crcFolfoxPfsPipdLo}{8.0}
\newcommand{\crcFolfoxPfsPipdHi}{8.9}

\newcommand{\crcFolfoxPfsWidthRatio}{1.0}

\newcommand{\crcFolfoxPfsMedDev}{-3.6}
\newcommand{\crcFolfoxOsMedDev}{+0.4}

\newcommand{\crcTribeLmDev}{0.005}

\newcommand{\crcTribeOsPub}{29.8}
\newcommand{\crcTribeOsPubLo}{26.0}
\newcommand{\crcTribeOsPubHi}{34.3}
\newcommand{\crcTribeOsCal}{29.7}
\newcommand{\crcTribeOsCalLo}{26.1}
\newcommand{\crcTribeOsCalHi}{34.0}

\newcommand{\crcTribeOsPipdLo}{26.1}
\newcommand{\crcTribeOsPipdHi}{34.3}

\newcommand{\crcTribeOsPipdSd}{1.89}
\newcommand{\crcTribeOsWidthRatio}{1.0}

\newcommand{\crcTribePfsPub}{12.1}
\newcommand{\crcTribePfsPubLo}{10.9}
\newcommand{\crcTribePfsPubHi}{13.2}
\newcommand{\crcTribePfsCal}{12.1}
\newcommand{\crcTribePfsCalLo}{10.7}
\newcommand{\crcTribePfsCalHi}{13.2}

\newcommand{\crcTribePfsPipdLo}{10.8}
\newcommand{\crcTribePfsPipdHi}{13.2}

\newcommand{\crcTribePfsWidthRatio}{1.1}

\newcommand{\crcTribeOsMedDev}{-4.3}
\newcommand{\crcTribePfsMedDev}{-1.1}
\newcommand{\crcNoRepSdMin}{0.48}
\newcommand{\crcNoRepSdMax}{0.50}
\newcommand{\crcTribeRepSd}{1.96}

\newcommand{\crcWidthRatioMin}{1.0}
\newcommand{\crcWidthRatioMax}{1.1}
\newcommand{\crcFoldSdMin}{0.002}
\newcommand{\crcFoldSdMax}{0.005}
\newcommand{\crcUnitFoldShareMax}{10}

\newcommand{\crcNoPipdGapAbsMax}{0.1}

\newcommand{\crcTribePipdCiMaxDiff}{0.1}
\newcommand{\crcWXeloxN}{378{,}837}
\newcommand{\crcWXeloxEss}{249{,}749}
\newcommand{\crcWXeloxEff}{0.66}
\newcommand{\crcWXeloxMax}{3.40}
\newcommand{\crcWXeloxTopFive}{15.5}
\newcommand{\crcWXeloxTopTen}{27.7}
\newcommand{\crcWXeloxQone}{0.367}
\newcommand{\crcWXeloxQfive}{0.367}
\newcommand{\crcWXeloxQtwentyfive}{0.498}
\newcommand{\crcWXeloxQfifty}{0.793}
\newcommand{\crcWXeloxQseventyfive}{1.075}
\newcommand{\crcWXeloxQninetyfive}{2.925}
\newcommand{\crcWXeloxQninetynine}{3.395}
\newcommand{\crcWFolfoxN}{379{,}085}
\newcommand{\crcWFolfoxEss}{337{,}212}
\newcommand{\crcWFolfoxEff}{0.89}
\newcommand{\crcWFolfoxMax}{1.69}
\newcommand{\crcWFolfoxTopFive}{8.4}
\newcommand{\crcWFolfoxTopTen}{16.6}
\newcommand{\crcWFolfoxQone}{0.553}
\newcommand{\crcWFolfoxQfive}{0.553}
\newcommand{\crcWFolfoxQtwentyfive}{0.712}
\newcommand{\crcWFolfoxQfifty}{1.041}
\newcommand{\crcWFolfoxQseventyfive}{1.342}
\newcommand{\crcWFolfoxQninetyfive}{1.687}
\newcommand{\crcWFolfoxQninetynine}{1.687}
\newcommand{\crcWTribeN}{180{,}008}
\newcommand{\crcWTribeEss}{68{,}465}
\newcommand{\crcWTribeEff}{0.38}
\newcommand{\crcWTribeMax}{5.68}
\newcommand{\crcWTribeTopFive}{25.1}
\newcommand{\crcWTribeTopTen}{45.0}
\newcommand{\crcWTribeQone}{0.111}
\newcommand{\crcWTribeQfive}{0.111}
\newcommand{\crcWTribeQtwentyfive}{0.111}
\newcommand{\crcWTribeQfifty}{0.952}
\newcommand{\crcWTribeQseventyfive}{0.952}
\newcommand{\crcWTribeQninetyfive}{3.979}
\newcommand{\crcWTribeQninetynine}{5.680}
\newcommand{\crcWBwN}{26{,}263}
\newcommand{\crcWBwEss}{4{,}208}
\newcommand{\crcWBwEff}{0.16}
\newcommand{\crcWBwMax}{72.02}
\newcommand{\crcWBwTopFive}{39.4}
\newcommand{\crcWBwTopTen}{53.1}
\newcommand{\crcWBwQone}{0.030}
\newcommand{\crcWBwQfive}{0.041}
\newcommand{\crcWBwQtwentyfive}{0.168}
\newcommand{\crcWBwQfifty}{0.395}
\newcommand{\crcWBwQseventyfive}{0.992}
\newcommand{\crcWBwQninetyfive}{3.580}
\newcommand{\crcWBwQninetynine}{9.326}
\newcommand{\crcWScaffMaxOverMean}{5.7}
\newcommand{\crcWScaffTopFive}{25.1}

\newcommand{\crcMhXeloxNChains}{500}
\newcommand{\crcMhXeloxItersMin}{6{,}000}
\newcommand{\crcMhXeloxItersMed}{6{,}000}
\newcommand{\crcMhXeloxItersMax}{8{,}000}
\newcommand{\crcMhXeloxNEarly}{500}
\newcommand{\crcMhXeloxHardMed}{0.010}
\newcommand{\crcMhXeloxHardMax}{0.015}

\newcommand{\crcMhFolfoxNChains}{500}
\newcommand{\crcMhFolfoxItersMin}{6{,}000}
\newcommand{\crcMhFolfoxItersMed}{6{,}000}
\newcommand{\crcMhFolfoxItersMax}{9{,}000}
\newcommand{\crcMhFolfoxNEarly}{500}
\newcommand{\crcMhFolfoxHardMed}{0.012}
\newcommand{\crcMhFolfoxHardMax}{0.019}

\newcommand{\crcMhTribeNChains}{500}
\newcommand{\crcMhTribeItersMin}{6{,}000}
\newcommand{\crcMhTribeItersMed}{6{,}000}
\newcommand{\crcMhTribeItersMax}{10{,}000}
\newcommand{\crcMhTribeNEarly}{499}
\newcommand{\crcMhTribeHardMed}{0.007}
\newcommand{\crcMhTribeHardMax}{0.020}
\newcommand{\crcMhIters}{10{,}000}
\newcommand{\crcMhBurnIn}{1{,}000}
\newcommand{\crcMhSpacing}{1{,}000}
\newcommand{\crcMhTailSamples}{1}
\newcommand{\crcMhAlpha}{0.1}

\newcommand{\crcMhEpsilonOs}{30}
\newcommand{\crcMhEpsilonPfs}{10}
\newcommand{\crcMhGammaZero}{1.0}
\newcommand{\crcMhRmA}{0.6}
\newcommand{\crcMhTZero}{1}
\newcommand{\crcMhDeltaMax}{0.001}
\newcommand{\crcMhStopTol}{0.02}
\newcommand{\crcMhStopChecks}{5}
\newcommand{\crcMhLogInterval}{1{,}000}
\newcommand{\crcCkptSelected}{500/400/400/500/500}
\newcommand{\crcCkptMinLoss}{200/250/200/200/250}
\newcommand{\crcCkptBand}{5}
\newcommand{\crcCkptLossRiseMin}{1.6}
\newcommand{\crcCkptLossRiseMax}{4.8}
\newcommand{\crcCkptNVsixMin}{15}
\newcommand{\crcCkptNVsixMax}{26}
\newcommand{\crcCkptRatioMinLoss}{0.88}
\newcommand{\crcCkptRatioSelected}{0.64}
\newcommand{\crcCkptRatioObserved}{0.45}
\newcommand{\crcCkptCMinLoss}{0.672}
\newcommand{\crcCkptCSelected}{0.689}
\newcommand{\crcCurveSurvDiffMax}{0.006}
\newcommand{\crcCurveOsCross}{14.70}
\newcommand{\crcCurvePfsCross}{7.75}
\newcommand{\crcFtN}{2{,}355}
\newcommand{\crcFtNVsix}{164}
\newcommand{\crcFtHeldMin}{443}
\newcommand{\crcFtHeldMax}{503}
\newcommand{\crcFtEpochs}{20}
\newcommand{\crcFtBatch}{32}
\newcommand{\crcFtStepsMin}{1{,}160}
\newcommand{\crcFtStepsMax}{1{,}200}
\newcommand{\crcFtStep}{600}
\newcommand{\crcFtLossRiseMin}{2.5}
\newcommand{\crcFtLossRiseMax}{3.8}
\newcommand{\crcFtLossRiseMean}{3.2}
\newcommand{\crcFtCStart}{0.695}
\newcommand{\crcFtCSelected}{0.688}

\newcommand{\crcFtTribeOsRatio}{0.56}

\newcommand{\crcFtTribePfsRatio}{0.73}

\newcommand{\crcFtNoOsRatio}{0.59}

\newcommand{\crcFtNoPfsRatio}{0.72}

\newcommand{\crcPfsSixExcessInv}{+6.5}
\newcommand{\crcPfsTwelveExcessInv}{-4.6}
\newcommand{\crcPfsEighteenExcessInv}{-3.1}

\theoremstyle{plain}
\newtheorem{theorem}{Theorem}[section]
\newtheorem{proposition}[theorem]{Proposition}

\theoremstyle{remark}
\newtheorem{remark}[theorem]{Remark}

\newcommand{\Y}{\mathcal{Y}}
\newcommand{\X}{\mathcal{X}}
\newcommand{\Qref}{Q_{0}}
\newcommand{\Qstar}{Q^{\ast}}
\newcommand{\R}{\mathbb{R}}
\newcommand{\E}{\mathbb{E}}
\newcommand{\KL}{D_{\mathrm{KL}}}
\newcommand{\Kref}{K_{0}}
\newcommand{\Kstar}{K^{\ast}}
\newcommand{\ind}[1]{\mathbf{1}\!\left\{#1\right\}}
\newcommand{\Pstar}{P^{\ast}}
\newcommand{\steplist}[1]{\StrSubstitute{#1}{/}{\slash}}
\DeclareMathOperator*{\argmin}{arg\,min}

\renewcommand{\mid}{\mathchoice{\,\vert\,}{\,\vert\,}{\,\vert\,}{\,\vert\,}}

\makeatletter
\let\@stdexp\exp
\renewcommand{\exp}{\@stdexp\,}
\makeatother

\begin{document}

\title{FRESH: Information-Geometric Calibration\\of Patient-Level Models to Aggregate Evidence}

\author{F.~Fuller}
\author{D.~Bertolini}
\author{S.~Liang}
\author{J.~Christopher}
\author{A.~Smith}
\affil{Unlearn.AI, Inc., San Francisco, CA, USA}

\date{October 2026}

\maketitle

\begin{abstract}
This note introduces \textsc{FRESH} (Fusion of Recent Evidence and Subject Histories), a method for incorporating population-level summary results --- published clinical trials, registries, natural-history studies, and indirect comparisons that report absolute arm-level summaries --- into generative models trained on patient-level trial or real-world data.  The output is a single patient-level representation of both data types that matches specified aggregate statistics of a target population.  The recalibration is sequentially minimal in an information-geometric sense: the baseline marginal is the $I$-projection of the model's marginal, and the outcome kernel is the minimal change to the model's kernel given that calibrated baseline.  The calibrated synthetic sample can be analyzed directly or used to fine-tune the model (not explored here).

Applications include contextualizing single-arm trial results against recent standard of care, trial simulation for design and probability-of-technical-success estimation, and comparative-effectiveness analysis of on-market therapies.

Two case studies follow.  The first simulates a head-to-head comparison in metastatic pancreatic cancer of two first-line standards of care with overlapping indications, full-dose FOLFIRINOX and gemcitabine plus nab-paclitaxel, which have not been compared in a randomized trial in a PRODIGE-4-like population.  \textsc{FRESH} and a model-based standardization step split the published cross-trial overall-survival difference into a baseline component, explained under the model's prognostic structure, and a residual that reflects the regimens only absent unmeasured cross-trial confounding; its intervals include the null.  The second is a retrospective, outcome-blinded (held-out) forecast of the control arm of the BREAKWATER trial in \textit{BRAF}~V600E-mutant metastatic colorectal cancer, composed from \textsc{FRESH} calibrations to two historical trials and scored against the published readout.
\end{abstract}

\tableofcontents

\section{Introduction and motivation}
\label{sec:intro}

Many decisions in clinical science and epidemiology --- estimating probability of technical success for a clinical trial, assessing comparative effectiveness of two therapies, imputing a placebo effect onto natural history data --- rely on combining information about a clinical cohort from different kinds of studies.  Specifically we contrast patient-level sources that provide granular pictures of individual disease course (clinical trials, registries, or electronic health records) with aggregate sources such as published clinical trial results and their tables, figures, and listings (TFLs).

One strategy for combining aggregate with patient-level data sources is to bring each into a common format for a unified analysis.  If one wants to maintain the analytic flexibility of patient-level data, then a natural solution is to convert the aggregate data information into a simulated patient-level dataset that recapitulates those aggregate statistics.  The idea has precedents: simulated treatment comparison predicts outcomes for a published population from a fitted outcome model~\cite{caro2010stc}, parametric G-computation standardizes such a model over a covariate distribution simulated to match published aggregates~\cite{remiroazocar2022gcomp}, moment-restriction weighting reweights a sample to reproduce known population moments~\cite{hellerstein1999imposing}, and synthetic-data integration generates individual-level data from external summary-level information~\cite{gu2023synthetic}.  This is an under-determined inverse problem in that there are many such datasets, and it cannot be well specified without further constraints.

\textsc{FRESH} (Fusion of Recent Evidence and Subject Histories) provides a well-defined method for solving this problem while retaining the analytic flexibility of patient-level data.  The approach frames the problem of combining patient-level data with aggregate results as a model calibration procedure.  The patient-level source data is represented as a generative model that summarizes the dataset.  The aggregate information is then incorporated as a calibration that adapts the model to respect the aggregate constraints through a sequentially minimal perturbation of the original model: the baseline marginal is the minimal change to the model's marginal that meets the baseline constraints, and the outcome kernel is the minimal change to the model's kernel given that calibrated baseline.  Seeking the generator closest in Kullback--Leibler (KL) divergence under expectation constraints and then fine-tuning toward it has precedents in controlled text generation~\cite{khalifa2021distributional} and in the calibration of generative models~\cite{smith2025calibrating}; \textsc{FRESH} differs in pooling conditional outcome constraints over a fixed calibrated baseline and in needing only black-box sampling from the model.

Among its strengths, \textsc{FRESH} is highly flexible and imposes few requirements on the generative model and aggregate constraint inputs. The only modeling input \textsc{FRESH} requires is a generative model of the joint distribution of baseline covariates $X$ and outcomes $Y$ --- that is, a sampler for the baseline marginal $p(x)$ and a sampler for the outcome conditional $p(y\mid x)$, which together draw $(X,Y)$ from the joint distribution $p(x)\,p(y\mid x)$; the conditional density itself need not be evaluated. The form and provenance of that generative model are unconstrained: it could be a parametric survival model, a graphical model, or a neural simulator. The aggregate evidence source is generic; it is most often a published randomized controlled trial, but the point calibration depends only on (a) a description of the target population (eligibility criteria and a table of baseline characteristics) and (b) a list of outcome summary statistics under the regimen of interest. We will refer to ``the source'' below without committing to its provenance.

\textsc{FRESH} comprises two stages, both of which are information projections (\S\ref{sec:info-geom}) in the sense of Csisz\'ar~\cite{csiszar1975idivergence,csiszar1984sanov}.  Stage~1 is a version of entropy balancing (\S\ref{sec:eb})~\cite{hainmueller2012entropy} within the broader family of balancing and calibration weights~\cite{zubizarreta2015stable,wang2020minimal,chan2016calibration}; Stage~2 is its conditional analogue, a conditional maximum-entropy update of the outcome kernel with the baseline marginal held fixed~\cite{berger1996maxent}. \emph{Stage~1} constructs a baseline distribution $\Qstar(x)$ --- a surrogate for the inaccessible true source baseline marginal --- as a weighted empirical measure on draws from the generative model, by entropy balancing against the eligibility set and the source's reported baseline statistics.  Each baseline statistic can be enforced as either a hard constraint (exact matching) or a soft constraint (approximate matching under penalty), and the same choice is available for each outcome statistic. \emph{Stage~2} perturbs the conditional kernel to $\Kstar(y\mid x)$, the conditional surrogate, so that the induced outcome marginal under $\Qstar$ matches the source's outcome statistics. The two stages compose: Stage~1's weights are passed to Stage~2 unchanged.

The contributions of this note are as follows. First, we formulate both stages of \textsc{FRESH} as $I$-projections --- Stage~1 on the baseline marginal, Stage~2 on the outcome kernel with that marginal held fixed --- and state their exponential-tilt solutions, which follow from classical maximum-entropy duality~\cite{csiszar1975idivergence,berger1996maxent,dudik2007maxent}, in the form needed here. Second, we build a Markov chain Monte Carlo (MCMC) sampler for Stage~2 on the weighted Stage~1 ensemble. As in any independence Metropolis--Hastings sampler~\cite{tierney1994markov}, the reference density cancels from the acceptance ratio, and, as in maximum-entropy refinement of molecular simulations~\cite{cesari2018maxent}, the multipliers are learned on-line during sampling; what we add is the ensemble design over a weighted particle cohort. We discuss its requirements, advantages, and limitations. Third, we demonstrate the utility of \textsc{FRESH} in two case studies: a simulated head-to-head comparison of two first-line therapy regimens in metastatic pancreatic cancer, and a retrospective, outcome-blinded (held-out) forecast of a phase~3 control arm in \textit{BRAF}~V600E-mutant colorectal cancer, scored against the trial's published readout.

Note that these simulations are intended as methods demonstrations on real and clinically relevant questions, not as a fully credentialed indirect treatment comparison (ITC) of the two pancreatic regimens nor as a validated forecasting tool. A complete comparative effectiveness analysis would require a fuller exposition of the underlying generative model --- training corpus, architecture, goodness-of-fit --- which we leave to a separate publication.

This note is organized as follows. In \S\ref{sec:setup} we precisely define the problem \textsc{FRESH} solves. In \S\ref{sec:info-geom} and \S\ref{sec:two-iprojs} we describe \textsc{FRESH} using the mathematics of information projections, and in \S\ref{sec:tilt-form} show that the solutions to the projections introduce an exponential tilt factor to the generative distribution. In \S\ref{sec:empirical} we reformulate the problem and solution in terms of finite samples drawn from the generative model, MCMC sampling from $\Kstar$ without evaluating the model's density, and Lagrange multipliers. In \S\ref{sec:constraints} we provide concrete examples of constraints to meet typical needs, and summarize the combined algorithm in \S\ref{sec:algorithm}. The procedure is then exercised in \S\ref{sec:case-study} demonstrating a head-to-head simulation of two first-line therapy regimens in metastatic pancreatic cancer.  These therapies were established individually in separate pivotal randomized trials (a phase 2--3 and a phase 3 trial), and no randomized trial has compared full-dose FOLFIRINOX with gemcitabine plus nab-paclitaxel in a PRODIGE-4-like population.  The two randomized comparisons reported in 2025 used modified FOLFIRINOX in different populations, and both found longer median overall survival on gemcitabine plus nab-paclitaxel~\cite{ohba2025generate,knox2025pass01}; we return to this contrast in \S\ref{sec:case-study}.  Here \textsc{FRESH} is used to produce patient-level simulations of each arm, and then --- by a model-based standardization step that reweights the simulated MPACT arm to PRODIGE-4's baseline statistics --- to decompose the naive cross-trial difference between the two regimens into a baseline-adjustment component, explained under the model's prognostic structure, and a residual component that admits a treatment-effect reading only if there is no unmeasured cross-trial confounding.  The intervals for the residual include the null.  This serves to elucidate the comparative effectiveness of one therapy regimen over the other, a thorny problem given the mismatched baseline distributions of these two trials. \S\ref{sec:case-study-crc} presents a second case study of a different kind: a retrospective, outcome-blinded (held-out) forecast of the standard-of-care arm of the BREAKWATER trial in \textit{BRAF}~V600E-mutant colorectal cancer, built by composing \textsc{FRESH} calibrations to two older unselected trials and scored, out of sample, against BREAKWATER's published readout. No BREAKWATER outcome enters the pipeline, but the readout was public when the analysis was built, and \S\ref{sec:case-study-crc} lists the choices made with it known. The case study exercises multi-source composition, simultaneous calibration of overall and progression-free survival, and a nested uncertainty design that separates fold-model (training-split) variability from source-bootstrap and sampler variability. We conclude with a discussion of the method in \S\ref{sec:discussion}, including the place of \textsc{FRESH} within ITC across trials with differing but overlapping populations --- a well-developed area of clinical-trial methodology that includes matching-adjusted indirect comparison (MAIC;~\cite{signorovitch2010maic,phillippo2016tsd,phillippo2018maic}), simulated treatment comparison (STC;~\cite{caro2010stc}), and multilevel network meta-regression (ML-NMR;~\cite{phillippo2020mlnmr,phillippo2026mlnmr}), among others --- and its relation to data integration with external summary statistics, external control arms, and simulator calibration in other fields. These relationships are unpacked in \S\ref{sec:maic-comparison}.

\section{Setup and problem statement}
\label{sec:setup}

Throughout this note we follow the convention that capital letters $X,Y$ denote random variables and lowercase letters $x,y$ denote their realizations; densities are written with respect to a dominating measure (typically Lebesgue on $\R^d$ or counting measure on a finite outcome support), and we use the same symbol for a kernel and its density. Outcomes are taken to be scalar (e.g.\ overall survival time) for notational simplicity, but the construction extends to vector-valued $Y$ unchanged.

\subsection{Generative model}
A generative model of the joint distribution of $X\in\X$ and $Y\in\Y$ is given, factored as
\begin{equation}
p(y,x) \;=\; \Qref(x)\,\Kref(y\mid x),\qquad \Qref(x):=p(x),\quad \Kref(y\mid x):=p(y\mid x).
\label{eq:gen-factor}
\end{equation}
The symbols $\Qref$ and $\Kref$ name the marginal and the conditional of the same generative density $p$; the only purpose of the renaming is to avoid repeated overloading of the bare symbol $p(\cdot)$ as we transform the marginal in Stage~1 and the conditional in Stage~2. We assume two interfaces: a sampler from the marginal $\Qref$ and a sampler from the conditional $y\sim\Kref(\cdot\mid x)$. Density evaluation $y\mapsto\Kref(y\mid x)$ is not required; where an implementation tabulates $\Kref(\cdot\mid x)$ once per sampled $x$, as in the case studies, this is an implementation detail. No further structure on $\Qref$ or $\Kref$ is assumed; either may be discrete or continuous, parametric or not.

When \textsc{FRESH} outputs are used to contrast treatment regimens, the regimen is part of the baseline vector: we write $X=(A,W)$, with $A$ the regimen and $W$ the remaining baseline covariates, so that $\Kref(y\mid a,w)$ is regimen-conditional. Both case studies impose the regimen as a hard Stage~1 constraint. A contrast across values of $A$ computed from \textsc{FRESH} outputs has a causal reading only if $A$ is unconfounded given $W$ across the pooled sources, i.e.\ the potential outcomes $Y(a)$ are independent of $A$ given $W$; for regimens observed in different trials this requires that no unmeasured factor affecting the outcome differs between the trial populations (\S\ref{sec:case-study-notes}).

\subsection{Aggregate evidence source}
\label{sec:source}
The source represents the results of some experiment linking a baseline distribution of covariates to a distribution of outcomes.  Following the same factorization, we denote the true joint distribution associated with the source as $Q(x)\,K(y\mid x)$, which is not directly observed.
The aggregate evidence source supplies three pieces of information about this distribution.
\begin{enumerate}
\item An eligibility set $\X_{\textsc{ie}}\subseteq\X$ (inclusion / exclusion criteria) defining the population to which its statistics refer.
\item Baseline summary statistics
\begin{equation}
    \E_{X\sim Q}[\phi_\ell(X)] \;=\; b_\ell,\qquad \ell=1,\dots,L,
\end{equation}
represented by functions $\{\phi_\ell(x)\}$, typically means, quantiles, or proportions of clinico-genomic features. A subset $\mathcal{H}_X\subseteq\{1,\dots,L\}$ of these is to be matched exactly (\emph{hard} constraints, e.g.\ critical prognostic factors); the remainder $\mathcal{S}_X=\{1,\dots,L\}\setminus\mathcal{H}_X$ is to be matched approximately under a quadratic penalty with weights $\rho_\ell>0$, $\ell\in\mathcal{S}_X$ (\emph{soft} constraints).
\item Outcome summary statistics
\begin{equation}
\E_{X\sim Q,\,Y\sim K(\cdot\mid X)}\!\left[f_j(X,Y)\right] \;=\; c_j,\qquad j=1,\dots,J,
\label{eq:constraints}
\end{equation}
represented by functions $\{f_j(x,y)\}$, typically medians, landmark survival probabilities, or means of the outcome, for which $f_j(x,y)=f_j(y)$ depends on $y$ alone. A statistic reported within a subgroup $\X_j\subseteq\X$ enters as $f_j(x,y)=\ind{x\in\X_j}\,f_j(y)$, with the reported subgroup value rescaled by the subgroup's probability (\S\ref{sec:subgroups}), as in survival models calibrated to auxiliary subgroup survival information~\cite{huang2016auxiliary}. As for the baseline statistics, a subset $\mathcal{H}_Y\subseteq\{1,\dots,J\}$ is hard and the remainder $\mathcal{S}_Y=\{1,\dots,J\}\setminus\mathcal{H}_Y$ is soft, with penalty weights $\rho_j>0$, $j\in\mathcal{S}_Y$.
For time-to-event endpoints $Y$ is the latent (uncensored) event time: survival-curve targets such as Kaplan--Meier (KM) landmarks~\cite{kaplan1958nonparametric} estimate its distribution under independent censoring, and the support of $\Kref(\cdot\mid x)$ must cover every landmark.
\end{enumerate}

All summary statistics admitted by the framework, both baseline (item~2) and outcome (item~3), are required to be \emph{representable} in the sense above: each one must be expressible as $\E[\phi(\cdot)]$ for some function $\phi$ of the relevant variable. Means and proportions are already of this form; quantile functionals admit, for a target $p_j$-th quantile equal to $n_j$, the indicator
\begin{equation}
f_j(y) \;=\; \ind{y\le n_j},\qquad c_j \;=\; p_j.
\label{eq:hardquantile}
\end{equation}
Composite statistics enter through representable components: a reported median or inter-quartile range (IQR) enters as indicator constraints of the form~\eqref{eq:hardquantile} at its reported quantile values, and a variance enters as a second-moment constraint, which fixes the variance exactly only when the mean is also a hard constraint.  Statistics that are non-linear functionals (e.g.\ hazard ratios, KL divergences, Wasserstein distances) are \emph{not} directly admissible in this form. Two regularity conditions are needed for the existence and uniqueness arguments in the sequel (\S\ref{sec:tilt-form} and Appendix~\ref{app:proofs}). The targets must be feasible, with the hard targets in the relative interior of the set of moments attainable under the reference. And the statistics must have finite exponential moments jointly: the set of multipliers $\nu$ with $\E[\exp(\nu^{\top}\phi)]<\infty$ under the reference is open, which is automatic for bounded statistics such as indicators~\cite[Thm.~3.3]{csiszar1975idivergence}.

\subsection{Goal}
The true source pair $Q(x)\,K(y\mid x)$ defined in \S\ref{sec:source} is not directly observed; only its summary statistics $\{b_\ell\}$ and $\{c_j\}$ are available. Our goal is to construct a surrogate factored distribution $\Qstar(x)\,\Kstar(y\mid x)$ (or a weighted sample therefrom) that is consistent with those summaries and is otherwise minimally perturbed from the generative model $\Qref(x)\,\Kref(y\mid x)$. Specifically, $\Qstar$ and $\Kstar$ should satisfy
\begin{itemize}
\item the baseline marginal $\Qstar$ is supported in $\X_{\textsc{ie}}$ and minimizes the $I$-divergence from the generative marginal $\Qref$ plus a quadratic penalty on the soft baseline statistics, subject to exact matching of the hard ones (problem~\eqref{eq:stage1-iproj}), and
\item the conditional kernel $\Kstar(y\mid x)$ minimizes the $\Qstar$-weighted conditional $I$-divergence from $\Kref$ plus a quadratic penalty on the soft outcome statistics~\eqref{eq:constraints}, subject to exact matching of the hard ones under $\Qstar$ (problem~\eqref{eq:stage2-iproj}).
\end{itemize}
The two requirements are met sequentially in Stages~1 and~2 below. With all baseline constraints hard, Stage~1 recovers the source baseline, $\Qstar=Q$, when $\log(dQ/d\Qref)$ restricted to $\X_{\textsc{ie}}$ lies in the span of $1$ and the $\phi_\ell$, the condition under which entropy balancing is consistent~\cite{zhao2017entropy,han2019empirical} (\S\ref{sec:maic-comparison}). We use the asterisked notation $\Qstar, \Kstar$ throughout to distinguish \textsc{FRESH}'s constructed surrogates from the inaccessible true source distributions $Q, K$; expectations under $Q$ or $K$ are population-level quantities, of which the source publication reports finite-sample estimates, while expectations under $\Qstar$ or $\Kstar$ are the corresponding quantities computed from \textsc{FRESH}'s outputs. \textsc{FRESH} treats the reported baseline values $b_\ell$ as fixed; the case studies propagate the sampling error of the outcome targets $c_j$ by bootstrap (\S\ref{sec:case-study-error-quant} and \S\ref{sec:crc-error-quant}).

\section{Information projections}
\label{sec:info-geom}

The mathematical core of \textsc{FRESH} is a single object: the \emph{information projection} ($I$-projection) of a reference probability measure onto a moment-constraint set. For a reference probability measure $\mu_0$ on a measurable space $(\Omega,\mathcal{F})$, a list of measurable constraint statistics $\{h_j\}_{j=1}^J$ with target values $\{c_j\}$ partitioned into a \emph{hard} index set $\mathcal{H}$ (matched exactly) and a \emph{soft} index set $\mathcal{S}$ (matched approximately under a quadratic penalty with weights $\rho_j>0$), the $I$-projection is the minimizer over probability measures $\mu\ll\mu_0$
\begin{equation}
\mu^\ast \;=\; \argmin_{\mu\ll\mu_0}\;\KL(\mu\|\mu_0) \;+\; \tfrac{1}{2}\sum_{j\in\mathcal{S}}\rho_j\bigl(\E_\mu[h_j]-c_j\bigr)^2
\quad\text{s.t.}\quad \E_\mu[h_j]=c_j,\,j\in\mathcal{H}.
\label{eq:iproj-def}
\end{equation}
The all-hard regime ($\mathcal{S}=\varnothing$) is the $I$-projection of Csisz\'ar~\cite{csiszar1975idivergence}; Problem~\eqref{eq:iproj-def}, with its mix of hard and soft constraints, is an instance of the \emph{generalized maximum-entropy} program of Dud{\'\i}k, Phillips, and Schapire~\cite{dudik2007maxent}, which builds on the original proposal of Jaynes~\cite{jaynes1957information}. Relative-entropy calibration of simulation output with exact or least-squares constraints was used earlier to reweight Monte Carlo paths of asset-pricing models to benchmark market prices~\cite{avellaneda2001weighted}.
Under mild regularity (bounded $h_j$ and hard targets in the relative interior of the attainable moments; cf.\ Proposition~\ref{prop:cond-iproj}), this program is convex with a unique solution $\mu^\ast$ of exponential-tilt form $d\mu^\ast/d\mu_0(\omega)\propto\exp\bigl(\sum_j\lambda_j^\ast h_j(\omega)\bigr)$, with hard multipliers $(\lambda_j^\ast)_{j\in\mathcal{H}}$ pinned by the equality constraints $\E_{\mu^\ast}[h_j]=c_j$ and soft multipliers $(\lambda_j^\ast)_{j\in\mathcal{S}}$ satisfying the implicit Tikhonov relation
\begin{equation}
\lambda_j^\ast \;=\; -\rho_j\bigl(\E_{\mu^\ast}[h_j]-c_j\bigr),\qquad j\in\mathcal{S}.
\label{eq:soft-multiplier}
\end{equation}
These statements follow by specializing Theorem~2 of Dud{\'\i}k, Phillips, and Schapire~\cite{dudik2007maxent} (finite domains) and Theorem~1 of Nghiem and Mar{\'e}chal~\cite{nghiem2026maxent} (extension to bounded continuous domains, under its conditions (ii)--(iii)) --- which are stated for an arbitrary closed proper convex relaxation potential $U$ --- to our mixed potential $U(s) = \delta(s_{\mathcal H}\mid\{c_{\mathcal H}\}) + \tfrac12\sum_{j\in\mathcal S}\rho_j(s_j-c_j)^2$, where $\delta(\cdot\mid C)$ is the convex indicator of a set $C$ (zero on $C$, $+\infty$ off it). For unbounded supports such as survival times, the all-hard case is covered by Csisz\'ar~\cite{csiszar1975idivergence}, and the mixed case for Stage~2 by the argument of Appendix~\ref{app:proofs}. The Fenchel conjugate of $U$ yields the dual objective
\begin{equation}
\Lambda(\lambda)\;=\;\log Z(\lambda)\;-\;\sum_{j}\lambda_j c_j\;+\;\sum_{j\in\mathcal S}\frac{\lambda_j^2}{2\rho_j},
\qquad Z(\lambda)\;=\;\E_{\mu_0}\Bigl[\exp\bigl({\textstyle\sum_j}\lambda_j h_j\bigr)\Bigr],
\label{eq:dual-general}
\end{equation}
a convex function of $\lambda$ minimized at $\lambda^\ast$, with $\log Z$ the log-partition function. Its stationarity conditions $\partial\Lambda/\partial\lambda_j=0$ are the hard equalities for $j\in\mathcal H$ and the soft-multiplier relation~\eqref{eq:soft-multiplier} for $j\in\mathcal S$.

This exponential-tilt structure is the Kullback--Leibler member of the minimum-discrepancy family of moment-condition estimators~\cite{newey2004higher}: the exponential-tilting estimators of Kitamura and Stutzer~\cite{kitamura1997information} and of Imbens, Spady, and Johnson~\cite{imbens1998information} reweight the data by exactly such a tilt, as an information-theoretic alternative to the quadratic-form generalized method of moments (GMM)~\cite{hansen1982gmm}. Empirical likelihood~\cite{owen2001empirical} is the reverse-KL member of the same family, and its weights have the rational form $1/\{n(1+\lambda^\top m_i)\}$, with $m_i$ the moment function at observation $i$, rather than an exponential one. The exponential tilt also underlies exponential tilting in importance sampling, entropy balancing in causal inference~\cite{hainmueller2012entropy}, and the relative-entropy calibration of simulator output to moment information in macroeconomic forecasting~\cite{robertson2005relative} and asset pricing~\cite{avellaneda2001weighted}. We adopt the $I$-projection vocabulary because it covers both stages of \textsc{FRESH} uniformly within a single conceptual frame.

The Gibbs conditioning principle of Csisz\'ar~\cite{csiszar1984sanov} offers a complementary heuristic for the all-hard case: for a large unweighted i.i.d.\ sample from $\mu_0$, conditioned on its empirical measure satisfying the moment constraints, the conditional law of each draw approaches $\mu^\ast$ as the sample size grows. This motivates reading the exponential tilt as the limit of empirical conditioning; it does not cover the soft penalty, and the validity of the Metropolis--Hastings realization of Stage~2 in \S\ref{sec:mh} rests instead on Proposition~\ref{prop:cond-iproj} and Remark~\ref{rem:cancellation}.

\section{\textsc{FRESH} as two information projections}
\label{sec:two-iprojs}

\textsc{FRESH} produces its surrogate $\Qstar(x)\,\Kstar(y\mid x)$ for the inaccessible source pair $Q(x)\,K(y\mid x)$ by composing two $I$-projections, applied sequentially to the factored generative model $\Qref(x)\,\Kref(y\mid x)$: Stage~1 is an instance of~\eqref{eq:iproj-def}, and Stage~2 is an $I$-projection of the joint with its $X$-marginal pinned (\S\ref{sec:stage2}).

\subsection{Stage 1: baseline projection}
\label{sec:stage1}
Stage~1 specializes the $I$-projection of \S\ref{sec:info-geom} to the generative marginal $\Qref$ and the baseline-statistic constraint set
\begin{equation}
\mathcal{C}_Q\;=\; \bigl\{q : q(\X_{\textsc{ie}})=1,\ \E_q[\phi_\ell(X)] = b_\ell\ \forall \ell\in\mathcal{H}_X\bigr\},
\label{eq:CQ}
\end{equation}
producing the surrogate baseline distribution
\begin{equation}
\Qstar \;=\; \argmin_{q\in\mathcal{C}_Q}\,\KL(q\,\|\,\Qref) \;+\; \tfrac{1}{2}\sum_{\ell\in\mathcal{S}_X}\rho_\ell\bigl(\E_q[\phi_\ell(X)]-b_\ell\bigr)^2.
\label{eq:stage1-iproj}
\end{equation}
The eligibility filter $q(\X_{\textsc{ie}})=1$ is not a regular moment constraint: written as $\E_q[\ind{X\notin\X_{\textsc{ie}}}]=0$, it is a boundary target with no finite multiplier. It enters instead as a restriction of the reference. For $q$ supported on $\X_{\textsc{ie}}$, $\KL(q\,\|\,\Qref)$ differs from $\KL(q\,\|\,\Qref^{\textsc{ie}})$ only by the constant $-\log\Qref(\X_{\textsc{ie}})$, where $\Qref^{\textsc{ie}}$ is $\Qref$ restricted to $\X_{\textsc{ie}}$ and renormalized. Stage~1 is therefore the $I$-projection of $\Qref^{\textsc{ie}}$ onto the remaining moment constraints, with the multipliers solved on $\X_{\textsc{ie}}$, and it keeps the exponential-tilt form (see \S\ref{sec:restricted-iproj}).

\subsection{Stage 2: conditional kernel projection}
\label{sec:stage2}
With the Stage-1 cohort $\Qstar$ now fixed, Stage~2 projects the reference conditional kernel $\Kref$ onto the outcome-statistic constraint set, with fidelity measured by the $\Qstar$-weighted conditional Kullback--Leibler divergence~\cite{kullback1951information,cover2006elements}:
\begin{equation}
\KL\bigl(k\,\big\|\,\Kref;\,\Qstar\bigr)
\;=\; \int \KL\bigl(k(\cdot\mid x)\,\big\|\,\Kref(\cdot\mid x)\bigr)\,d\Qstar(x).
\label{eq:condKL}
\end{equation}
The constraint set on candidate kernels (with $\{c_j\}$ the source values from~\eqref{eq:constraints}) is
\begin{equation}
\mathcal{C}_K \;=\; \Bigl\{k : \E_{X\sim \Qstar,\,Y\sim k(\cdot\mid X)}[f_j(X,Y)] = c_j\ \forall j\in\mathcal{H}_Y\Bigr\},
\label{eq:CK}
\end{equation}
and the Stage-2 output is
\begin{equation}
\Kstar \;=\; \argmin_{k\in\mathcal{C}_K}\,\KL\bigl(k\,\big\|\,\Kref;\,\Qstar\bigr) \;+\; \tfrac{1}{2}\sum_{j\in\mathcal{S}_Y}\rho_j\bigl(\E_{X\sim\Qstar,\,Y\sim k(\cdot\mid X)}[f_j(X,Y)]-c_j\bigr)^2.
\label{eq:stage2-iproj}
\end{equation}
The statistics $f_j(x,y)$ are those of~\eqref{eq:constraints}; a subgroup-scoped statistic $\ind{x\in\X_j}\,f_j(y)$ takes as target the reported subgroup value times $\Qstar(\X_j)$ (\S\ref{sec:subgroups}), and an unrestricted statistic has $f_j(x,y)=f_j(y)$.

The two outputs jointly define the \textsc{FRESH} surrogate $\Qstar\otimes\Kstar$ for the unobserved source pair $Q\otimes K$. Stage~2 is not itself of the form~\eqref{eq:iproj-def}, but it is an $I$-projection of the joint with a pinned $X$-marginal. By the chain rule of relative entropy, $\KL(\Qstar\otimes k\,\|\,\Qstar\otimes\Kref)=\KL(k\,\|\,\Kref;\Qstar)$, so~\eqref{eq:stage2-iproj} is the $I$-projection of $\Qstar\otimes\Kref$ onto $\{P: P_X=\Qstar,\ \E_P[f_j(X,Y)]=c_j\ \forall j\in\mathcal{H}_Y\}$, with the soft block penalized as before. The marginal constraint $P_X=\Qstar$ carries a function-valued multiplier $\eta(x)$; at the optimum $\eta(x)=-\log Z(\lambda^\ast\mid x)$, which yields the per-$x$ tilt of $\Kref$ in Proposition~\ref{prop:cond-iproj} (\S\ref{sec:tilt-form}).

\subsection{Two stages versus one}
\label{sec:two-vs-one}
A natural alternative formulation, which we call \emph{joint} \textsc{FRESH}, replaces the two-stage composition with a single $I$-projection of the joint distribution $\Qref\otimes\Kref$ onto the lifted joint constraint set $\mathcal{C}_J=\{P: P_X(\X_{\textsc{ie}})=1,\ \E_P[\phi_\ell(X)]=b_\ell\ \forall \ell\in\mathcal{H}_X,\ \E_P[f_j(X,Y)]=c_j\ \forall j\in\mathcal{H}_Y\}$, with both soft blocks penalized. Both procedures yield calibrated joint distributions that meet the hard constraints and penalize departures from the soft ones, and both are built from the $I$-projection construction of \S\ref{sec:info-geom}; they generally produce different objects and satisfy different guarantees. We adopt the two-stage composition as the \textsc{FRESH} default. The detailed comparison and trade-offs --- composability across analyses and conditions for coincidence of the two procedures --- are deferred to the discussion (\S\ref{sec:composition}).

\section{Form of the solutions: exponential tilting}
\label{sec:tilt-form}

Both Stage-1 and Stage-2 outputs inherit the exponential-tilt structure of the general $I$-projection.

\subsection{Stage~1}
The Stage-1 solution is the generative marginal $\Qref$ multiplied pointwise by an eligibility-filtered exponential tilt:
\begin{equation}
\Qstar(x) \;=\; w^\ast(x)\,\Qref(x),\qquad w^\ast(x)\;=\;\frac{\ind{x\in\X_{\textsc{ie}}}\,\exp\bigl(\sum_{\ell=1}^L \nu_\ell^\ast \phi_\ell(x)\bigr)}{Z_X(\nu^\ast)},
\label{eq:stage1-form}
\end{equation}
where $Z_X(\nu)=\E_{\Qref}[\ind{X\in\X_{\textsc{ie}}}\exp(\sum_\ell\nu_\ell\phi_\ell(X))]$ is the eligibility-restricted normalizing constant. By the generalized maxent result~\cite{dudik2007maxent,nghiem2026maxent}, the multipliers split by block: hard multipliers $(\nu_\ell^\ast)_{\ell\in\mathcal{H}_X}$ are determined by the moment-matching constraints $\E_{\Qstar}[\phi_\ell]=b_\ell$, and soft multipliers $(\nu_\ell^\ast)_{\ell\in\mathcal{S}_X}$ satisfy the implicit Tikhonov relation $\nu_\ell^\ast=-\rho_\ell(\E_{\Qstar}[\phi_\ell]-b_\ell)$.

\subsection{Stage~2}
The Stage-2 problem~\eqref{eq:stage2-iproj} is a \emph{conditional} $I$-projection: the $X$-marginal $\Qstar$ is held fixed and only the conditional kernel $k(\cdot\mid x)$ is varied, so the induced joint $\Qstar\otimes k$ ranges over the slice of joints whose $X$-marginal coincides with $\Qstar$. It is a conditional maximum-entropy problem in the sense of Berger, Della Pietra, and Della Pietra~\cite{berger1996maxent}, who fix the empirical $X$-marginal and obtain the same per-$x$ tilt as solution; their setting has a uniform reference, a finite outcome space and hard constraints only. For a general $\Qstar$ the slice constraint is an infinite family of moment equalities, so the result does not reduce directly to the finite-moment generalized maxent of~\cite{dudik2007maxent,nghiem2026maxent} applied to the joint; it requires a function-valued Lagrange multiplier for the marginal (\S\ref{sec:stage2}), alongside the finite-dimensional multipliers for the outcome moments. The result is recorded as Proposition~\ref{prop:cond-iproj}; the proof, in Appendix~\ref{app:proofs}, proceeds via the chain rule of relative entropy, Fenchel duality on the soft block, and a direct verification that the per-$x$ tilt is optimal. Entropy minimization under convex constraints of this kind, including infinitely many marginal constraints, is treated in generality by L\'eonard~\cite[Thm.~5.2 and \S6.2]{leonard2008minimization}.

\begin{proposition}\label{prop:cond-iproj}
Let $\mathcal{M}_{\mathcal{H}}$ be the set of moment vectors $\bigl(\E_{\Qstar\otimes k}[f_j(X,Y)]\bigr)_{j\in\mathcal{H}_Y}$ attained by kernels $k$ with $\KL(k\,\|\,\Kref;\Qstar)<\infty$, where $\E_{\Qstar\otimes k}$ denotes expectation under $X\sim\Qstar$, $Y\sim k(\cdot\mid X)$, and assume
\begin{itemize}
\item[(H1)] each $f_j$ is bounded and measurable;
\item[(H2)] the hard targets lie in the relative interior, $(c_j)_{j\in\mathcal{H}_Y}\in\operatorname{ri}\mathcal{M}_{\mathcal{H}}$.
\end{itemize}
Then the Stage-2 problem~\eqref{eq:stage2-iproj} has a solution $\Kstar$, unique for $\Qstar$-almost every $x$, of the per-$x$ exponential-tilt form
\begin{equation}
\Kstar(y\mid x) \;=\; \frac{\Kref(y\mid x)\,\exp\bigl(\sum_{j=1}^J \lambda_j^\ast f_j(x,y)\bigr)}{Z(\lambda^\ast\mid x)},
\label{eq:expt-kernel}
\end{equation}
with conditional partition function
\begin{equation}
Z(\lambda\mid x)\;=\;\int \Kref(y\mid x)\,\exp\bigl({\textstyle\sum_j} \lambda_j f_j(x,y)\bigr)\,dy.
\label{eq:partition-fn}
\end{equation}
The multipliers $\lambda^\ast$ minimize the convex dual
\begin{equation}
\Lambda(\lambda)\;=\;\int\log Z(\lambda\mid x)\,d\Qstar(x)\;-\;\sum_{j}\lambda_j c_j\;+\;\sum_{j\in\mathcal{S}_Y}\frac{\lambda_j^2}{2\rho_j},
\label{eq:dual-stage2}
\end{equation}
whose stationarity conditions split by block as in Stage~1: the hard multipliers $(\lambda_j^\ast)_{j\in\mathcal{H}_Y}$ are determined by the moment equalities $\E_{\Qstar\otimes\Kstar}[f_j(X,Y)]=c_j$, and soft multipliers $(\lambda_j^\ast)_{j\in\mathcal{S}_Y}$ satisfy the Tikhonov relation $\lambda_j^\ast=-\rho_j\bigl(\E_{\Qstar\otimes\Kstar}[f_j(X,Y)]-c_j\bigr)$. The soft multipliers are unique, and the hard multipliers are unique under
\begin{itemize}
\item[(H3)] no nontrivial linear combination of $(f_j)_{j\in\mathcal{H}_Y}$ is $\Qstar\otimes\Kref$-almost surely a function of $X$ alone.
\end{itemize}
\end{proposition}
\noindent (Proof in Appendix~\ref{app:proofs}.) A hard target outside the attainable set or on its boundary has no finite multiplier; the two TRIBE progression-free survival (PFS) landmarks beyond the simulator's support (Table~\ref{tab:crc-outcome-targets}) are an example, and they were dropped.

\subsection{Interpretation of the multipliers}
\label{sec:multipliers}
The exponent in each tilt is linear in the multipliers, so each multiplier $\lambda_j^\ast$ measures the strength of the perturbation that constraint $j$ imposes on the reference. Precisely, let $V(c)$ denote the optimal value of~\eqref{eq:stage2-iproj}, the minimal penalized conditional KL, as a function of the targets $c=(c_j)$. By strong duality (Appendix~\ref{app:proofs}) $V(c)=-\min_\lambda\Lambda(\lambda)$, with $\Lambda$ as in~\eqref{eq:dual-stage2}, so the envelope theorem gives $\lambda_j^\ast=\partial V/\partial c_j$ wherever $\lambda^\ast$ is unique (and analogously $\nu_\ell^\ast=\partial V_X/\partial b_\ell$ for the Stage-1 optimal value $V_X(b)$): the multiplier is the marginal cost (in fidelity) of moving the constraint target by a unit. Large $|\lambda_j^\ast|$ flags constraints that pull the reference far, with two cautions. First, the multipliers are determined jointly: $\lambda_j^\ast=0$ means that constraint $j$ is non-binding given the other targets, not that the reference satisfies it in isolation, and a target the reference already meets can carry a nonzero multiplier when other constraints are active. Second, $\lambda_j^\ast$ has the inverse units of $c_j$, and with nested landmark indicators each $\lambda_j^\ast$ is a jump in the log-tilt $y\mapsto\sum_j\lambda_j^\ast f_j(x,y)$ rather than a separate effect. We therefore compare magnitudes only after standardization (e.g.\ $\lambda_j^\ast\times\mathrm{SE}(c_j)$, the approximate change in $V$ from a one-standard-error move of the target) or through the implied log-tilt curve. With these provisos, the readability of the multipliers holds in both stages and is one of several reasons we adopt KL as the fidelity functional (see \S\ref{sec:fidelity} for a fuller discussion).

\begin{remark}[Maximum-entropy interpretation]
Equivalently, $\Kstar$ minimizes the $\Qstar$-averaged divergence~\eqref{eq:condKL}; that is, it is the maximum-entropy kernel relative to $\Kref$ under $\Qstar$-weighting, subject to the aggregate constraints. Because those constraints couple all $x$ through $\Qstar$, this reading applies to the kernel as a whole, not to each $\Kstar(\cdot\mid x)$ separately. Writing $H(k(\cdot\mid x))=-\int k(y\mid x)\log k(y\mid x)\,dy$ for the entropy of $k(\cdot\mid x)$, the decomposition
\[
\KL(k\,\|\,\Kref;\Qstar)\;=\;\int\Bigl[-H\bigl(k(\cdot\mid x)\bigr)+\E_{Y\sim k(\cdot\mid x)}\bigl[-\log\Kref(Y\mid x)\bigr]\Bigr]\,d\Qstar(x)
\]
exhibits two distinct fidelities to the reference: the entropy term resists collapse to a degenerate distribution, while the cross-entropy term resists assigning mass where the reference says outcomes are unlikely.
\end{remark}

\section{Empirical realizations}
\label{sec:empirical}

The exponential-tilt expressions of \S\ref{sec:tilt-form} are population-level objects. In practice, both stages are realized on finite samples drawn from the generative model. The multipliers $(\nu^\ast,\lambda^\ast)$ are computed by either convex optimization (Stage~1) or by an MCMC scheme that samples from the projected target without ever evaluating $\Kref$ densely (Stage~2). The choice between the two routes is dictated by the interfaces the underlying generative model exposes; the design considerations are discussed in \S\ref{sec:realization}.

\subsection{Convex optimization: empirical entropy balancing for Stage~1}
\label{sec:eb}
Draw $M$ i.i.d.\ baselines $\widetilde X_1,\dots,\widetilde X_M\sim\Qref$, drop those failing the eligibility filter, and relabel survivors as $x_1,\dots,x_N$. The empirical Stage-1 problem is to find weights $w\in\R_{\ge 0}^N$ solving
\begin{equation}
\min_{w\ge 0,\,\mathbf{1}^\top w=N}\quad \tfrac{1}{N}\sum_{i=1}^N w_i\log w_i \;+\; \tfrac{1}{2}\sum_{\ell\in\mathcal{S}_X}\rho_\ell\Bigl(\tfrac{1}{N}\sum_{i=1}^N w_i\phi_\ell(x_i)-b_\ell\Bigr)^2
\label{eq:stage1-empirical}
\end{equation}
subject to $\tfrac{1}{N}\sum_{i=1}^N w_i\phi_\ell(x_i)=b_\ell$ for all $\ell\in\mathcal{H}_X$, where the hard/soft partition $(\mathcal{H}_X,\mathcal{S}_X)$ and the weights $\rho_\ell$ are those of \S\ref{sec:source}. The normalization $\sum_i w_i=N$ (mean weight one) is the convention used downstream by Stage~2.

Problem~\eqref{eq:stage1-empirical} is the empirical generalized maxent of Dud{\'\i}k, Phillips, and Schapire~\cite{dudik2007maxent} on the finite domain $\{x_1,\dots,x_N\}$, with the empirical reference $\widehat Q_{0,N}$. Writing $\widehat Q_{0,N}=\tfrac{1}{N}\sum_i\delta_{x_i}$ and $q_w=\tfrac{1}{N}\sum_i w_i\delta_{x_i}$, the entropy term is $\tfrac{1}{N}\sum_i w_i\log w_i=\KL(q_w\|\widehat Q_{0,N})$, so it carries the same scale as the penalty, as in the population problem. Assuming the hard-block target lies in the relative interior of the convex hull, $b_{\mathcal H_X}\in\operatorname{ri}\operatorname{conv}\{\phi_{\mathcal H_X}(x_i)\}_{i=1}^N$ --- which guarantees feasibility, finite Lagrange multipliers, and nondegenerate weights --- the maxent duality theorem~\cite[Thm.~2]{dudik2007maxent}, specialized to our mixed hard+$\ell_2^2$ potential, gives the solution in Gibbs form
\begin{equation}
w_i \;\propto\; \exp\Bigl(\sum_{\ell=1}^L \nu_\ell^\ast \phi_\ell(x_i)\Bigr),
\label{eq:stage1-weights}
\end{equation}
normalized to $\sum_i w_i=N$, where the hard multipliers $(\nu_\ell^\ast)_{\ell\in\mathcal{H}_X}$ are determined by the equality constraints and the soft multipliers $(\nu_\ell^\ast)_{\ell\in\mathcal{S}_X}$ satisfy $\nu_\ell^\ast=-\rho_\ell(\tfrac{1}{N}\sum_i w_i\phi_\ell(x_i)-b_\ell)$. The weights are unique by strict convexity of the relative entropy~\cite[\S4]{dudik2007maxent}. The multipliers are unique if and only if $\{1\}\cup\{\phi_\ell\}_{\ell\in\mathcal H_X}$ is affinely independent on $\{x_1,\dots,x_N\}$, i.e.\ no nontrivial linear combination of the hard statistics is constant on the cohort; in practice this means dropping one level of each complete categorical block (e.g.\ one of male/female, one Eastern Cooperative Oncology Group (ECOG) performance-status level).

The all-hard case recovers Hainmueller's entropy balancing~\cite{hainmueller2012entropy} --- equivalently, MAIC weights~\cite{phillippo2020equivalence} --- which is also the raking case of calibration weighting in survey sampling~\cite{deville1992calibration}; the exact-balance calibration estimators of Chan, Yam, and Zhang~\cite{chan2016calibration} generalize it to other divergences. The soft block is the Tikhonov analog of the box-constraint relaxations in the approximately-balancing weights literature~\cite{zubizarreta2015stable,wang2020minimal}. As $N\to\infty$ the empirical solution converges to the population $I$-projection of \S\ref{sec:two-iprojs}: the multipliers $\nu^\ast$ minimize a sample-average convex dual, so they are M-estimators of the population multipliers and are consistent under standard conditions (cf.~\cite{wang2020minimal}; see~\cite{glasserman2005weighted} for the large-sample theory of the analogous weighted Monte Carlo estimators). The output of Stage~1 is the cohort $\{(x_i,w_i)\}_{i=1}^N$, defining the empirical measure $\Qstar(dx)=\tfrac{1}{N}\sum_i w_i\,\delta_{x_i}(dx)$. We call each $x_i$, carried with its weight $w_i$, a \emph{particle}; the outcomes the Stage-2 chain generates for a particle are \emph{draws}, and we reserve \emph{sample} for a collection of particles or draws.

\subsection{Ensemble Metropolis--Hastings sampler for Stage~2}
\label{sec:mh}
In the regime the method is designed for, the reference kernel $\Kref$ is a generative simulator that is cheap to sample but expensive, or impossible, to evaluate densely: direct evaluation of $Z(\lambda\mid x)$ would then require a one-dimensional integral against $\Kref(\cdot\mid x)$ for every particle at every update of $\lambda$. We bypass $Z(\lambda\mid x)$ entirely with a Metropolis--Hastings (MH) sampler~\cite{metropolis1953equation,hastings1970mcmc}. In both case studies $\Kref(\cdot\mid x_i)$ is in fact tabulated once per particle on a discrete time grid, so there $Z$ is a finite sum and exact tilted draws would also be available; we keep MH because it returns unweighted outcome draws and applies unchanged to simulators that expose only a sampler.

Fix a cohort $\{(x_i,w_i)\}_{i=1}^N$ with $\sum_i w_i = N$ representing $\Qstar$. Initialize $y_i\sim\Kref(\cdot\mid x_i)$ independently; the state of the chain is the ensemble $\mathbf{y}=(y_1,\dots,y_N)$. The target on $\Y^N$ (with the $x_i$ fixed) is the product of per-particle exponentially tilted kernels:
\begin{equation}
\pi(\mathbf{y}\mid\boldsymbol\lambda)
\;\propto\;\prod_{i=1}^N \Kref(y_i\mid x_i)\,\exp\Bigl(\sum_{j=1}^J \lambda_j f_j(x_i,y_i)\Bigr).
\label{eq:target}
\end{equation}
Each particle is independently tilted, so the marginal of $y_i$ under~\eqref{eq:target} is exactly the per-$x_i$ exponential tilt $\Kstar(\cdot\mid x_i)\propto\Kref(\cdot\mid x_i)\exp(\sum_j\lambda_j f_j(x_i,\cdot))$ of~\eqref{eq:expt-kernel}. The cohort weights $w_i$ do not enter~\eqref{eq:target}; they enter only in the empirical constraint estimator $\widehat g_j = \tfrac{1}{N}\sum_i w_i f_j(x_i,y_i)$, which is unbiased for $\E_{X\sim\Qstar,Y\sim\Kstar(\cdot\mid X)}[f_j(X,Y)]$ at $\boldsymbol\lambda^\ast$ and drives the Robbins--Monro update of \S\ref{sec:rm} to the empirical Stage-2 KKT point.

At each iteration, draw a per-particle Bernoulli$\bigl(\min(1,\alpha w_i)\bigr)$ inclusion mask $\mathcal{I}\subseteq\{1,\dots,N\}$, propose $\tilde y_i\sim\Kref(\cdot\mid x_i)$ for $i\in\mathcal{I}$, and leave $y_i$ unchanged for $i\notin\mathcal{I}$. Because both the target~\eqref{eq:target} and the proposal factorize over particles, each proposed particle is accepted or rejected on its own, and the $\Kref$-density factors cancel exactly between target and proposal in each acceptance ratio:

\begin{remark}[Reference-density cancellation]\label{rem:cancellation}
Let $f_1,\dots,f_J$ be any constraint functions for which~\eqref{eq:target} is normalizable, including the sigmoid approximations $f_j^\varepsilon$ of \S\ref{sec:sigmoid}. For each $i\in\mathcal{I}$, the MH update of $y_i$ with proposal $\tilde y_i\sim\Kref(\cdot\mid x_i)$ and target~\eqref{eq:target} accepts $\tilde y_i$, independently across $i\in\mathcal{I}$, with probability
\begin{equation}
\min\!\bigl(1,\,\exp(\Delta_i)\bigr),\qquad
\Delta_i \;=\; \sum_{j=1}^J\lambda_j\bigl[f_j(x_i,\tilde y_i)-f_j(x_i,y_i)\bigr].
\label{eq:accept}
\end{equation}
Accepting or rejecting all of $\mathcal{I}$ jointly, with probability $\min(1,\exp(\sum_{i\in\mathcal{I}}\Delta_i))$, is also a valid MH update for the same target, but its acceptance rate typically falls as $|\mathcal{I}|$ grows.
\end{remark}
\noindent (Proof in Appendix~\ref{app:proofs}. This is the standard independence-sampler identity~\cite{tierney1994markov}, applied particle by particle.)

This is the engineering payoff of using $\Kref$ as the proposal: the chain never evaluates $\Kref(y\mid x)$, only samples from it. The simulator is treated as a black-box generator, and the only quantities evaluated each acceptance step are the constraint functions $f_j$ at the proposed draw. What is specific to \textsc{FRESH} is the ensemble design: one chain over the whole weighted cohort, with weight-modulated proposal rates and a single set of multipliers driven by the weighted ensemble average $\widehat g_j$.

\subsection{Adaptive Lagrange multipliers via Robbins--Monro}
\label{sec:rm}
The multipliers $\boldsymbol\lambda^\ast$ are not known in advance. We adapt them on-line by stochastic approximation~\cite{robbins1951stochastic}, as in on-the-fly maximum-entropy refinement of molecular simulations against experimental averages~\cite{cesari2018maxent}. After each MH step, update
\begin{equation}
\lambda_j^{(t+1)}\;=\;\lambda_j^{(t)} \;+\; \gamma_t\,\Bigl(c_j - \widehat{g}_j^{(t)} - \ind{j\in\mathcal{S}_Y}\,\tfrac{\lambda_j^{(t)}}{\rho_j}\Bigr),\qquad
\widehat{g}_j^{(t)}\;=\;\frac{1}{N}\sum_{i=1}^N w_i\,f_j(x_i,y_i^{(t)}),
\label{eq:rm}
\end{equation}
with step size $\gamma_t$ satisfying the Robbins--Monro conditions $\sum_t\gamma_t=\infty$ and $\sum_t\gamma_t^2<\infty$. (In the implementation $\gamma_t=\gamma_0/(t_0+t)^\beta$ with $\beta\in(\tfrac{1}{2},1]$, a standard polynomial-decay schedule that satisfies both conditions in the $t\to\infty$ sense; both case studies use $\beta=0.6$, $t_0=1$, and $\gamma_0=1$ (Appendices~\ref{app:mh-hyperparameters} and~\ref{app:crc-mh-hyperparameters}). Each per-step change in $\lambda_j$ is also clipped to $[-\delta_{\max},\delta_{\max}]$ to suppress transient instability when the bracketed residual is far from zero. Numerical diagnostics for the pancreatic runs --- acceptance rates, multiplier summaries, across-partition Gelman--Rubin $\widehat R$, achieved landmark deviations, and the stopping iteration --- are reported in Appendix~\ref{app:mh-diagnostics}. The colorectal chains ran a fixed iteration budget without acceptance or multiplier logging, so their convergence evidence is limited to the achieved-versus-target deviations of Table~\ref{tab:crc-mh-diagnostics}.)

The update is a noisy gradient-descent step on the Stage-2 dual $\Lambda(\lambda)$ of~\eqref{eq:dual-stage2}. Write
\[
K_\lambda(y\mid x)\;=\;\Kref(y\mid x)\exp\Bigl(\sum_j\lambda_j f_j(x,y)\Bigr)\Big/Z(\lambda\mid x)
\]
for the reference kernel tilted at $\lambda$, so that $\Kstar=K_{\lambda^\ast}$. The gradient of the dual is
\[
\frac{\partial\Lambda}{\partial\lambda_j}\;=\;\E_{X\sim\Qstar,\,Y\sim K_\lambda(\cdot\mid X)}[f_j(X,Y)]-c_j+\ind{j\in\mathcal{S}_Y}\,\frac{\lambda_j}{\rho_j},
\]
of which $\widehat g_j^{(t)}-c_j+\ind{j\in\mathcal{S}_Y}\,\lambda_j^{(t)}/\rho_j$ is the running estimate. Setting it to zero recovers the stationarity conditions of Proposition~\ref{prop:cond-iproj}: $\E_{\Qstar\otimes\Kstar}[f_j(X,Y)]=c_j$ for $j\in\mathcal{H}_Y$ and $\lambda_j^\ast=-\rho_j(\E_{\Qstar\otimes\Kstar}[f_j(X,Y)]-c_j)$ for $j\in\mathcal{S}_Y$. The extra ridge term $-\gamma_t\lambda_j^{(t)}/\rho_j$ in the soft branch is a shrinkage that pulls soft multipliers toward zero in proportion to the looseness of the penalty weight. Written out, the soft branch multiplies $\lambda_j^{(t)}$ by $(1-\gamma_t/\rho_j)$, which is stable only while $\gamma_t<2\rho_j$; since $\gamma_t\le\gamma_0$, we require $\gamma_0<2\min_{j\in\mathcal{S}_Y}\rho_j$. Neither case study uses a soft Stage-2 constraint. The combined chain --- sampling $\mathbf{y}$ given $\boldsymbol\lambda$ and updating $\boldsymbol\lambda$ given $\mathbf{y}$ --- is an instance of adaptive Markov chain Monte Carlo~\cite{andrieu2006ergodicity}.

\paragraph{Adaptation, variants, and limitations.} The multipliers are not frozen after burn-in: $\boldsymbol\lambda$ keeps adapting, with decaying step $\gamma_t$, while draws are retained, so the retained draws come from a chain whose target is still slowly moving. The colorectal case study runs the sampler as described above: a single chain over the cohort, with per-particle acceptance. The pancreatic case study uses a partitioned variant: the cohort is split into $n_{\mathrm{part}}=400$ disjoint partitions (Appendix~\ref{app:mh-hyperparameters}), each running its own copy of the chain, with its own multipliers $\boldsymbol\lambda^{(p)}$, its own estimator $\widehat g_j^{(p)}$ computed on the partition's particles, and its own accept step. Each partition is thus calibrated to the targets separately, there is no single $\boldsymbol\lambda$, and a reported multiplier is a summary (mean and spread) across partitions. Finally, the independence proposal does not adapt to the tilt, so its efficiency falls as $\Kstar$ moves away from $\Kref$: for $f_j\in[0,1]$ (indicators and their sigmoid approximations) each $|\Delta_i|\le\sum_j|\lambda_j|$, so the per-particle acceptance probability is at least $\exp(-\sum_j|\lambda_j|)$, a bound that is informative only while the multipliers stay moderate.

\section{Constraint extensions}
\label{sec:constraints}

Several extensions of the basic moment-constraint setup are routinely needed in clinical practice. We record them here as variations on the same $I$-projection theme; each preserves the exponential-tilt structure of \S\ref{sec:tilt-form}.

\subsection{Sub-support restriction}
\label{sec:restricted-iproj}
When the projection must be supported on a measurable subset $\Omega_0\subseteq\Omega$ --- in \textsc{FRESH} this is the eligibility set $\X_{\textsc{ie}}\subseteq\X$ --- the reference $\mu_0$ is replaced by its restriction to $\Omega_0$, renormalized (for Stage~1, the reference $\Qref^{\textsc{ie}}$ of \S\ref{sec:stage1}), and the $I$-projection is taken against this restricted reference. The solution still has the exponential-tilt form, now as a tilt of $\mu_0$ restricted to $\Omega_0$, with the multipliers re-solved on $\Omega_0$. (Multiplying the unrestricted projection density by $\ind{\omega\in\Omega_0}$ and renormalizing would in general violate the moment constraints.) Eligibility thus enters as a restriction of the reference rather than as one more moment constraint; this is the formal device by which inclusion / exclusion criteria enter Stage~1.

\subsection{Subgroup-scoped constraints}
\label{sec:subgroups}
Many clinical reports give summary statistics restricted to a subgroup $\X_j\subseteq\X$ (e.g.\ ``median overall survival (OS) in the $\ge 65$ stratum''); auxiliary subgroup survival probabilities of this kind have also been used to sharpen Cox-model estimation~\cite{huang2016auxiliary}. A subgroup target $c_j$ for $\E[f_j(Y)\mid X\in\X_j]$ is accommodated by replacing $f_j(y)$ with the statistic $f_j(x,y)=\ind{x\in\X_j}\,f_j(y)$ of~\eqref{eq:constraints} and $c_j$ with $c_j\,\Qstar(\X_j)$; the two forms are equivalent because Stage~2 holds $\Qstar$ fixed. Correspondingly, the empirical estimate $\widehat g_j$ in~\eqref{eq:rm} is replaced with the subgroup-normalized average
\begin{equation*}
\widehat g_j \;=\; \frac{\sum_{i:\,x_i\in\X_j}w_i\,f_j(y_i)}{\sum_{i:\,x_i\in\X_j}w_i}.
\end{equation*}
For a soft subgroup constraint, a weight $\rho_j$ on the residual of $\ind{x\in\X_j}\,f_j(y)$ is a weight $\rho_j\,\Qstar(\X_j)^2$ on the subgroup-normalized residual, and the fixed point of the normalized update corresponds to the penalty weight $\rho_j/\Qstar(\X_j)$ on the unnormalized statistic, so the stated $\rho_j$ is implicitly rescaled. Multiple subgroup constraints are handled by carrying one $\lambda_j$ and one indicator $\ind{x\in\X_j}$ per constraint. The target~\eqref{eq:target} is unchanged, since there $f_j(x_i,y_i)=\ind{x_i\in\X_j}\,f_j(y_i)$.

\subsection{Smooth approximations for quantile constraints}
\label{sec:sigmoid}
A quantile constraint at threshold $n_j$ and target percentile $p_j$ takes the indicator form $f_j(y)=\ind{y\le n_j}$ with $c_j=p_j$. Because $d\Kstar/d\Kref$ depends on $y$ only through $\sum_j\lambda_j f_j(x,y)$, a step in the indicator produces a step in $\Kstar$: the strict indicator gives a $\Kstar$ that is discontinuous in $y$ at $n_j$ whenever $\lambda_j^\ast\ne0$. We replace the indicator with the smooth sigmoid approximation
\begin{equation}
f_j^\varepsilon(y)\;=\;\sigma\!\bigl(\tfrac{n_j-y}{\varepsilon}\bigr),\qquad
\sigma(u)=\frac{1}{1+e^{-u}}.
\label{eq:sigmoid-surrogate}
\end{equation}
As $\varepsilon\downarrow0$, $f_j^\varepsilon(y)\to\ind{y<n_j}+\tfrac12\ind{y=n_j}$ pointwise (since $\sigma(0)=\tfrac12$), and by dominated convergence the soft-CDF functional $\E_{P}[f_j^\varepsilon(Y)]$ converges to $\Pr_P[Y\le n_j]$ whenever $P$ has no atom at $n_j$. For positive $\varepsilon$ the approximation is smooth in $y$ and in $\varepsilon$, and the functional is smooth in any parameter on which $\Kref$ depends smoothly.

For typical clinical reporting (e.g.\ overall survival in months), $\varepsilon$ of $10$--$30$~days (the values used in the two case studies) produces an approximation operationally indistinguishable from the exact quantile constraint while removing the discontinuity in $\Kstar$. With $f_j^\varepsilon$ in place of $f_j$ in~\eqref{eq:expt-kernel} (for a subgroup-scoped constraint, $f_j^\varepsilon(x,y)=\ind{x\in\X_j}\,f_j^\varepsilon(y)$), the resulting kernel
\begin{equation}
\Kstar(y\mid x)\;=\;\frac{\Kref(y\mid x)\,\exp\bigl(\sum_j\lambda_j^\ast f_j^\varepsilon(x,y)\bigr)}{Z(\lambda^\ast\mid x)}
\label{eq:smoothed-tilt}
\end{equation}
is positive wherever $\Kref$ is, inherits the smoothness of $\Kref$ in $y$ and $x$ (the tilt factor itself is smooth, so a gridded $\Kref$ keeps its bin edges), and converges to the hard tilt as $\varepsilon\to 0$. The MH sampler of \S\ref{sec:mh} runs against this smoothed target without modification.

\section{Combined algorithm}
\label{sec:algorithm}

The two stages compose end-to-end into Algorithm~\ref{alg:fresh}. Stage~1 produces the cohort and weights; Stage~2 is the MH chain over outcome draws conditioned on the Stage~1 cohort.

\begin{algorithm}[H]
\caption{\textsc{FRESH}}
\label{alg:fresh}
\begin{algorithmic}[1]
\Require Generative model with samplers $X\sim\Qref$ and $Y\sim\Kref(\cdot\mid x)$ (no density evaluation is needed); eligibility set $\X_{\textsc{ie}}$; baseline statistics $\{(\phi_\ell,b_\ell)\}_{\ell=1}^L$ with hard/soft partition $(\mathcal{H}_X,\mathcal{S}_X)$ and soft penalties $\{\rho_\ell\}_{\ell\in\mathcal{S}_X}$; outcome statistics $\{(f_j,c_j)\}_{j=1}^J$ with hard/soft partition $(\mathcal{H}_Y,\mathcal{S}_Y)$ and soft penalties $\{\rho_j\}_{j\in\mathcal{S}_Y}$; cohort size $M$; sigmoid scale $\varepsilon$; proposal rate $\alpha\in(0,1]$; step-size schedule $\{\gamma_t\}$; clip $\delta_{\max}$; iterations $T$; burn-in $T_b$; spacing $s$; retained draws per particle $n_{\mathrm{keep}}$.
\Statex \textbf{Stage 1 --- Cohort construction.}
\State Draw $\widetilde X_1,\dots,\widetilde X_M\sim\Qref$ and discard those with $\widetilde X_m\notin\X_{\textsc{ie}}$; relabel survivors as $x_1,\dots,x_N$.
\State Solve~\eqref{eq:stage1-empirical}, e.g.\ via its dual in $\nu$, for $\nu^\ast\in\R^L$ and set $w_i\propto\exp(\sum_\ell\nu_\ell^\ast\phi_\ell(x_i))$, normalized to $\sum_i w_i=N$.
\Statex \textbf{Stage 2 --- Kernel calibration via MH.}
\State Initialize $y_i\sim\Kref(\cdot\mid x_i)$ for $i=1,\dots,N$ and $\lambda_j\gets 0$ for $j=1,\dots,J$.
\For{$t = 1,\dots,T$}
  \State Draw inclusion mask $\xi_i\sim\mathrm{Bernoulli}(\min(1,\alpha w_i))$; set $\mathcal{I}=\{i:\xi_i=1\}$.
  \For{$i\in\mathcal{I}$}
    \State Propose $\tilde y_i\sim\Kref(\cdot\mid x_i)$ and set $\Delta_i\gets\sum_j\lambda_j\bigl[f_j^\varepsilon(x_i,\tilde y_i)-f_j^\varepsilon(x_i,y_i)\bigr]$.
    \State Draw $u_i\sim\mathrm{Uniform}(0,1)$. \textbf{If} $u_i<\min(1,\exp(\Delta_i))$, accept: $y_i\gets\tilde y_i$.
  \EndFor
  \State For $j=1,\dots,J$: $\widehat g_j\gets\tfrac{1}{N}\sum_i w_i f_j^\varepsilon(x_i,y_i)$; $\lambda_j\gets\lambda_j+\mathrm{clip}\!\bigl(\gamma_t(c_j-\widehat g_j-\ind{j\in\mathcal{S}_Y}\,\lambda_j/\rho_j),\pm\delta_{\max}\bigr)$.
  \If{$t>T_b$ \textbf{and} $(t-T_b)\bmod s=0$}
    \State Append the current $\mathbf{y}$ to the per-particle chain (ring buffer of depth $n_{\mathrm{keep}}$).
  \EndIf
\EndFor
\State \textbf{Return:} cohort weights $\{w_i\}$ from Stage~1, and per-particle outcome chains $\{y_i^{(t_1)},\dots,y_i^{(t_{n_{\mathrm{keep}}})}\}_{i=1}^N$ from Stage~2.
\end{algorithmic}
\end{algorithm}

In the partitioned variant of \S\ref{sec:rm}, the Stage~2 loop runs separately within each partition, with that partition's own $\boldsymbol\lambda^{(p)}$ and $\widehat g_j^{(p)}$.

The Stage~2 output is a chain of $n_{\mathrm{keep}}$ Monte Carlo draws per particle; pooled across $i$ these form a sample from the calibrated joint $\tfrac{1}{N}\sum_i w_i\,\delta_{x_i}\otimes\Kstar(\cdot\mid x_i)$, the calibrated predictive distribution. The per-particle chain spread describes outcome variability under $\Kstar(\cdot\mid x_i)$ (there is none to read when $n_{\mathrm{keep}}=1$, as in the colorectal case study); it is not a posterior, since $\boldsymbol\lambda$ is a point estimate. Uncertainty about the calibration itself is quantified separately, by the bootstrap and fold-model designs of the case studies (\S\ref{sec:case-study-error-quant} and \S\ref{sec:crc-error-quant}).

\subsection{Complexity}
Each iteration proposes in expectation at most $\alpha N$ particles (the weights have mean one), so a run of $T$ iterations needs only about $T\alpha N$ draws from $\Kref$ and $O(T\alpha NJ)$ evaluations of the smooth functionals, plus, in the case studies, a one-time tabulation of $\Kref(\cdot\mid x_i)$ on the time grid for every particle. Both implementations in fact draw a fresh proposal for all $N$ particles at every iteration and discard the non-proposed ones, which keeps array shapes static at a per-iteration cost of $O(N)$ draws and $O(NJ)$ functional evaluations; both vectorize trivially over the particle index. The Robbins--Monro update recomputes $\widehat g_j$ in $O(NJ)$; updating it incrementally from the accepted particles reduces this to $O(\alpha NJ)$. The case studies use $\alpha=10^{-3}$ (pancreatic) and $\alpha=\crcMhAlpha$ (colorectal). The dominant cost is the $\Kref$ draws, which can be batched across $i$ wherever the underlying simulator supports it.

\section{Case study: head-to-head in metastatic pancreatic adenocarcinoma}
\label{sec:case-study}

\subsection{Clinical question}
Two pivotal randomized trials, a phase 2--3 and a phase 3 trial, defined the modern first-line standard of care in metastatic pancreatic adenocarcinoma (mPAAD) and reported markedly different overall survival (OS) on their experimental arms:
\begin{itemize}
\item \textbf{PRODIGE-4 / ACCORD 11} (NCT00112658)~\cite{conroy2011folfirinox}: FOLFIRINOX (\textsc{FX}) vs.\ gemcitabine; experimental-arm median OS \textbf{11.1 months}.
\item \textbf{MPACT} (NCT00844649)~\cite{vonhoff2013mpact}: gemcitabine + nab-paclitaxel (\textsc{GN}) vs.\ gemcitabine; experimental-arm median OS \textbf{8.5 months}.
\end{itemize}
We use the abbreviations \textsc{FX} and \textsc{GN} for the two experimental regimens throughout the rest of this section.
The $\sim$\ebNativeGapDays-day median OS gap suggests a substantial treatment-effect advantage of \textsc{FX}, but the trials enrolled meaningfully different populations: PRODIGE-4 restricted enrollment to patients aged 18--75 with ECOG performance status 0--1, while MPACT enrolled patients with a Karnofsky performance status (KPS) as low as 70 (ECOG~2 under the mapping below) and geographically across more sites with somewhat older participants. The clinical decision relevant to a sponsor or a payer is the within-population effect: \emph{had the same patients received both regimens, how much of the published $\sim$\ebNativeGapDays-day gap would remain?} Direct cross-trial comparison of medians cannot answer this because demographic and treatment-effect contributions are confounded.

Formally, the target estimand of this case study is the average treatment effect on the treated (ATT) of \textsc{FX} versus \textsc{GN}, with \textsc{FX} as the index treatment and \textsc{GN} as the comparator, evaluated on the PRODIGE-4 population: the PRODIGE-4 covariate distribution is the inferential target throughout. On the eligibility criteria evaluable here (Appendix~\ref{app:ie-specs}), PRODIGE-4's eligibility set (ECOG~0--1, ages 18--75) is contained in MPACT's. Taking the broader MPACT population as the target population would instead require extrapolating \textsc{FX} performance onto patients excluded from PRODIGE-4 by protocol --- an extrapolation the data do not support --- so the direction of transport (MPACT~$\to$~PRODIGE-4 baseline, not the reverse) is determined by the estimand, not by analyst convenience.

The \textsc{FX}-versus-\textsc{GN} comparison is a salient clinical question in mPAAD: the two regimens are the dominant first-line cytotoxic options with overlapping indications, and no randomized trial has compared full-dose FOLFIRINOX with \textsc{GN} in a PRODIGE-4-like population. Two randomized trials reported in 2025 compared \emph{modified} FOLFIRINOX with \textsc{GN}: the Japanese phase 2--3 GENERATE/JCOG1611 trial~\cite{ohba2025generate} ($527$ patients aged 20--75 with ECOG~0--1 and metastatic or recurrent disease; irinotecan reduced to $150$~mg/m$^2$ from the $180$~mg/m$^2$ of PRODIGE-4) and the North American phase~2 PASS-01 trial~\cite{knox2025pass01} ($160$ patients with previously untreated metastatic disease, germline \textit{BRCA1/2} and \textit{PALB2} carriers excluded; primary endpoint progression-free survival). Both reported longer median OS on \textsc{GN}, the opposite direction to the \textsc{FX}-favoring estimate below. They differ from the comparison studied here in dose (modified rather than full-dose FOLFIRINOX), population (a Japanese cohort including recurrent disease; a North American cohort with germline carriers excluded) and era (PASS-01, for example, randomized from 2020 to 2024, against 2005--2012 for PRODIGE-4 and MPACT), so they do not test it directly, and treatment selection between the full-dose regimens still rests largely on indirect evidence. Two prior comparisons help frame the problem. Gresham et al.~\cite{gresham2014nma} performed a Bayesian network meta-analysis (NMA) anchored on the gemcitabine common comparator. Boyne et al.~\cite{boyne2023mpanc} performed a target-trial emulation (TTE)~\cite{hernan2016target} on a population-level real-world cohort from Alberta. Neither was designed for head-to-head trial \emph{simulation}. The TTE answers a comparative-effectiveness question on a real-world cohort whose weighted median OS was 8.3~months on \textsc{FX} and 5.1~months on \textsc{GN} --- a regime substantially worse than either pivotal trial reported --- and is silent on how the regimens would have performed had they been administered in a trial-protocol setting. The NMA, conversely, pools the log hazard ratios of each experimental arm against the shared gemcitabine arm, which makes it largely robust to purely prognostic imbalance between the trials; its validity rests instead on balance in effect modifiers, and the evident differences in baseline covariates between PRODIGE-4 and MPACT (age, performance status, geography) may include modifiers of each regimen's effect relative to gemcitabine.  \textsc{FRESH} sits between these two approaches: it produces patient-level trial-population surrogates rather than emulating a real-world cohort, and it then carries out the cross-trial within-population contrast on those surrogates by a model-based standardization --- a second pass of Stage~1 balancing that transports the MPACT cohort's baseline distribution onto PRODIGE-4's.

This case study uses \textsc{FRESH} to decompose the gap. The procedure produces, from a single calibrated patient-level surrogate of the MPACT experimental arm, a counterfactual ``\textsc{GN}-into-PRODIGE-4-baseline'' arm whose KM curve quantifies the part of the gap explained, under the model's prognostic structure, by baseline differences between the trial populations; the residual difference against the calibrated (model-predicted) PRODIGE-4 arm is the remaining, treatment-effect contribution, which bears a causal reading only under the assumptions stated in \S\ref{sec:case-study-notes}.

\subsection{Generative model}
\textsc{FRESH} treats the underlying generative model as a black box. For this case study the model is a generative transformer pretrained on more than $300{,}000$ oncogene panels from tumor biopsies with a subset linked to clinical features, treatment regimens, and outcomes; we use it through the two interfaces required by Algorithm~\ref{alg:fresh} --- a sampler from the baseline marginal $\Qref$ and a conditional sampler $Y\sim\Kref(\cdot\mid x)$. As an implementation detail, the conditional sampler is realized by tabulating $\Kref(\cdot\mid x_i)$ once per particle on the OS time grid and drawing proposals from that table. Nothing about the architecture, training procedure, or pretraining corpus appears anywhere in the calibration; the sections below would read identically were the underlying model replaced by any other simulator with the same two interfaces.

\subsection{Procedure}
The case study begins by drawing a large sample of synthetic baselines from the generative model.  It then runs Algorithm~\ref{alg:fresh} twice --- once each for MPACT and PRODIGE-4 --- to produce per-trial synthetic cohorts whose baseline and outcome marginals match the desired constraints.  Then it estimates the within-population \textsc{FX}-vs-\textsc{GN} treatment effect by the model-based standardization of \S\ref{sec:case-study-eb}, which transports the MPACT (\textsc{GN}) synthetic cohort's covariate distribution onto PRODIGE-4's published baseline statistics.

\subsubsection{Stage 1 (per trial): baseline entropy balancing}
\label{sec:case-study-stage1}
For each trial, the candidate baselines drawn from $\Qref$ are filtered to the trial's eligibility set $\X_{\textsc{ie}}$ (Appendix~\ref{app:ie-specs}), which also requires a non-null metastatic site. The Stage~1 constraints listed in Table~\ref{tab:mpanc-baseline-targets} are all hard and match each trial's own Table~1: age (the median, and for MPACT the quantile constraint $\Pr(\text{age}\le 65)=0.589$; the published minimum and maximum enter as support bounds, not moment equalities), sex, ECOG performance status, race (MPACT only), number of metastatic sites (proportions for MPACT, the median for PRODIGE-4), indication, and treatment regimen. The single soft term in this case study is a missingness penalty: the underlying generative model $\Qref$ is capable of generating baselines with missing values for some baseline features, and the implementation adds to the objective the linear penalty $\rho_{\text{miss}}\sum_i w_i\, m(x_i)$, where $m(x_i)$ counts the missing baseline features of particle $i$ and $\rho_{\text{miss}}=10^{-3}$. Unlike a soft moment penalty, this term has no target to match: it pushes the weight of incomplete particles toward zero wherever the hard constraints can be met without them, a near-exclusion of incomplete particles. A hard constraint $\E_q[\phi_{\text{miss}}]=0$ on the indicator $\phi_{\text{miss}}(x)=\ind{x\text{ has any missing baseline feature}}$ would be equivalent to discarding incomplete particles before balancing; the linear penalty differs from that only in letting incomplete particles keep weight where the complete ones alone cannot meet the hard constraints. Output: per-trial entropy-balanced cohort weights $\{w_i\}$ via~\eqref{eq:stage1-empirical}.

Note that MPACT reports performance status on the Karnofsky scale while PRODIGE-4 reports it on the ECOG scale, so the constraint on performance status requires a mapping between the two. We adopt the rule-of-thumb mapping between these imperfectly correlated scales of Ma et al.~\cite{ma2010interconversion}, namely ECOG~0~$=$~KPS~100, ECOG~1~$=$~KPS~80--90, and ECOG~2~$=$~KPS~60--70, and read MPACT's reported KPS distribution into ECOG buckets via this lookup before entropy balancing. Prasad et al.~\cite{prasad2018kps} agree for ECOG~0--1 but assign KPS~60 to ECOG~3; the difference affects only the two MPACT \textsc{GN} patients with a baseline KPS of 60~\cite{vonhoff2013mpact}.

\subsubsection{Stage 2 (per trial): MH calibration}
\label{sec:case-study-stage2}
For each trial, Algorithm~\ref{alg:fresh}'s Stage~2 runs the MH chain against the trial's published OS landmarks (medians plus 6/12/18/24/30-month landmark survival probabilities: MPACT's 6- to 24-month rates are printed in its Table~2~\cite{vonhoff2013mpact} and its 30-month value is read from the published Kaplan--Meier~\cite{kaplan1958nonparametric} curve; PRODIGE-4's values are taken from the text and published curve of Conroy et al.~\cite{conroy2011folfirinox}). All Stage~2 outcome constraints in this case study are hard equalities ($\mathcal{S}_Y=\varnothing$); the soft-multiplier branch of the Robbins--Monro update~\eqref{eq:rm} therefore reduces to the standard hard form. The published landmarks are themselves finite-sample estimates. We nonetheless impose them exactly and represent their sampling uncertainty separately, through the bootstrap of reconstructed pseudo individual patient data (pseudo-IPD) of \S\ref{sec:case-study-error-quant}, which re-solves Stage~2 under perturbed targets: each calibrated ensemble is then the $I$-projection onto a single target vector, and uncertainty about the targets is carried by the replicate design rather than by penalty weights. The output is a per-particle chain of OS draws targeting the $I$-projection $\Kstar$ of $\Kref$ onto the trial's outcome-constraint set under the trial-specific $\Qstar$. Numerical hyperparameter settings used in the runs are listed in Appendix~\ref{app:mh-hyperparameters}; convergence diagnostics (acceptance rates, multi-chain Gelman--Rubin $\widehat R$, achieved-vs-target landmark deviations, stopping evidence) are in Appendix~\ref{app:mh-diagnostics}.

\subsubsection{Cross-trial comparison: model-based standardization of MPACT onto PRODIGE-4's baseline}
\label{sec:case-study-eb}
The within-population treatment-effect contrast is computed by model-based standardization: a second application of Stage~1 balancing to the MPACT synthetic cohort. Because the outcomes carried by the rebalanced MPACT particles are draws from the calibrated outcome kernel rather than observed outcomes, the adjustment for baseline differences rests on the covariate dependence of the model's outcome kernel, as in parametric G-computation~\cite{remiroazocar2022gcomp}, rather than on observed outcomes in the PRODIGE-4 population. Each MPACT particle enters this step carrying its trial-specific Stage~1 weight $w_i$ as the empirical reference; the moment targets are now PRODIGE-4's published baseline statistics (Table~\ref{tab:mpanc-baseline-targets}) rather than MPACT's. Solving~\eqref{eq:stage1-empirical} with these substituted targets yields rebalanced weights $w'_i$ on the MPACT cohort whose induced baseline distribution matches PRODIGE-4's. The cross-trial outcome contrast is then a weighted Kaplan--Meier comparison of (i) the calibrated (model-predicted) PRODIGE-4 (\textsc{FX}) arm under its trial-specific Stage~1 weights $w_i$ and (ii) the MPACT (\textsc{GN}) MH-sampled arm under the rebalanced weights $w'_i$; medians, landmark survival, and the difference in restricted mean survival time ($\Delta$RMST)~\cite{royston2013rmst} are read off the two weighted KM curves directly (\S\ref{sec:case-study-results}).

\subsubsection{Error quantification: pseudo-IPD bootstrap of Stage~2 targets}
\label{sec:case-study-error-quant}
The $95\%$ confidence intervals (CIs) reported in Table~\ref{tab:mpanc-eb-results} propagate the underlying \emph{trial-population uncertainty} --- uncertainty about the published Kaplan--Meier curves on which the Stage~2 outcome targets are based --- through the cross-trial contrast. We (i) reconstruct pseudo-IPD from each published KM curve with the algorithm of Guyot et al.~\cite{guyot2012reconstruction}, (ii) bootstrap that pseudo-IPD~\cite{efron1981censored} to obtain $B$ perturbed Stage~2 outcome targets per arm, and (iii) fan out warm-started MH chains over those targets to produce $B$ perturbed calibrated outcome ensembles per arm. The downstream cross-trial step (here the model-based standardization followed by weighted Cox~\cite{cox1972regression} / KM / RMST analyses) is then iterated over $B_{\textsc{gn}}\times B_{\textsc{fx}}$ replicate pairs and aggregated to a pointwise percentile envelope.

\paragraph{Step 1. Pseudo-IPD reconstruction.} For each trial arm we apply the Guyot et al.~\cite{guyot2012reconstruction} algorithm to the digitized KM curve $\widehat S(t)$, the published at-risk table $\{(t_k,n_k)\}$, and the published total-event count, producing a synthetic IPD of size $n_0$ (one row per pseudo-patient with columns $(\text{time},\,\text{event})$). Within each at-risk interval, events are assigned from the KM jumps and censoring times are placed uniformly, with the number censored iterated until the estimated number at risk matches the next published $n_k$; in the final interval the number censored is adjusted so that the total number of events matches the published total. The reconstructed IPDs round-trip to within $\ebIpdRoundTrip$ of the published landmark survivals on both arms (PRODIGE-4 \textsc{FX}: $126/126$ events; MPACT \textsc{GN}: $333/333$ events).

\paragraph{Step 2. Bootstrap and target perturbation.} Subject-level resampling of the pseudo-IPD with deterministic per-replicate seeds yields $B$ replicate cohorts per arm. For each replicate we re-fit the Kaplan--Meier estimator and recompute every percentile-typed entry in the trial's metric-target configuration --- the same landmark survival probabilities and median-OS quantiles listed in Table~\ref{tab:mpanc-outcome-targets}. Each perturbed target spec is a draw from the sampling distribution of the published outcome statistics under nonparametric uncertainty about the underlying survival curve.

\paragraph{Step 3. Fan-out warm-started MH.} For each replicate we run Algorithm~\ref{alg:fresh}'s Stage~2 against the perturbed targets, warm-started from the final draws of the reference (unperturbed) MH chain. The Stage~1 cohort weights $\{w_i\}$ are reused unchanged across replicates by design --- only the metric-target specification and the chain random number generator vary --- so the reference and replicate ensembles share particles and the chains converge in much shorter time than a cold start. The output is $B$ perturbed calibrated outcome ensembles per arm, sharing the row-aligned particle base.

\paragraph{Step 4. Cross-trial contrast under propagated uncertainty.} The downstream cross-trial step is iterated over all $B_{\textsc{gn}}\times B_{\textsc{fx}}$ replicate pairs (Cartesian product). The reported $95\%$ CIs in Table~\ref{tab:mpanc-eb-results} are the pointwise $(q_{0.025},\,q_{0.975})$ envelope of the per-pair statistic across the resulting collection. The headline run uses $B_{\textsc{gn}}=B_{\textsc{fx}}=100$ replicate ensembles per arm, so that the Cartesian product yields $10{,}000$ replicate pairs.

The resulting CIs reflect the variability of the trial-population KM curves themselves and not just the within-cohort residual variation. This is materially wider than a single-realization CI when the underlying KM curves are estimated from modest event counts, as is the case for both arms here: the propagated hazard ratio (HR) confidence interval (Table~\ref{tab:mpanc-eb-results}) crosses unity, while a single-realization weighted Cox fit on the reference cohorts alone would have produced a much narrower interval excluding it.

\subsection{Calibration constraints supplied to each trial}
\label{sec:case-study-constraints}
The numerical content of Stage~1 (entropy balancing) is the list of trial-specific baseline-statistic targets in Table~\ref{tab:mpanc-baseline-targets}; the numerical content of Stage~2 (MH sampling) is the list of OS landmark-survival targets in Table~\ref{tab:mpanc-outcome-targets}.

\begin{table}[H]
\centering
\small
\setlength{\tabcolsep}{6pt}
\renewcommand{\arraystretch}{1.1}
\begin{tabular}{@{}l c c@{}}
\toprule
Statistic & MPACT & PRODIGE-4 \\
\midrule
\multicolumn{3}{l}{\textit{Age (years)}} \\
\quad median                          & 62           & 61           \\
\quad min, max                        & 27, 86       & 25, 76       \\
\quad $\Pr(\text{age}\le 65)$          & 0.589        & ---          \\
\multicolumn{3}{l}{\textit{Sex (proportions)}} \\
\quad Male                            & 0.569        & 0.620        \\
\quad Female                          & 0.431        & 0.380        \\
\multicolumn{3}{l}{\textit{ECOG performance status (proportions, remapped, renormalized)}} \\
\quad 0                               & 0.161        & 0.376        \\
\quad 1                               & 0.764        & 0.624        \\
\quad 2                               & 0.075        & ---          \\
\multicolumn{3}{l}{\textit{Race (proportions)}} \\
\quad White                           & 0.877        & ---          \\
\multicolumn{3}{l}{\textit{Number of metastatic locations}} \\
\quad proportion with 1               & 0.077        & ---          \\
\quad proportion with 2               & 0.469        & ---          \\
\quad proportion with 3               & 0.316        & ---          \\
\quad median                          & ---          & 2            \\
\multicolumn{3}{l}{\textit{Indication (proportions)}} \\
\quad pancreatic adenocarcinoma       & 1.00         & 1.00         \\
\textit{Regimen}                      & \textsc{GN}  & \textsc{FX}  \\
\bottomrule
\end{tabular}
\caption{Stage~1 baseline-statistic constraints supplied to entropy balancing for each trial's experimental arm. All tabulated entries are hard constraints; the age minimum and maximum enter as support bounds, not moment equalities. PRODIGE-4's published age maximum of 76 exceeds the protocol's upper age limit of 75 (Table~\ref{tab:mpanc-ie-specs}), so it cannot bind after the eligibility filter. Each Stage-1 objective also carries the linear missingness penalty described in \S\ref{sec:case-study-stage1} (not shown). Dashes indicate statistics not reported (or not constrained) in the corresponding trial's published Table~1.}
\label{tab:mpanc-baseline-targets}
\end{table}

\begin{table}[H]
\centering
\small
\setlength{\tabcolsep}{6pt}
\renewcommand{\arraystretch}{1.15}
\begin{tabular}{@{}l c c c@{}}
\toprule
Landmark & Time (d) & MPACT $S(t)$ & PRODIGE-4 $S(t)$ \\
\midrule
6 months                       & 183 & 0.67 & 0.76 \\
Median OS (MPACT)              & 259   & 0.50 & ---  \\
Median OS (PRODIGE-4)          & 337.9 & ---  & 0.50 \\
12 months                      & 365   & 0.35 & --- \\
18 months                      & 548   & 0.16 & 0.19 \\
24 months                      & 731   & 0.09 & 0.10 \\
30 months                      & 913   & 0.05$^{\dag}$ & 0.06 \\
\bottomrule
\end{tabular}
\caption{Stage~2 outcome calibration targets supplied to the MH chain for each trial. Each row specifies a time $t$ and target landmark survival probability $S(t)$, equivalent to a quantile constraint at percentile $p=100\,(1-S(t))$ on overall survival. The two trial-specific median-OS rows are supplied as explicit $50$th-percentile constraints at the trials' published medians (8.5 mo for MPACT \textsc{GN}, 11.1 mo for PRODIGE-4 \textsc{FX}). MPACT's 6- to 24-month rates are printed in its Table~2~\cite{vonhoff2013mpact}; $^{\dag}$read from the published Kaplan--Meier curve. All entries are hard equality constraints ($\mathcal{S}_Y=\varnothing$).}
\label{tab:mpanc-outcome-targets}
\end{table}

\subsection{Results}
\label{sec:case-study-results}

\paragraph{Calibration recovery.} Figure~\ref{fig:mpanc-km-full} overlays five curves: the published Kaplan--Meier curves of the two experimental arms (MPACT \textsc{GN}, PRODIGE-4 \textsc{FX}), the corresponding model-predicted per-trial arms after Stage~1 + Stage~2 calibration to each trial's own targets, and the rebalanced MPACT$\to$PRODIGE-4 arm. The model-predicted per-trial curves track their paper-reported counterparts closely, confirming that calibration hits the trial-specified landmarks. The rebalanced curve sits between the MPACT and PRODIGE-4 paper-reported arms over most of the first year and a half, but crosses the PRODIGE-4 curves at about 17 months and runs slightly above them thereafter; it captures the effect of baseline alignment of the source cohort onto the target's covariate distribution.

\paragraph{Baseline gap between cohorts.} Figure~\ref{fig:mpanc-baseline-dists} shows the entropy-balanced baseline distributions (ECOG, sex, race, age, metastatic-site count) in the two MH-sampled cohorts. The model-based standardization (\S\ref{sec:case-study-eb}) matches the rebalanced MPACT cohort only on the moments PRODIGE-4 reports (Table~\ref{tab:mpanc-baseline-targets}); differences in features PRODIGE-4 does not report, such as race, are not targeted by that step.

\paragraph{Cross-trial within-population contrast.} The primary head-to-head contrast --- \textsc{FX} vs.\ \textsc{GN} in the PRODIGE-4 covariate distribution --- is computed by the model-based standardization of the MPACT synthetic cohort onto PRODIGE-4 baseline targets (\S\ref{sec:case-study-eb}); 95\% CIs are the propagated pseudo-IPD bootstrap envelope of \S\ref{sec:case-study-error-quant}. Headline results are in Table~\ref{tab:mpanc-eb-results}.
\begin{table}[H]
\centering
\small
\setlength{\tabcolsep}{6pt}
\renewcommand{\arraystretch}{1.15}
\begin{tabular}{@{}l c c@{}}
\toprule
Quantity & Estimate & 95\% CI \\
\midrule
Median OS, MPACT \textsc{GN}-rebalanced                     & \ebMedOsRwBootDays~d (\ebMedOsRwBootMos~mo) & $[\ebMedOsRwDaysLo,\ \ebMedOsRwDaysHi]$~d \\
$S(12\text{ mo})$, MPACT \textsc{GN}-rebalanced             & $\ebStwlvBoot$   & $[\ebStwlvLo,\ \ebStwlvHi]$ \\
$S(24\text{ mo})$, MPACT \textsc{GN}-rebalanced             & $\ebStwfrBoot$   & $[\ebStwfrLo,\ \ebStwfrHi]$ \\
$\Delta$RMST(1y), \textsc{FX} vs.\ \textsc{GN}-rebalanced   & $+\ebDrOne$~d    & $[\ebDrOneLo,\ +\ebDrOneHi]$~d \\
$\Delta$RMST(2y), \textsc{FX} vs.\ \textsc{GN}-rebalanced   & $+\ebDrTwo$~d    & $[\ebDrTwoLo,\ +\ebDrTwoHi]$~d \\
HR (\textsc{FX} vs.\ \textsc{GN}-rebalanced)                & $\ebHr$          & $[\ebHrLo,\ \ebHrHi]$ \\
\bottomrule
\end{tabular}
\caption{Cross-trial within-population contrast on the MH-sampled cohorts after the model-based standardization of the MPACT cohort onto PRODIGE-4's published baseline statistics (Table~\ref{tab:mpanc-baseline-targets}). All rows report the median across replicate pairs as the point estimate and the $(q_{0.025},\,q_{0.975})$ envelope across $B_{\textsc{gn}}{\times}B_{\textsc{fx}}=100\times 100=10{,}000$ replicate pairs as the CI; per-arm landmark and median-OS rows ($S(12\text{ mo})$, $S(24\text{ mo})$, median OS on the rebalanced MPACT curve) report the source-arm statistic across pairs and cross-arm rows ($\Delta$RMST, HR) report the cross-trial contrast. All rows propagate trial-population uncertainty via the pseudo-IPD bootstrap pipeline of \S\ref{sec:case-study-error-quant}.}
\label{tab:mpanc-eb-results}
\end{table}

On the unperturbed reference run, the model-based standardization shifts the MPACT-derived weighted-median OS from $\ebMpactPredDaysDec$~days (calibrated MPACT arm) to $\ebMedOsRwDaysDec$~days, a $+\ebShiftDaysDec$-day displacement, out of a synthetic-frame median-OS gap of $\ebGapDaysDec$~days between the calibrated MPACT and PRODIGE-4 arms. This shift is the part of the gap explained under the model's prognostic structure by the baseline differences between the trial populations; its share of the gap depends on the scale on which the gap is measured (median, landmark survival or RMST) and is reported here without an interval. The remaining gap (calibrated PRODIGE-4 arm $\ebProdigePredDaysDec$~days vs.\ rebalanced MPACT $\ebMedOsRwDaysDec$~days, $\ebTreatDaysDec$~days on the median scale) is the model's estimate of the \textsc{FX}-vs-\textsc{GN} \emph{treatment-effect} contribution within the PRODIGE-4 covariate distribution, under the assumptions of \S\ref{sec:case-study-notes}. $\Delta$RMST is the primary estimand. The propagated bootstrap envelope across $B_{\textsc{gn}}\times B_{\textsc{fx}}=100\times 100=10{,}000$ replicate pairs gives a one-year $\Delta$RMST in favor of \textsc{FX} of $+\ebDrOne$~days (95\% CI $[\ebDrOneLo,\ +\ebDrOneHi]$~days) and a two-year $\Delta$RMST of $+\ebDrTwo$~days (95\% CI $[\ebDrTwoLo,\ +\ebDrTwoHi]$~days); both intervals straddle zero, but the per-pair effect is positive in $\ebDrOnePos\%$ of replicate pairs at one year and $\ebDrTwoPos\%$ at two years. The Cox hazard ratio is $\mathrm{HR}=\ebHr$; proportional hazards was not assessed, and since the curves cross (Figure~\ref{fig:mpanc-km-full}) the HR is best read as a time-averaged summary. Its 95\% envelope $[\ebHrLo,\ \ebHrHi]$ includes 1, and $\mathrm{HR}<1$ in $\ebHrBelowOne\%$ of replicate pairs.

\begin{figure}[H]
\centering
\includegraphics[width=0.92\linewidth]{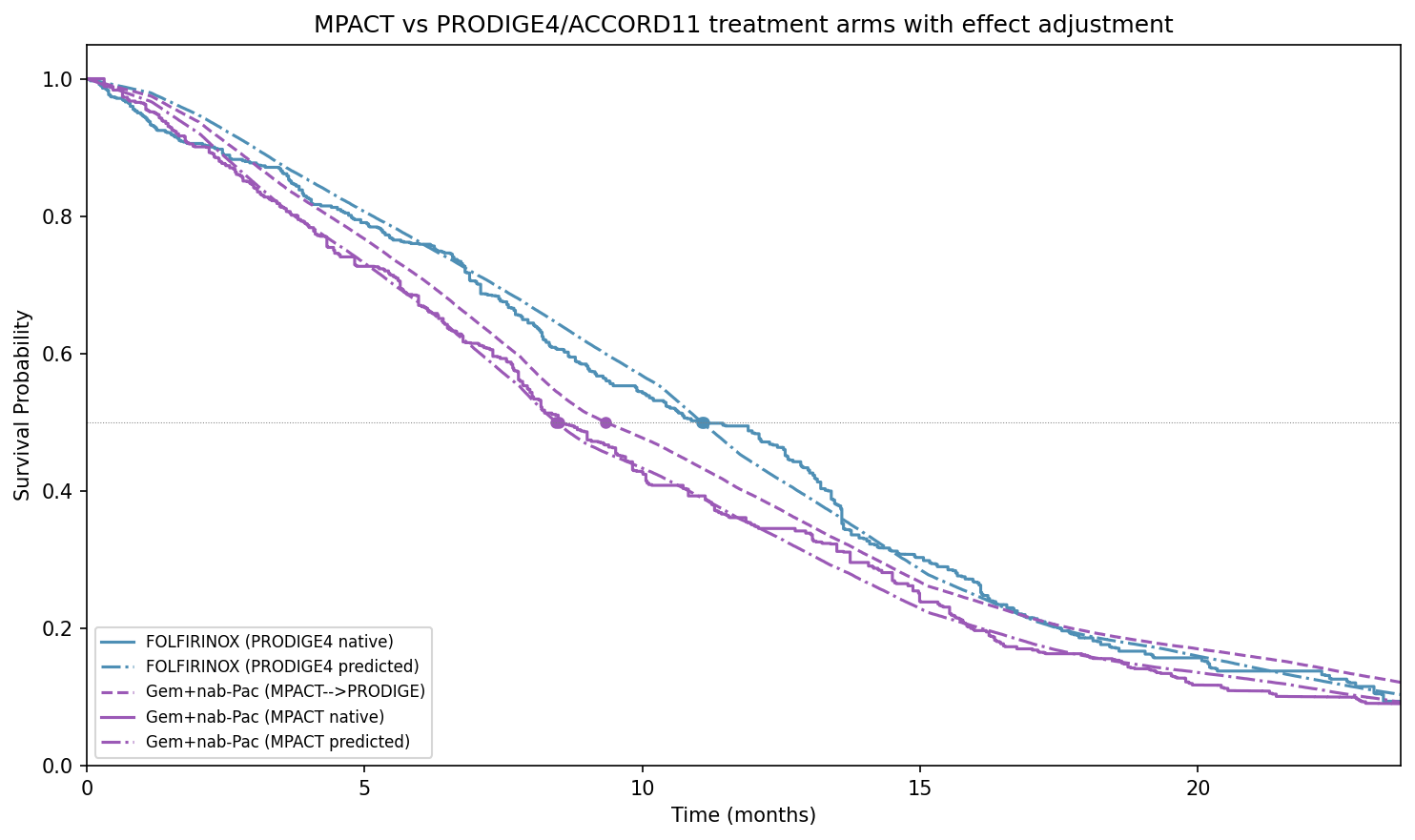}
\caption{Calibration-recovery overlay. Solid curves are published Kaplan--Meier curves (PRODIGE-4 \textsc{FX} in blue, MPACT \textsc{GN} in purple; ``native'' in the legend). Dash-dot curves are the model-predicted arms after \textsc{FRESH} calibration to each trial's own targets (Stage~1 + Stage~2 of Algorithm~\ref{alg:fresh}). The dashed purple curve is the rebalanced MPACT$\to$PRODIGE-4 arm: MPACT MH-sampled outcomes reweighted by the cross-trial entropy-balancing weights of the model-based standardization (\S\ref{sec:case-study-eb}). The primary within-population \textsc{FX}-vs-\textsc{GN} treatment-effect estimate is the $\Delta$RMST in Table~\ref{tab:mpanc-eb-results}.}
\label{fig:mpanc-km-full}
\end{figure}

\begin{figure}[H]
\centering
\includegraphics[width=1.0\linewidth]{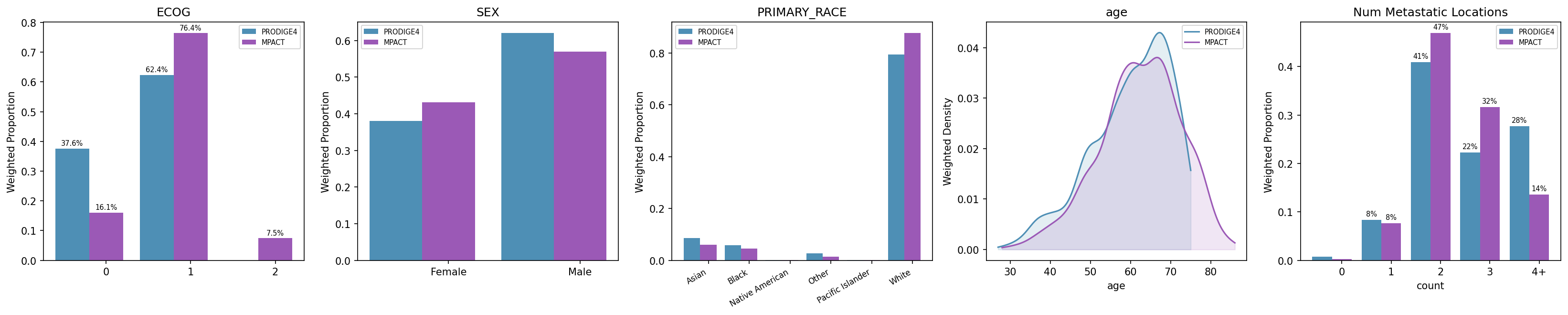}
\caption{Baseline distributions (ECOG, sex, race, age, metastatic-site count) in the MH-sampled MPACT and PRODIGE-4 cohorts after Stage~1 entropy balancing to each trial's own Table~1 statistics. The model-based standardization (\S\ref{sec:case-study-eb}) closes these gaps only on the moments PRODIGE-4 reports (Table~\ref{tab:mpanc-baseline-targets}) when computing the within-population \textsc{FX}-vs-\textsc{GN} contrast.}
\label{fig:mpanc-baseline-dists}
\end{figure}

\subsection{Case study notes}
\label{sec:case-study-notes}
\begin{itemize}
\item \textbf{Two trial-specific calibrations.} Both trial-specific surrogates are produced by the same end-to-end algorithm; nothing in the calibration is bespoke to either trial except the supplied baseline and outcome statistics.
\item \textbf{Black-box use of the underlying model.} No part of the \textsc{FRESH} procedure inspects, modifies, or differentiates through the generative transformer. Every step is either an entropy balancing of generative draws (Stage~1) or a Metropolis--Hastings chain whose proposal is direct simulation from the model (Stage~2); the model is invoked only through its two sampling interfaces, and its conditional density enters Stage~2 only as a target factor that cancels in the acceptance ratio (Remark~\ref{rem:cancellation}), so it is never evaluated by the calibration. The downstream model-based standardization likewise touches only the model's outputs.
\item \textbf{Causal interpretation requires assumptions \textsc{FRESH} does not supply.} The residual median-OS gap admits a treatment-effect interpretation only under two conditions: each calibrated arm must reproduce its trial population's joint distribution of the matched baseline features and OS, and, conditional on those features, OS under \textsc{GN} must be distributed the same way in the PRODIGE-4 population as in the rebalanced MPACT cohort (cf.\ the conditional constancy assumptions of Phillippo et al.~\cite{phillippo2018maic}). The second condition fails if anything that differs between the trials and affects OS is missing from the matched features --- unreported prognostic factors such as carbohydrate antigen 19-9 (CA19-9), or trial-level differences in era, centers and supportive care. This cross-trial unmeasured confounding is the primary threat to the contrast, and matching on reported baseline statistics cannot remove it. In short, the baseline statistics balanced in Stage~1 and in the model-based standardization must be sufficient to remove confounding between the regimen and the outcome. Neither \textsc{FRESH} nor the cross-trial step is informative about whether the matched feature set is sufficient; the assumption is the user's to defend, as with all causal-inference frameworks for observational studies.
\end{itemize}

\section[Case study: retrospective, outcome-blinded forecast of a control arm in \textit{BRAF}~V600E-mutant colorectal cancer]{Case study: retrospective, outcome-blinded forecast of a control arm in \textit{BRAF}~V600E-mutant metastatic colorectal cancer}
\label{sec:case-study-crc}

The pancreatic case study compared two simulated trial arms that were each calibrated to a single aggregate source. This second case study exercises a different application of our method --- forecasting the control arm of a trial from historical evidence alone --- and stresses three features of \textsc{FRESH} that the first did not: composition of \emph{several} calibrated sources into one calibrated population, calibration of a \emph{vector-valued} outcome (overall and progression-free survival simultaneously), and a target population that is a molecular \emph{subset} of every population the sources report on. The forecast is retrospective and outcome-blinded (held-out): nothing from the target trial's outcome data enters the computational pipeline, and its published outcomes serve as the answer key against which the forecast is scored, but the trial had read out and its results were public when the analysis was built.

\subsection{Clinical question}
\textbf{BREAKWATER} (NCT04607421)~\cite{elez2025breakwater,kopetz2025breakwater} is a phase~3 trial of encorafenib plus cetuximab, with or without mFOLFOX6, against standard-of-care chemotherapy in first-line \textit{BRAF}~V600E-mutant, microsatellite-stable metastatic colorectal cancer (mCRC). Its control arm (standard of care, \textsc{SOC}; $n=243$) is investigator's choice of mFOLFOX6, FOLFOXIRI, or CAPOX, each with or without bevacizumab (bev); the planned regimen mix at randomization was mFOLFOX6+bev $39.5\%$, FOLFOXIRI+bev $24.3\%$, CAPOX+bev $16.9\%$, mFOLFOX6 $7.8\%$, FOLFOXIRI $3.3\%$, CAPOX $2.5\%$ ($80.7\%$ with bevacizumab; Fig.~1 of~\cite{kopetz2025breakwater}), and $197$ of the $243$ patients ($81.1\%$) received bevacizumab~\cite{elez2025breakwater}. The primary PFS analysis and updated interim OS analysis~\cite{elez2025breakwater} (data cut-off 6~Jan~2025; $242$ of the $297$ deaths expected at the final OS analysis, across both arms) report control-arm median OS of $15.1$~months ($95\%$~CI $13.7$--$17.7$; $148$ deaths) and median PFS by blinded central review of $7.1$~months ($6.8$--$8.5$; $132$ events). The OS answer key is therefore interim, and may move most at the tail landmarks.

One of the most important design parameters of a clinical trial is the expected control arm performance.  It can be tricky to make such an estimate when there is little or piecemeal historical evidence pertaining to the cohort.  The usual approach is a literature summary. The BREAKWATER protocol, for example, cites retrospective \textit{BRAF}-mutant subgroup analyses of phase~3 trials of standard regimens as showing a median PFS of about $7$~months and a median OS of about $12$~months~\cite{pfizer2024breakwaterprotocol}; the largest of them, a pooled analysis of the CAIRO, CAIRO2, COIN and FOCUS trials, reports a median OS of $11.4$~months in $250$ \textit{BRAF}-mutant patients~\cite{venderbosch2014mmr}. Such a summary pools subgroups treated with different regimens, in different eras and populations, and adjusts none of them to the trial being designed; \textsc{FRESH} is meant to improve on it. \textit{BRAF}~V600E-mutant mCRC is a small, poor-prognosis subgroup (roughly $8$--$12\%$ of metastatic CRC~\cite{kopetz2025breakwater,elez2025breakwater}); the historical phase~3 trials that established the chemotherapy backbones enrolled unselected mCRC populations in an era before \textit{BRAF}-targeted second-line therapy existed. A forecast must therefore (i) recover the chemotherapy outcome levels of the unselected historical trials, (ii) carry those levels down onto the \textit{BRAF}~V600E, microsatellite-stable subpopulation the target trial enrolls, and (iii) transport them onto the target trial's covariate distribution.

\paragraph{Calibration sources.} Two historical trials serve as the aggregate evidence sources of \S\ref{sec:setup} for the procedure; we refer to them as the \emph{source trials}:
\begin{itemize}
\item \textbf{NO16966 (XELOX-1)} (NCT00069095)~\cite{cassidy2008xelox,saltz2008bevacizumab,cassidy2011no16966}: a $2{\times}2$ factorial phase~3 trial of XELOX (CAPOX) versus FOLFOX4, each with bevacizumab or placebo, in unselected first-line mCRC ($n=2{,}034$; enrolled 2003--2005, $634$ patients in an initial two-arm XELOX-versus-FOLFOX4 part and $1{,}400$ in the factorial part~\cite{cassidy2011no16966}). Its publications report OS and PFS for the two \emph{backbone-pooled} populations --- all XELOX patients ($n=1{,}017$; median OS $19.8$, PFS $8.0$~months) and all FOLFOX4 patients ($n=1{,}017$; $19.5$ and $8.5$~months) --- and that is the only backbone-specific cut for which both an OS and a PFS Kaplan--Meier curve exist. NO16966 is the sole source for the CAPOX and FOLFOX backbones.
\item \textbf{TRIBE} (NCT00719797)~\cite{loupakis2014tribe,cremolini2015tribe}: a phase~3 trial of FOLFOXIRI+bev versus FOLFIRI+bev in unselected first-line mCRC (Italy; $n=508$; enrolled 2008--2011). The FOLFOXIRI+bev arm ($n=252$; median OS $29.8$~months, $95\%$~CI $26.0$--$34.3$; PFS $12.1$, $10.9$--$13.2$) is the sole source for the FOLFOXIRI backbone. TRIBE also publishes medians within its \textit{BRAF}-mutant stratum ($n=16$ on FOLFOXIRI+bev: OS $19.0$, $8.2$--$28.6$), which we use only as a diagnostic, never as a constraint. The TRIBE-cohort particles of the final forecast, which are \textit{BRAF}~V600E and weighted to BREAKWATER's baseline (\S\ref{sec:crc-composition}), have a forecast median OS of $\crcTribeFcOs$~months, inside that stratum's interval; with $16$ patients and two different populations this is a weak check.
\end{itemize}
Between them they cover $91\%$ of the control arm's patients (five of its six reported regimen sets; FOLFOXIRI without bevacizumab, $3.3\%$, is the one left out), and both contain the target population on every eligibility axis the simulator can evaluate (unselected first-line mCRC, ECOG~0--1, no brain metastases). NO16966 reports neither \textit{BRAF} nor microsatellite status; TRIBE reports \textit{BRAF} status, but its $16$ \textit{BRAF}-mutant FOLFOXIRI+bev patients are too few to calibrate on; and neither reports Asian enrollment. Neither trial therefore supplies a \textit{BRAF}~V600E calibration target, so the \textit{BRAF}~V600E prognosis within each calibrated population and the race transport are supplied entirely by the generative model.

Formally, the target estimand is the outcome distribution --- OS and PFS marginals, summarized by medians and landmark survival probabilities --- of the BREAKWATER control-arm population, i.e.\ the expectations $\E_{X\sim Q_{\textsc{bw}},\,Y\sim K_{\textsc{soc}}(\cdot\mid X)}[f_j(Y)]$ of the outcome summaries $f_j$, where $Q_{\textsc{bw}}$ is defined by BREAKWATER's eligibility set and published Table~1 and $K_{\textsc{soc}}$ is the (unobserved) conditional outcome distribution under standard-of-care chemotherapy. \textsc{FRESH} constructs the surrogate for $K_{\textsc{soc}}$ from the source trials and the surrogate for $Q_{\textsc{bw}}$ from the target trial's baseline report; the scored quantities are the outcome summaries of their composition.

\subsection{Generative model}
\label{sec:crc-model}
The generative model is the same class of black-box simulator used in \S\ref{sec:case-study} --- a generative transformer trained on a real-world clinico-genomic corpus of more than $300{,}000$ oncogene panels with linked clinical features, regimens, and outcomes --- accessed only through the two interfaces of Algorithm~\ref{alg:fresh}. Two features of the deployment matter for what follows. First, the outcome is vector-valued, $Y=(Y_{\textsc{os}},Y_{\textsc{pfs}})$, and the conditional sampler $\Kref(\cdot\mid x)$ draws both components jointly; the construction of \S\ref{sec:setup} extends unchanged, with each outcome constraint $f_j$ acting on one component. Second, the generative model is trained on a random five-fold split of its training corpus, which yields five independently trained conditional outcome models $\Kref^{(1)},\dots,\Kref^{(5)}$, one per training fold; we refer to these as \emph{fold-models} and index them by $r=1,\dots,F$, with $F=5$. The whole calibration chain is run once per fold-model, and the spread across fold-models is one of the two uncertainty axes carried into the reported intervals (\S\ref{sec:crc-error-quant}). Baseline draws from $\Qref$ are likewise pooled equally across the five fold-models of the baseline sampler. As before, nothing about architecture, training, or corpus enters the calibration, with two exceptions that we state explicitly: which saved training checkpoint each fold-model starts from, and a fine-tuning of each fold-model on the first-line metastatic part of its own training corpus.

\paragraph{Fold-model checkpoints.} The forecast inside the \textit{BRAF}~V600E subset depends on how well each fold-model separates V600E from wild-type prognosis, and the minimum-validation-loss checkpoints (training steps \steplist{\crcCkptMinLoss} for the five folds) under-separate it: on out-of-fold real first-line patients, each predicted by the one fold-model that did not train on it, their predicted V600E-to-wild-type median-OS ratio is \crcCkptRatioMinLoss{} against an observed \crcCkptRatioObserved. We therefore chose each fold-model's checkpoint by a subgroup criterion. Among the saved checkpoints whose validation loss lies within $\crcCkptBand\%$ of that fold's minimum, we took the one with the smallest error in the predicted first-line V600E-versus-wild-type gap in 48-month restricted-mean OS on the fold's own out-of-fold first-line patients ($\crcCkptNVsixMin$--$\crcCkptNVsixMax$ V600E patients per fold). This selects steps \steplist{\crcCkptSelected}, at a validation-loss cost of \crcCkptLossRiseMin--\crcCkptLossRiseMax\%, and moves the out-of-fold median-OS ratio to \crcCkptRatioSelected{} (C-index \crcCkptCSelected{} against \crcCkptCMinLoss). No BREAKWATER data enter the criterion.

\paragraph{First-line fine-tuning.} Each fold-model was fine-tuned on the first-line metastatic part of the corpus it was trained on (colorectal adenocarcinoma, stage~IV, no prior systemic therapy; $\crcFtN$ patients with a recorded OS outcome, $\crcFtNVsix$ of them \textit{BRAF}~V600E), starting from its selected checkpoint. Fold-model $r$ is fine-tuned on the training folds of this subset and validated on its held-out fold ($\crcFtHeldMin$--$\crcFtHeldMax$ patients). This is a fine-tune on real-world patients, not on calibrated synthetic data. A single first-line subset serves all three source cohorts; per-regimen subsets are too small.

The fine-tune runs for $\crcFtEpochs$ epochs ($\crcFtStepsMin$--$\crcFtStepsMax$ steps) at batch size $\crcFtBatch$, with learning rate $3\times10^{-6}$ under a cosine schedule. The rate was chosen from $\{3\times10^{-5},10^{-5},3\times10^{-6}\}$. Because BREAKWATER enrolls only \textit{BRAF}~V600E-mutant patients, the accuracy of the forecast depends heavily on the prognosis the fold-models associate with that mutation. We therefore stopped the fine-tune at the step that best fits the differential V600E prognosis on validation data, subject to a limit on the loss of overall fit. Among the scored steps whose validation loss, averaged over the five folds, stays within about $3\%$ of the starting checkpoint's and whose held-out OS C-index falls by at most $0.01$, we took the one whose predicted V600E-to-wild-type median ratios on held-out trial-eligible patients lie closest to the ratios observed in those patients. This is step~$\crcFtStep$, used for all five fold-models: its validation loss is $\crcFtLossRiseMean\%$ above the starting checkpoint's on average ($\crcFtLossRiseMin$--$\crcFtLossRiseMax\%$ by fold), and its OS C-index is $\crcFtCSelected$ against $\crcFtCStart$. At that step the predicted out-of-fold V600E-to-wild-type median ratio is $\crcFtTribePfsRatio$ for PFS and $\crcFtTribeOsRatio$ for OS on TRIBE-eligible patients, and $\crcFtNoPfsRatio$ and $\crcFtNoOsRatio$ on NO16966-eligible ones. No BREAKWATER data enter the fine-tune or the choice of step; the step and the reported ratios are read on the same held-out patients.

\subsection{Procedure}
\label{sec:crc-procedure}
The chain has four steps: (a) draw one large \emph{candidate sample} of unselected first-line mCRC baselines from $\Qref$; (b) partition it by chemotherapy backbone into three \emph{source cohorts}, and run Algorithm~\ref{alg:fresh} in each cohort against that cohort's own aggregate calibration data --- the cohort's eligibility set, Table~1, and OS/PFS curves; (c) pool the three calibrated cohorts, apply BREAKWATER's eligibility filter, and run Stage~1 a final time against BREAKWATER's Table~1; (d) compare the resulting weighted cohort against BREAKWATER's published control-arm outcomes.

\subsubsection{Candidate sample}
We draw a large sample of baselines from $\Qref$ conditioned on the strata common to all three trials --- colorectal adenocarcinoma, TNM stage~IV, no prior systemic therapy, ECOG~0--1 --- and on one of six regimen sets: five of the six sets BREAKWATER's control arm reports (CAPOX, CAPOX+bev, FOLFOX, FOLFOX+bev, FOLFOXIRI+bev; FOLFOXIRI alone, $3.3\%$ of the arm, is dropped) plus FOLFIRI+bev, TRIBE's control regimen. The FOLFIRI+bev draws fall in no source cohort (TRIBE's cohort is its FOLFOXIRI+bev arm) and play no part in the calibration or the forecast. Every other covariate --- age, sex, race, metastatic sites, microsatellite status, CEA, primary sidedness, and the tumor genotype --- is sampled freely from the model. In particular \textit{BRAF}, \textit{KRAS}, and \textit{NRAS} are \emph{not} conditioned on at this step: the candidate sample must be an unselected population, because each source trial's published curves are the right calibration targets only for the population in which they were observed. The sampler was required to emit a value (rather than a missing indicator) for each covariate that enters a constraint or an eligibility filter; the same missingness penalty as in \S\ref{sec:case-study} ($\rho_{\text{miss}}=10^{-3}$) is carried in every Stage-1 problem below.

\subsubsection{Stage 1 (per cohort): baseline entropy balancing}
Each source cohort $u\in\{\textsc{XEL},\textsc{FOX},\textsc{TRI}\}$ is a regimen-defined slice of the candidate sample: \textsc{XEL} (CAPOX, CAPOX+bev) is calibrated to NO16966's pooled-XELOX population, \textsc{FOX} (FOLFOX, FOLFOX+bev) to NO16966's pooled-FOLFOX4 population, and \textsc{TRI} (FOLFOXIRI+bev) to TRIBE's FOLFOXIRI+bev arm. Within a cohort, the candidate draws are filtered to the cohort's eligibility set $\X_{\textsc{ie}}^{u}$ (Appendix~\ref{app:crc-ie-specs}) and entropy-balanced to the source trial's published Table~1 (Table~\ref{tab:crc-baseline-targets}): age, sex, ECOG, indication, and the within-cohort regimen mix (two thirds backbone alone, one third with bevacizumab, for the two NO16966 cohorts). All tabulated constraints are hard equalities; the regimen constraint assigns zero weight to draws outside the cohort's regimen sets, so each candidate draw carries weight in at most one cohort (the FOLFIRI+bev draws in none). Output: per-cohort weights $\{w_i^{u}\}$ via~\eqref{eq:stage1-empirical}, whose weighted empirical distribution is the cohort's calibrated baseline $\Qstar_u$.

\subsubsection{Stage 2 (per cohort): simultaneous OS/PFS calibration by MH}
Stage~2 runs the MH chain of \S\ref{sec:mh} in each cohort against its source trial's OS \emph{and} PFS targets simultaneously (Table~\ref{tab:crc-outcome-targets}). All Stage-2 constraints are hard ($\mathcal{S}_Y=\varnothing$). Each cohort runs $F\times B$ chains, one per fold-model $\Kref^{(r)}$, $r=1,\dots,F$, and bootstrap replicate $b=1,\dots,B$, each cold-started from $\Kref^{(r)}$ draws and run against its own bootstrap-perturbed version of the targets (\S\ref{sec:crc-error-quant}); no chain uses the unperturbed targets of Table~\ref{tab:crc-outcome-targets}. Hyperparameters are in Appendix~\ref{app:crc-mh-hyperparameters}. Within a chain every target is exact; the sampling error of the published curves enters through the bootstrap rather than through soft penalties. Soft constraints with $\rho_j$ set from each target's published standard error~\cite{hummer2015bayesian} would absorb that error inside each chain, but would not by themselves yield the source-sampling component of the forecast interval that this design reports. The output of a cohort is, for each $(r,b)$ cell, one calibrated outcome draw per particle, $y_i\sim K^{\ast,(r,b)}_{u}(\cdot\mid x_i)$.

\subsubsection{Composition: pooling, target eligibility, and target-trial entropy balancing}
\label{sec:crc-composition}
For each fold-model $r$, the three calibrated cohorts are concatenated into one sample $\{(x_i, y_i)\}$ carrying a source-cohort label. BREAKWATER's eligibility filter $\X_{\textsc{ie}}^{\textsc{bw}}$ (Appendix~\ref{app:crc-ie-specs}) is then applied: \textit{BRAF}~V600E present; not microsatellite-instable; \textit{KRAS} and \textit{NRAS} wild-type; ECOG~0--1; first line; no brain metastases. The cohort-specific Stage-1 weights are \emph{not} carried forward: what each cohort contributes is its per-particle calibrated kernel $\Kstar_u(\cdot\mid x)$, and the forecast population is set by BREAKWATER's baseline alone (\S\ref{sec:crc-notes}). The pooled eligible set enters a final Stage-1 problem, the \emph{target-trial entropy-balancing step}, with uniform reference weights and BREAKWATER's published Table~1 as targets (last column of Table~\ref{tab:crc-baseline-targets}): age, sex, ECOG, race, liver metastases, number of involved metastatic organs, CEA, and the five-way regimen mix. The regimen constraint sets the share of each cohort in the forecast (FOLFOX4-cohort particles $52.0\%$, TRIBE $26.7\%$, XELOX $21.3\%$).

\subsubsection{Error quantification: a nested fold-model~$\times$~bootstrap design}
\label{sec:crc-error-quant}
The reported intervals propagate two uncertainty sources through the whole chain: which fold-model calibrates the particles ($F=5$), and how well each source trial's published curves pin its own population ($B$ subject-level bootstrap replicates of each source trial's reconstructed pseudo-IPD). The pseudo-IPD pipeline is that of \S\ref{sec:case-study-error-quant} (reconstruction with the algorithm of Guyot et al.~\cite{guyot2012reconstruction}, then a subject-level bootstrap~\cite{efron1981censored}), with one difference forced by the sources: TRIBE's curves print numbers at risk and reconstruct directly (deviation from the digitized curve $0.010$ for OS, $0.017$ for PFS; published event totals reproduced exactly), but none of NO16966's figures prints an at-risk table; for those four curves we reconstruct against a \emph{derived} at-risk table $n(t)=\mathrm{round}(n_0\,\widehat S(t)\,G(t))$ with $G$ the survival function of a uniform censoring window, choosing the window so that the reconstruction reproduces both the digitized curve (deviation $\le 0.002$) and the published event total (to within one event). The censoring pattern behind the NO16966 replicates is therefore assumed rather than observed, and their share of the interval is an approximation in a way TRIBE's is not.

Every $(r,b)$ cell of every cohort draws its own subject-level bootstrap resample of the source trial's pseudo-IPD and its own MH seed, so the $B$ replicates of one fold-model are independent of those of another; pairing the three cohorts' cells by index into one $(r,b)$ cell of the pooled forecast is a legitimate joint draw, since the cohorts' source populations are disjoint. Within a cohort the two endpoints are resampled independently from their own reconstructions, so per-endpoint intervals are unaffected; joint OS/PFS intervals ignore the sampling correlation between the endpoints and may be too wide or too narrow. All $F\times B$ Stage-2 chains are cold-started. The headline interval on each statistic is the $(q_{0.025},q_{0.975})$ envelope of the $F\times B$ cell estimates under the common BREAKWATER weights. Because replicates are nested within fold-models, the total sum of squares across cells partitions exactly into a between-fold-model component and a within-fold-model component; the latter carries the calibration-source bootstrap and MH chain noise together, which this design does not separate. We report that decomposition alongside the interval.

The run reported below uses $B=\crcNumReps$ replicates per source cohort, so each source cohort is calibrated in $F\times B=\crcNumCells$ cells and the three cohorts together in \crcNumChains{} Stage-2 chains; the within-fold-model variance component of the forecast therefore rests on $F(B-1)$ degrees of freedom.

What the interval contains: the choice of fold-model, the sampling error of each source trial's published curves as it enters the landmark targets, and MH chain noise. What it excludes by design: BREAKWATER's own Table-1 sampling error (the target-trial weights are identical in every cell), the finite-sample error of the candidate sample, digitization and reconstruction error including the assumed NO16966 censoring pattern, between-trial heterogeneity --- trial and era effects on outcome that the model's covariates do not capture, which meta-analytic-predictive priors for historical controls carry as a variance component~\cite{neuenschwander2010map} and which can survive adjustment for prognostic factors in pooled historical benchmarks, at least for PFS~\cite{korn2008melanoma} --- and every structural limitation of the generative model (\S\ref{sec:crc-notes}).

\subsection{Calibration constraints supplied to each cohort}
\label{sec:crc-constraints}
Table~\ref{tab:crc-baseline-targets} lists the Stage-1 baseline targets for the three source cohorts and for the target-trial entropy-balancing step onto BREAKWATER's Table~1; Table~\ref{tab:crc-outcome-targets} lists the Stage-2 OS and PFS targets for the three cohorts. BREAKWATER's published outcomes (Table~\ref{tab:crc-forecast}) are never supplied to any stage.

\begin{table}[H]
\centering
\small
\setlength{\tabcolsep}{5pt}
\renewcommand{\arraystretch}{1.1}
\resizebox{\linewidth}{!}{%
\begin{tabular}{@{}l c c c c@{}}
\toprule
 & \multicolumn{3}{c}{Calibration cohorts} & Target \\
\cmidrule(lr){2-4}\cmidrule(lr){5-5}
Statistic & \textsc{XEL}: NO16966 XELOX & \textsc{FOX}: NO16966 FOLFOX4 & \textsc{TRI}: TRIBE FOLFOXIRI+bev & BREAKWATER \textsc{SOC} \\
 & pooled, $n=1{,}017$ & pooled, $n=1{,}017$ & $n=252$ & $n=243$ \\
\midrule
\multicolumn{5}{l}{\textit{Age (years)}} \\
\quad median                          & 61        & ---$^{\dag}$ & 60.5     & 62 \\
\quad min, max                        & 18, 86    & 19, 83       & 18, 75   & 28, 84 \\
\multicolumn{5}{l}{\textit{Sex (proportions)}} \\
\quad Male                            & 0.602     & 0.585        & 0.595    & 0.490 \\
\quad Female                          & 0.398     & 0.415        & 0.405    & 0.510 \\
\multicolumn{5}{l}{\textit{ECOG performance status (proportions, renormalized)}} \\
\quad 0                               & 0.565     & 0.566        & 0.901    & 0.572 \\
\quad 1                               & 0.435     & 0.434        & 0.099    & 0.428 \\
\multicolumn{5}{l}{\textit{Race (proportions, renormalized)}} \\
\quad White                           & ---       & ---          & ---      & 0.613 \\
\quad Asian                           & ---       & ---          & ---      & 0.387 \\
\multicolumn{5}{l}{\textit{Metastatic disease}} \\
\quad liver metastases (proportion)   & ---       & ---          & ---      & 0.658 \\
\quad involved organs, $\Pr(\text{organs}\le 2)$ & --- & ---  & ---      & 0.523 \\
\multicolumn{5}{l}{\textit{Serum CEA ($\mu$g/L)}} \\
\quad $\Pr(\text{CEA}\le 5)$           & ---       & ---          & ---      & 0.279 \\
\multicolumn{5}{l}{\textit{Indication (proportions)}} \\
\quad colorectal adenocarcinoma       & 1.00      & 1.00         & 1.00     & 1.00 \\
\multicolumn{5}{l}{\textit{Regimen (proportions, renormalized)}} \\
\quad CAPOX                           & 0.656     & ---          & ---      & 0.027 \\
\quad CAPOX + bev                     & 0.344     & ---          & ---      & 0.186 \\
\quad FOLFOX                          & ---       & 0.657        & ---      & 0.086 \\
\quad FOLFOX + bev                    & ---       & 0.343        & ---      & 0.434 \\
\quad FOLFOXIRI + bev                 & ---       & ---          & 1.00     & 0.267 \\
\bottomrule
\end{tabular}}
\caption{Baseline-statistic constraints for the colorectal case study. The three source-cohort columns are the targets of the per-cohort Stage-1 problems; the BREAKWATER column is the target of the target-trial entropy-balancing step of \S\ref{sec:crc-composition}, applied to the pooled, BREAKWATER-eligible, calibrated particles. All entries are hard equalities; each Stage-1 problem also carries the missingness penalty of \S\ref{sec:case-study}. Proportion blocks marked ``renormalized'' are the published proportions renormalized to sum to one over the listed levels (BREAKWATER reports $5.8\%$ ECOG missing, $3.3\%$ other race, and $9.0\%$ of the arm on regimens outside the five listed sets or untreated). The organ count excludes lymph-node and unspecified metastatic locations, and the CEA share is over the $226$ patients with a recorded value. Sources: NO16966 pooled columns from Cassidy et al.~\cite{cassidy2008xelox} Table~1; TRIBE from Loupakis et al.~\cite{loupakis2014tribe} Table~1 and Cremolini et al.~\cite{cremolini2015tribe}; BREAKWATER from Elez et al.~\cite{elez2025breakwater} Table~1, with the regimen mix from Kopetz et al.~\cite{kopetz2025breakwater} Fig.~1 (data cut-off 22~Dec~2023). $^{\dag}$The pooled-FOLFOX4 age median is not published and is not derivable from the constituent-arm medians ($62/60/60$); only the range is constrained.}
\label{tab:crc-baseline-targets}
\end{table}

\begin{table}[H]
\centering
\small
\setlength{\tabcolsep}{5pt}
\renewcommand{\arraystretch}{1.15}
\resizebox{\linewidth}{!}{%
\begin{tabular}{@{}l l l@{}}
\toprule
Cohort & Endpoint & Landmark targets $t$~(d)~:~$S(t)$; \textbf{median} $t$~(d) \\
\midrule
\textsc{XEL}: NO16966 XELOX pooled   & OS  & 170:0.90, 288:0.80, 383:0.70, 489:0.60, \textbf{603}, 722:0.40, 906:0.30, 1231:0.20 \\
                                   & PFS & 90:0.90, 136:0.80, 177:0.70, 208:0.60, \textbf{244}, 285:0.40, 330:0.30, 393:0.20, 548:0.10 \\
\addlinespace
\textsc{FOX}: NO16966 FOLFOX4 pooled & OS  & 196:0.90, 306:0.80, 409:0.70, 509:0.60, \textbf{594}, 708:0.40, 912:0.28, 1131:0.20 \\
                                   & PFS & 98:0.90, 147:0.80, 183:0.70, 222:0.60, \textbf{259}, 295:0.40, 342:0.30, 412:0.20, 505:0.10 \\
\addlinespace
\textsc{TRI}: TRIBE FOLFOXIRI+bev    & OS  & 365:0.83, 730:0.60, \textbf{907}, 1095:0.40, 1460:0.27 \\
                                   & PFS & 183:0.84, \textbf{368}, 548:0.30, 730:0.19, 913:0.14, 1096:0.11 \\
\bottomrule
\end{tabular}}
\caption{Stage~2 outcome calibration targets supplied to the MH chain in each source cohort; OS and PFS are calibrated simultaneously. These are the unperturbed targets; every chain runs against its own bootstrap-perturbed version of them (\S\ref{sec:crc-error-quant}). Landmark entries are pairs $t:S(t)$, each a hard constraint $\Pr[Y_e\le t]=1-S(t)$ on the component $Y_e$ of $Y$ for endpoint $e\in\{\textsc{os},\textsc{pfs}\}$, read off the digitized published Kaplan--Meier curve. For NO16966 the landmark times lie near the decile crossings of the curve and the survival values are read off the digitized curve at those times (hence FOLFOX4's $912{:}0.28$); the NO16966 OS grids stop at $S=0.20$, so below that level the OS tails of these cohorts are not targeted. TRIBE's landmarks are at 6-month intervals for PFS and 12-month intervals for OS. Bold entries are the published medians, supplied as $50$th-percentile constraints (NO16966 XELOX / FOLFOX4 pooled: OS $19.8$ / $19.5$~mo, PFS $8.0$ / $8.5$~mo~\cite{cassidy2008xelox,cassidy2011no16966}; TRIBE: OS $29.8$, PFS $12.1$~mo~\cite{cremolini2015tribe,loupakis2014tribe}; days $=$ months $\times\,30.4375$). NO16966 PFS is investigator-assessed (ITT population) at the 31~Jan~2006 cut-off~\cite{cassidy2008xelox}; NO16966 OS is from the updated analysis~\cite{cassidy2011no16966}. Two TRIBE PFS landmarks at 42 and 48 months ($S=0.09$, $0.07$) are not used: a support check found no simulated PFS draw beyond $40.2$ months. All entries are hard constraints ($\mathcal{S}_Y=\varnothing$).}
\label{tab:crc-outcome-targets}
\end{table}

\subsection{Results}
\label{sec:crc-results}

\paragraph{Calibration recovery within each cohort.} Table~\ref{tab:crc-cohort-recovery} reports, for each cohort, the published OS and PFS medians with their $95\%$ confidence intervals alongside the calibrated medians under the cohort's own Stage-1 weights, summarized over the $F\times B=\crcNumCells$ cells. Two things are being checked. The first is the point estimate: every cohort reproduces its source trial's published medians to within the replicate spread. The second is uncertainty propagation: the calibrated envelope should be about as wide as the source's own confidence interval, because the replicate axis is a subject-level bootstrap of that source and the other contributions within a cohort are small (the fold-model term, measured below; MH noise, which the design does not separate). This checks that the source uncertainty is carried through, not that the forecast variance is right. The envelope-to-interval width ratio runs \crcWidthRatioMin--\crcWidthRatioMax{} across the six endpoints (OS \crcXeloxOsWidthRatio, \crcFolfoxOsWidthRatio, \crcTribeOsWidthRatio{} and PFS \crcXeloxPfsWidthRatio, \crcFolfoxPfsWidthRatio, \crcTribePfsWidthRatio{} for \textsc{XEL}, \textsc{FOX}, \textsc{TRI}). Read as standard deviations, the replicate-to-replicate spread of a cohort's own OS median (\crcNoRepSdMin--\crcNoRepSdMax~months for the two $n=1{,}017$ NO16966 columns, \crcTribeRepSd~months for TRIBE's $n=252$ arm) sits against a bootstrap standard deviation of the source median of \crcXeloxOsPipdSd, \crcFolfoxOsPipdSd, and \crcTribeOsPipdSd~months respectively. Within a cohort the fold-model axis is negligible (the fold-model means of landmark survival differ by a median standard deviation of \crcFoldSdMin--\crcFoldSdMax, and between-fold-model variation accounts for at most $\crcUnitFoldShareMax\%$ of the variance of a cohort's calibrated median), because each cohort is pinned to its own curve; the fold-model term in the forecast arises downstream.

\begin{table}[H]
\centering
\small
\setlength{\tabcolsep}{6pt}
\renewcommand{\arraystretch}{1.15}
\resizebox{\linewidth}{!}{%
\begin{tabular}{@{}l c c c c@{}}
\toprule
 & \multicolumn{2}{c}{OS, months} & \multicolumn{2}{c}{PFS, months} \\
\cmidrule(lr){2-3} \cmidrule(lr){4-5}
Cohort & Published ($95\%$ CI) & Calibrated [$95\%$] & Published ($95\%$ CI) & Calibrated [$95\%$] \\
\midrule
\textsc{XEL}: NO16966 XELOX pooled   & $\crcXeloxOsPub$ ($\crcXeloxOsPipdLo$--$\crcXeloxOsPipdHi$)$^{\dag}$ & $\crcXeloxOsCal$ [$\crcXeloxOsCalLo$, $\crcXeloxOsCalHi$] & $\crcXeloxPfsPub$ ($\crcXeloxPfsPipdLo$--$\crcXeloxPfsPipdHi$)$^{\dag}$ & $\crcXeloxPfsCal$ [$\crcXeloxPfsCalLo$, $\crcXeloxPfsCalHi$] \\
\textsc{FOX}: NO16966 FOLFOX4 pooled & $\crcFolfoxOsPub$ ($\crcFolfoxOsPipdLo$--$\crcFolfoxOsPipdHi$)$^{\dag}$ & $\crcFolfoxOsCal$ [$\crcFolfoxOsCalLo$, $\crcFolfoxOsCalHi$] & $\crcFolfoxPfsPub$ ($\crcFolfoxPfsPipdLo$--$\crcFolfoxPfsPipdHi$)$^{\dag}$ & $\crcFolfoxPfsCal$ [$\crcFolfoxPfsCalLo$, $\crcFolfoxPfsCalHi$] \\
\textsc{TRI}: TRIBE FOLFOXIRI+bev    & $\crcTribeOsPub$ ($\crcTribeOsPubLo$--$\crcTribeOsPubHi$)          & $\crcTribeOsCal$ [$\crcTribeOsCalLo$, $\crcTribeOsCalHi$] & $\crcTribePfsPub$ ($\crcTribePfsPubLo$--$\crcTribePfsPubHi$)      & $\crcTribePfsCal$ [$\crcTribePfsCalLo$, $\crcTribePfsCalHi$] \\
\bottomrule
\end{tabular}}
\caption{Per-cohort calibration recovery. Published medians are the Stage-2 median constraints of Table~\ref{tab:crc-outcome-targets}; TRIBE's confidence intervals are as printed (OS from Cremolini et al.~\cite{cremolini2015tribe}, PFS from Loupakis et al.~\cite{loupakis2014tribe}). $^{\dag}$NO16966 prints no confidence interval on its pooled medians~\cite{cassidy2008xelox,cassidy2011no16966}; the interval shown is the Brookmeyer--Crowley $95\%$ interval~\cite{brookmeyer1982confidence} (untransformed Greenwood variance~\cite{greenwood1926natural}) of the median of the pseudo-IPD reconstruction of \S\ref{sec:crc-error-quant} (reconstructed medians $\crcXeloxOsPipd$ and $\crcXeloxPfsPipd$ for XELOX, $\crcFolfoxOsPipd$ and $\crcFolfoxPfsPipd$ for FOLFOX4, all within $\crcNoPipdGapAbsMax$~months of the printed medians). The same procedure applied to TRIBE's curves, which do print at-risk rows, recovers the published intervals to within $\crcTribePipdCiMaxDiff$~months (OS $\crcTribeOsPipdLo$--$\crcTribeOsPipdHi$, PFS $\crcTribePfsPipdLo$--$\crcTribePfsPipdHi$). Calibrated medians are weighted Kaplan--Meier medians under the cohort's own Stage-1 weights, mean and $(q_{0.025},q_{0.975})$ across the $\crcNumCells$ fold-model~$\times$~replicate cells.}
\label{tab:crc-cohort-recovery}
\end{table}

\paragraph{Forecast against the held-out readout.} After pooling and BREAKWATER eligibility, the BREAKWATER Table-1 targets enter the target-trial entropy-balancing step as hard constraints; achieved balance for that step is not reported separately. Table~\ref{tab:crc-forecast} and Figure~\ref{fig:crc-km} compare the forecast to the published control arm. We score each statistic in both directions: whether the forecast mean lies inside the trial's $95\%$ confidence interval, and whether the published value lies inside the forecast's $95\%$ envelope. The first holds for $\crcNumPredInPub$ of the $\crcNumBenchmarks$ statistics (OS $\crcOsNumPredInPub$ of $\crcOsNumBenchmarks$, PFS $\crcPfsNumPredInPub$ of $\crcPfsNumBenchmarks$), the second for $\crcNumPubInPred$ (OS $\crcOsNumPubInPred$ of $\crcOsNumBenchmarks$, PFS $\crcPfsNumPubInPred$ of $\crcPfsNumBenchmarks$). Neither check includes the sampling error of the published arm in the forecast interval. The OS median forecast is $\crcOsMedPred$~months ($95\%$ envelope $[\crcOsMedLo, \crcOsMedHi]$~months) against a published $\crcOsMedPub$~months ($\crcOsMedPubLo$--$\crcOsMedPubHi$~months): the published value lies inside the forecast interval, and the forecast mean lies inside the trial's confidence interval. The OS landmark forecast differs from the published curve by $\crcOsSixExcess$ percentage points at 6 months and by $\crcOsTwelveExcess$, $\crcOsEighteenExcess$, $\crcOsTwentyFourExcess$ and $\crcOsThirtyExcess$ points at 12, 18, 24 and 30 months. The forecast mean lies inside the published confidence interval at every landmark, but from 12 months on the published values lie above the forecast envelope. The PFS median forecast is $\crcPfsMedPred$~months ($[\crcPfsMedLo, \crcPfsMedHi]$~months) against $\crcPfsMedPub$~months ($\crcPfsMedPubLo$--$\crcPfsMedPubHi$~months): the forecast lies inside the published interval, but the published median lies $\crcPfsMedLoMinusPub$~months below the forecast envelope. The landmark shape differs too. The forecast PFS curve is above the published one at 6 months, where the published value lies inside the forecast envelope, and below it from 12 months on, where every published landmark lies above the envelope.

This gap does not mean that the source curves lack such a tail: TRIBE's own targets are $0.19$ at 24 months and $0.14$ at 30 months (Table~\ref{tab:crc-outcome-targets}). It comes from how the forecast tail is formed. The Stage-2 statistics here depend on $y$ alone, so the tilt is the same function of $y$ for every particle, an intercept-only recalibration in the sense of van Houwelingen~\cite{vanhouwelingen2000validation}; with indicator landmarks it leaves each particle's hazard beyond the last landmark equal to $\Kref$'s (approximately, with the sigmoid of \S\ref{sec:sigmoid}), so the tail is extrapolated by the model~\cite{jackson2017extrapolating}. NO16966's last PFS landmarks sit at 548 and 505 days, and its two cohorts carry $\crcNoShare\%$ of the forecast weight, so beyond 17--18 months most of the forecast PFS tail is the model's own \textit{BRAF}~V600E prognosis carried through the tilt, which no source curve pins. The published tail is itself thin: the printed at-risk row has $11$ patients at 18 months and $3$ at 24 (Figure~\ref{fig:crc-km}). Part of the gap may also be a difference between the trials, which the interval excludes (\S\ref{sec:crc-error-quant}). On the investigator-assessed basis the source trials use, BREAKWATER's control arm has $160$ PFS events rather than $132$, a median of $7.0$ ($6.0$--$8.3$)~months and a lower tail ($0.575$, $0.263$ and $0.116$ at 6, 12 and 18 months~\cite{ema2026braftovi}); against these values the 12- and 18-month gaps shrink from $\crcPfsTwelveExcess$ and $\crcPfsEighteenExcess$ to $\crcPfsTwelveExcessInv$ and $\crcPfsEighteenExcessInv$ points, while the 6-month excess grows from $\crcPfsSixExcess$ to $\crcPfsSixExcessInv$ points.

\begin{table}[H]
\centering
\small
\setlength{\tabcolsep}{4.5pt}
\renewcommand{\arraystretch}{1.15}
\resizebox{\linewidth}{!}{%
\begin{tabular}{@{}l c c c c c c c c@{}}
\toprule
Statistic & Published & 95\% CI (published / reconstructed) & Published, INV & Forecast [95\%] & In CI / in envelope & $\sigma_{\text{fold}}$ & $\sigma_{\text{within}}$ & SS share, fold / within \\
\midrule
Median OS (mo) & $\crcOsMedPub$ & ($\crcOsMedPubLo$--$\crcOsMedPubHi$) & --- & $\crcOsMedPred$ [$\crcOsMedLo$, $\crcOsMedHi$] & \crcOsMedPredInPub{} / \crcOsMedPubInPred{} & $\crcOsMedSigFold$ & $\crcOsMedSigSrc$ & $\crcOsMedShareFold\%$ / $\crcOsMedShareSrc\%$ \\
$S_{\textsc{os}}(6\text{ mo})$ & $\crcOsSixPub$ & ($\crcOsSixPubLo$--$\crcOsSixPubHi$) & --- & $\crcOsSixPred$ [$\crcOsSixLo$, $\crcOsSixHi$] & \crcOsSixPredInPub{} / \crcOsSixPubInPred{} & $\crcOsSixSigFold$ & $\crcOsSixSigSrc$ & $\crcOsSixShareFold\%$ / $\crcOsSixShareSrc\%$ \\
$S_{\textsc{os}}(12\text{ mo})$ & $\crcOsTwelvePub$ & ($\crcOsTwelvePubLo$--$\crcOsTwelvePubHi$) & --- & $\crcOsTwelvePred$ [$\crcOsTwelveLo$, $\crcOsTwelveHi$] & \crcOsTwelvePredInPub{} / \crcOsTwelvePubInPred{} & $\crcOsTwelveSigFold$ & $\crcOsTwelveSigSrc$ & $\crcOsTwelveShareFold\%$ / $\crcOsTwelveShareSrc\%$ \\
$S_{\textsc{os}}(18\text{ mo})$ & $\crcOsEighteenPub$ & ($\crcOsEighteenPubLo$--$\crcOsEighteenPubHi$) & --- & $\crcOsEighteenPred$ [$\crcOsEighteenLo$, $\crcOsEighteenHi$] & \crcOsEighteenPredInPub{} / \crcOsEighteenPubInPred{} & $\crcOsEighteenSigFold$ & $\crcOsEighteenSigSrc$ & $\crcOsEighteenShareFold\%$ / $\crcOsEighteenShareSrc\%$ \\
$S_{\textsc{os}}(24\text{ mo})$ & $\crcOsTwentyFourPub$ & ($\crcOsTwentyFourPubLo$--$\crcOsTwentyFourPubHi$) & --- & $\crcOsTwentyFourPred$ [$\crcOsTwentyFourLo$, $\crcOsTwentyFourHi$] & \crcOsTwentyFourPredInPub{} / \crcOsTwentyFourPubInPred{} & $\crcOsTwentyFourSigFold$ & $\crcOsTwentyFourSigSrc$ & $\crcOsTwentyFourShareFold\%$ / $\crcOsTwentyFourShareSrc\%$ \\
$S_{\textsc{os}}(30\text{ mo})$ & $\crcOsThirtyPub$ & ($\crcOsThirtyPubLo$--$\crcOsThirtyPubHi$) & --- & $\crcOsThirtyPred$ [$\crcOsThirtyLo$, $\crcOsThirtyHi$] & \crcOsThirtyPredInPub{} / \crcOsThirtyPubInPred{} & $\crcOsThirtySigFold$ & $\crcOsThirtySigSrc$ & $\crcOsThirtyShareFold\%$ / $\crcOsThirtyShareSrc\%$ \\
\addlinespace
Median PFS (mo) & $\crcPfsMedPub$ & ($\crcPfsMedPubLo$--$\crcPfsMedPubHi$) & $7.0$ & $\crcPfsMedPred$ [$\crcPfsMedLo$, $\crcPfsMedHi$] & \crcPfsMedPredInPub{} / \crcPfsMedPubInPred{} & $\crcPfsMedSigFold$ & $\crcPfsMedSigSrc$ & $\crcPfsMedShareFold\%$ / $\crcPfsMedShareSrc\%$ \\
$S_{\textsc{pfs}}(6\text{ mo})$ & $\crcPfsSixPub$ & ($\crcPfsSixPubLo$--$\crcPfsSixPubHi$) & $0.575$ & $\crcPfsSixPred$ [$\crcPfsSixLo$, $\crcPfsSixHi$] & \crcPfsSixPredInPub{} / \crcPfsSixPubInPred{} & $\crcPfsSixSigFold$ & $\crcPfsSixSigSrc$ & $\crcPfsSixShareFold\%$ / $\crcPfsSixShareSrc\%$ \\
$S_{\textsc{pfs}}(12\text{ mo})$ & $\crcPfsTwelvePub$ & ($\crcPfsTwelvePubLo$--$\crcPfsTwelvePubHi$) & $0.263$ & $\crcPfsTwelvePred$ [$\crcPfsTwelveLo$, $\crcPfsTwelveHi$] & \crcPfsTwelvePredInPub{} / \crcPfsTwelvePubInPred{} & $\crcPfsTwelveSigFold$ & $\crcPfsTwelveSigSrc$ & $\crcPfsTwelveShareFold\%$ / $\crcPfsTwelveShareSrc\%$ \\
$S_{\textsc{pfs}}(18\text{ mo})$ & $\crcPfsEighteenPub$ & ($\crcPfsEighteenPubLo$--$\crcPfsEighteenPubHi$) & $0.116$ & $\crcPfsEighteenPred$ [$\crcPfsEighteenLo$, $\crcPfsEighteenHi$] & \crcPfsEighteenPredInPub{} / \crcPfsEighteenPubInPred{} & $\crcPfsEighteenSigFold$ & $\crcPfsEighteenSigSrc$ & $\crcPfsEighteenShareFold\%$ / $\crcPfsEighteenShareSrc\%$ \\
$S_{\textsc{pfs}}(24\text{ mo})$ & $\crcPfsTwentyFourPub$ & ($\crcPfsTwentyFourPubLo$--$\crcPfsTwentyFourPubHi$) & --- & $\crcPfsTwentyFourPred$ [$\crcPfsTwentyFourLo$, $\crcPfsTwentyFourHi$] & \crcPfsTwentyFourPredInPub{} / \crcPfsTwentyFourPubInPred{} & $\crcPfsTwentyFourSigFold$ & $\crcPfsTwentyFourSigSrc$ & $\crcPfsTwentyFourShareFold\%$ / $\crcPfsTwentyFourShareSrc\%$ \\
\bottomrule
\end{tabular}}
\caption{Forecast of the BREAKWATER control arm against its published readout~\cite{elez2025breakwater} (data cut-off 6~Jan~2025). Published medians with their confidence intervals, and the 12- and 24-month OS survivals, are printed in the paper; the remaining published landmarks are read off the digitized curves. Every landmark interval is reconstructed: it is the pointwise $95\%$ interval of the Kaplan--Meier estimate, with Greenwood variance~\cite{greenwood1926natural} on the log--log scale, on a pseudo-IPD reconstruction of the published curve and its at-risk rows~\cite{guyot2012reconstruction}. The forecast column is the mean and $(q_{0.025},q_{0.975})$ envelope across the $F\times B=\crcNumFolds\times\crcNumReps$ cells of \S\ref{sec:crc-error-quant} under the common BREAKWATER weights. ``In CI / in envelope'' scores each statistic both ways: whether the forecast mean lies inside the published $95\%$ CI, and whether the published value lies inside the forecast envelope. $\sigma_{\text{fold}}$ and $\sigma_{\text{within}}$ are the square-rooted variance components of the nested decomposition --- between fold-models, and within fold-models (calibration-source bootstrap plus MH noise, which the design does not separate); $\sigma_{\text{fold}}$ is not truncated at zero. The last column is each component's share of the total sum of squares (not of $\sigma^2$); the two shares sum to $100\%$. PFS is scored on the investigator-assessed basis of the source trials against the trial's central-review primary. The ``Published, INV'' column gives the trial's investigator-assessed (INV) PFS (EMA assessment report~\cite{ema2026braftovi}, Table~26; $160$ events, median $95\%$ CI $6.0$--$8.3$~months) as a sensitivity comparison: its median differs immaterially from the central-review one, but its landmarks lie lower.}
\label{tab:crc-forecast}
\end{table}

\begin{figure}[H]
\centering
\includegraphics[width=1.0\linewidth]{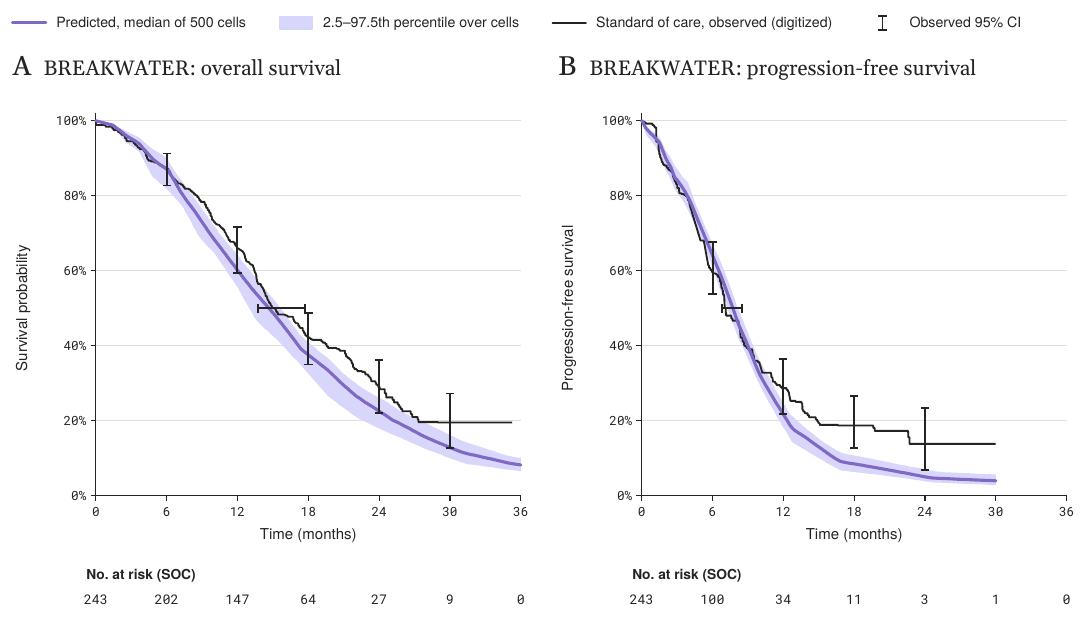}
\caption{Forecast-versus-readout overlay for the BREAKWATER control arm: (A) overall survival, (B) progression-free survival. Black: published Kaplan--Meier curves, digitized from Elez et al.~\cite{elez2025breakwater}, with the published numbers at risk below each panel. Black error bars: the readout's $95\%$ intervals, vertical for the survival at its 6-month landmarks (OS 6--30, PFS 6--24~months) and horizontal, at $50\%$, for its median, as in Table~\ref{tab:crc-forecast}. Purple: the \textsc{FRESH} forecast under the BREAKWATER weights, pointwise median (line) and pointwise $(q_{0.025},q_{0.975})$ envelope (band) across the $\crcNumCells$ fold-model~$\times$~replicate cells of \S\ref{sec:crc-error-quant}. Table~\ref{tab:crc-forecast} reports cell means rather than medians. The pointwise-median and pointwise-mean curves differ by at most $\crcCurveSurvDiffMax$ in survival probability. The median curves cross $0.5$ at $\crcCurveOsCross$ (OS) and $\crcCurvePfsCross$~months (PFS), against mean cell medians of $\crcOsMedPred$ and $\crcPfsMedPred$~months. Panel B is drawn to 30 months, where the plotted published PFS curve is truncated.}
\label{fig:crc-km}
\end{figure}

\paragraph{Where the uncertainty lives.} The variance decomposition in Table~\ref{tab:crc-forecast} shows why fold-model uncertainty has to be carried into the interval. On the OS median, $\crcOsMedShareFold\%$ of the across-cell variance is between fold-models and $\crcOsMedShareSrc\%$ within them (bootstrap replicates of the calibration sources plus sampler noise); on the PFS median the split is $\crcPfsMedShareFold\%$/$\crcPfsMedShareSrc\%$, and the five fold-models' OS medians span $\crcOsMedFoldMin$--$\crcOsMedFoldMax$~months. The fold-model term carries most of the variance of every OS statistic ($\crcOsShareFoldMin$--$\crcOsShareFoldMax\%$) and between $\crcPfsShareFoldMin\%$ and $\crcPfsShareFoldMax\%$ of that of the PFS statistics. It is larger than the source term for OS because few patients in the real-world training corpus match the narrow BREAKWATER eligibility criteria: the model-fit prognostic effect of a \textit{BRAF}~V600E mutation is more uncertain than the thousand-patient historical-trial KM curves used as calibration targets.

\subsection{Case study notes}
\label{sec:crc-notes}
\begin{itemize}
\item \textbf{Retrospective, outcome-blinded validation.} BREAKWATER's outcome data enter nowhere in the chain; only its eligibility criteria and baseline Table~1 do. The comparison in Table~\ref{tab:crc-forecast} is therefore an out-of-sample test of the composed calibrated population, but a retrospective one: the readout was public while the analysis was built. The readout's OS median falls inside the forecast interval; its PFS median falls $\crcPfsMedLoMinusPub$~months below it, while the forecast PFS median lies inside the trial's interval. Scored both ways, the forecast mean lies inside the published interval for $\crcNumPredInPub$ of $\crcNumBenchmarks$ statistics (OS $\crcOsNumPredInPub$, PFS $\crcPfsNumPredInPub$) and the published value inside the forecast envelope for $\crcNumPubInPred$ (OS $\crcOsNumPubInPred$, PFS $\crcPfsNumPubInPred$).
\item \textbf{Composition of several sources.} The forecast composes three independently calibrated surrogates $\Qstar_u\otimes\Kstar_u$ with a fourth Stage-1 projection onto the target's baseline constraints. This is the composability property of the two-stage design (\S\ref{sec:composition}) used in earnest: each $\Kstar_u(\cdot\mid x)$ is a per-particle conditional distribution whose multipliers were fixed under $\Qstar_u$ and which is then evaluated under a different $X$-marginal. Joint \textsc{FRESH} (\S\ref{sec:two-vs-one}) would also give a transportable per-particle kernel, since the conditional of the joint $I$-projection is again a per-$x$ tilt of $\Kref$; the difference is that its multipliers would be fitted under a baseline that the outcome targets had themselves tilted away from the published one, whereas the two-stage multipliers are fitted under $\Qstar_u$, the cohort's published baseline.
\item \textbf{The target is a subset of the sources.} Applying the target's eligibility \emph{after} calibration keeps every source-trial constraint attached to the population in which it was observed. Therefore the outcome forecast inside the target subset relative to the unselected population is determined by the fitted model's subgroup prognosis and the landmark tilt, not by source-trial subgroup data.
\item \textbf{Subsequent therapy.} \textsc{FRESH} does not adjust for treatment after progression. In BREAKWATER's control arm, $100$ of the $139$ patients who received any later systemic therapy received \textit{BRAF}-inhibitor-based treatment~\cite{elez2025breakwater}, which the source trials, run before \textit{BRAF}-targeted second-line therapy existed, could not reflect. OS, unlike PFS, is therefore confounded by subsequent therapy, which plausibly raises later OS in the target arm relative to the forecast, as from 12 months on in Table~\ref{tab:crc-forecast}. Any residual under-separation of \textit{BRAF}~V600E prognosis by the fold-models (\S\ref{sec:crc-model}) would push the forecast tail the other way, upward, so the two effects can partly offset and this comparison cannot separate them.
\end{itemize}

\section{Method discussion}
\label{sec:discussion}

\subsection{The general variational form}
\label{sec:fidelity}

Both stages of \textsc{FRESH} instantiate a single variational template: find a probability measure $\mu^\ast$ that is close to a reference $\mu_0$ while satisfying aggregate constraints, with the constraints partitioned into a hard block $\mathcal{H}$ matched exactly and a soft block $\mathcal{S}$ matched approximately under a convex penalty $R_j$,
\begin{equation}
\mu^\ast \;=\; \argmin_{\mu}\;\underbrace{\mathcal{D}(\mu\,\|\,\mu_0)}_{\text{fidelity}} \;+\; \underbrace{\sum_{j\in\mathcal{S}} R_j\!\bigl(\E_\mu[h_j] - c_j\bigr)}_{\text{soft constraint penalty}}\quad\text{s.t.}\quad\underbrace{\E_\mu[h_j] = c_j,\,j\in\mathcal{H}}_{\text{hard constraints}}.
\label{eq:fidelity-template}
\end{equation}
The fidelity functional $\mathcal{D}$ measures departure of $\mu$ from $\mu_0$; the soft-constraint penalty $R_j$ is a convex (possibly non-smooth) function of the residual (we use the Tikhonov form $R_j(\cdot) = \tfrac{1}{2}\rho_j(\cdot)^2$ throughout the paper, but other choices --- e.g.\ Huber, $\ell_1$ --- are also sensible; an $\ell_1$ penalty makes the corresponding multipliers box-constrained in the dual). The choice of $\mathcal{D}$ has substantive consequences both for the form of $\mu^\ast$ and for the operational interpretability of the multipliers conjugate to the constraints; this subsection records the considerations that led us to KL.

Many alternative choices for $\mathcal{D}$ are available, each yielding different forms for $\mu^\ast$, different parameterizations of the conjugate multipliers, and different sampling algorithms: reverse KL, the divergence behind empirical likelihood~\cite{owen2001empirical}; $\chi^2$ and the rest of the $f$-divergence family (Hellinger, total variation), with $\chi^2$ behind the linear calibration estimator~\cite{deville1992calibration} and minimum-variance balancing weights~\cite{zubizarreta2015stable}; Wasserstein-$p$, whose $p=2$ geometry underlies the gradient-flow formulations of~\cite{jordan1998variational,otto2001convergence}; cross-entropy alone~\cite{cover2006elements}; and, outside the divergence framework, Approximate-Bayesian-Computation surrogates~\cite{marin2012approximate}. In moment-condition estimation the KL, reverse-KL and $\chi^2$ choices correspond to exponential tilting, empirical likelihood and the continuous-updating estimator, all members of the generalized-empirical-likelihood family~\cite{newey2004higher} that offers an alternative to GMM~\cite{hansen1982gmm}. We adopt KL throughout. The decisive consideration is the interpretability of the calibrated kernel and its multipliers. With KL as the fidelity functional, the optimum is the per-$x$ tilt~\eqref{eq:expt-kernel}, $\Kstar(y\mid x)\propto\Kref(y\mid x)\exp\bigl(\sum_j\lambda_j^\ast f_j(x,y)\bigr)$: the exponential family generated by $\{f_j\}$ over the reference $\Kref$, with the constraint targets $c_j$ as expectation parameters. Within the $f$-divergences $\mathcal{D}_\psi(\mu\|\mu_0)=\int \psi(d\mu/d\mu_0)\,d\mu_0$, with convex generator $\psi$ (we write $\psi$ to keep $f_j$ for the constraint statistics), KL is the unique member (up to multiplicative scaling) that produces this canonical-form parameterization~\cite{csiszar1991why}: stationarity of the Lagrangian gives $d\mu^\ast/d\mu_0=(\psi')^{-1}\bigl(\eta+\sum_j\lambda_j h_j\bigr)$, with $\eta$ the multiplier of the normalization, and this factors as a normalizer times $\exp(\sum_j\lambda_j h_j)$ for every choice of constraint statistics only if $(\psi')^{-1}$ is exponential, i.e.\ $\psi(t)$ is a positive multiple of $t\log t$ up to terms affine in $t$, which leave $\mathcal{D}_\psi$ unchanged on probability measures. Under KL, $\lambda^\ast$ are the natural parameters and the $\E[f_j]$ the mean parameters; the Stage-2 dual $\Lambda(\lambda)$ of~\eqref{eq:dual-stage2}, built from the log-partition function $\log Z(\lambda\mid x)$ of~\eqref{eq:partition-fn}, is convex in $\lambda$, and the optimal-value function $V(c)$ of \S\ref{sec:multipliers} is its negative minimum, $V(c)=-\min_\lambda\Lambda(\lambda)$ (Appendix~\ref{app:proofs}). Three properties follow directly from this structure.

\begin{enumerate}
\item \emph{Decision-relevant multiplier magnitudes.} Each $\lambda_j^\ast=\partial V/\partial c_j$ is the marginal cost (in fidelity) of moving the constraint target $c_j$ by a unit, to be read with the cautions of \S\ref{sec:multipliers}: the multipliers are determined jointly, and their magnitudes are comparable only after standardization or through the implied log-tilt curve. With these provisos, stakeholders can read the multiplier vector as a contribution-by-constraint audit of the calibrated prediction, which is otherwise an opaque function of a black-box simulator.

\item \emph{Additive composition of constraints.} The exponent is linear in $(\lambda,f)$, so adding a new aggregate target appends a term to the exponent and a multiplier to the parameter vector without changing the form of the solution. The multipliers must, however, be re-solved jointly: adding a constraint generally moves every multiplier, and each keeps its meaning as $\partial V/\partial c_j$ for the enlarged constraint set. Soft and hard regimes connect smoothly: by~\eqref{eq:soft-multiplier} the residual of a soft constraint is $-\lambda_j^\ast/\rho_j$, so letting $\rho_j\to\infty$ recovers the hard constraint whenever the hard problem has finite multipliers.

\item \emph{Normalizer-free sampling.} Under KL the density ratio $d\Kstar/d\Kref=\exp(\sum_j\lambda_j f_j(x,y))/Z(\lambda\mid x)$ depends on $y$ only through $\sum_j\lambda_j f_j(x,y)$, and its per-$x$ normalizer enters multiplicatively. In a Metropolis--Hastings chain at fixed $x$ that proposes from $\Kref$, both the reference density and $Z(\lambda\mid x)$ therefore cancel (Remark~\ref{rem:cancellation}), and only the $f_j$ are evaluated. The optima under other $f$-divergences, $\Kref(y\mid x)$ times $(\psi')^{-1}\bigl(\eta(x)+\sum_j\lambda_j f_j(x,y)\bigr)$, are also reweightings of $\Kref$ by a function of the $f_j$, so the reference density cancels for them too; but their per-$x$ normalizer $\eta(x)$ enters non-multiplicatively, does not cancel, and must be estimated for each particle from $\Kref$ draws. Wasserstein projections move mass rather than reweighting $\Kref$, so no such cancellation is available for them.
\end{enumerate}

These three properties --- direct readability of $\lambda^\ast$, additive composability, and normalizer-free black-box sampling --- are what we mean by KL's interpretability: fidelity and calibrated kernel are parameterized by the same exponential family, with multipliers that decision-makers can audit.

\subsection{Selecting computational realizations of the \texorpdfstring{$I$}{I}-projections}
\label{sec:realization}
Stage~1 and Stage~2 are both instances of the same $I$-projection problem~\cite{csiszar1975idivergence}, and each admits two computational realizations, both of which solve the dual in the multipliers: a deterministic solve by convex optimization, or a stochastic one, in which Robbins--Monro updates (\S\ref{sec:rm}) are driven by a Metropolis--Hastings sampler of the projected exponential tilt. The choice between them is dictated in part by the interfaces the underlying generative model exposes, and in part by downstream use cases.

The deterministic route requires the dual gradient to be computable. For Stage~1 on the empirical reference $\widehat Q_{0,N}=\tfrac{1}{N}\sum_i\delta_{x_i}$, the dual~\eqref{eq:dual-general} with $\mu_0=\widehat Q_{0,N}$ and statistics $\phi_\ell$ is $\Lambda_X(\nu)=\log\bigl(\tfrac{1}{N}\sum_i e^{\nu^\top\phi(x_i)}\bigr)-\sum_\ell\nu_\ell b_\ell+\sum_{\ell\in\mathcal{S}_X}\nu_\ell^2/(2\rho_\ell)$, with gradient
\[
\frac{\partial\Lambda_X}{\partial\nu_\ell}\;=\;\frac{1}{N}\sum_{i=1}^N w_i(\nu)\,\phi_\ell(x_i)\;-\;b_\ell\;+\;\ind{\ell\in\mathcal{S}_X}\,\frac{\nu_\ell}{\rho_\ell},\qquad w_i(\nu)\propto e^{\nu^\top\phi(x_i)},\ \textstyle\sum_i w_i(\nu)=N,
\]
a weighted average over the $N$ particles whose zero recovers the hard and Tikhonov conditions of~\eqref{eq:stage1-weights}. With the Tikhonov penalty this dual is unconstrained, smooth, and convex, and the cost of solving it is small against the one-time particle-draw cost; we solve Stage~1 with the interior-point conic solver Clarabel. The stochastic route requires only the ability to sample from the reference (and, in our case, to evaluate the constraint functions $f_j$ at the proposed draw); neither the reference density nor the normalizer need be evaluated, by Remark~\ref{rem:cancellation}. Cost is dominated by the chain length needed for convergence.

In both case studies, the marginal $\Qref$ admits direct sampling and the eligibility-filtered draws form an empirical reference cohort large enough to richly interpolate the trial populations of each study but small enough to solve Stage~1 directly; therefore we compute Stage~1 by convex optimization. The upshot of this choice is that we obtain a synthetic cohort of \emph{weighted} particles.  Some use cases may find incorporation of weights conceptually cumbersome. For example, if one wants a database of realistic synthetic patient records, attaching a weight to each is undesirable.  One might opt to use the MH sampler to generate baseline draws if that were the case, and the model supported it.

The choice we make to use MH sampling in Stage~2 is partially motivated by a similar desire not to generate synthetic data with weights compounded from two reweighting procedures --- the chain returns unweighted draws from $\Kstar(\cdot\mid x_i)$ for each particle --- and partially by portability: the sampler needs only draws from $\Kref$, so it carries over unchanged to simulators that expose no density. In our case the conditional kernel $\Kref(y\mid x)$ is the output of an autoregressive transformer. As noted in \S\ref{sec:mh}, both case studies tabulate $\Kref(\cdot\mid x_i)$ once per particle on a time grid, at the scale of the Stage~1 dataset, which makes Stage~2 relatively efficient on modern hardware.

\subsection{Two-stage versus joint composition}
\label{sec:composition}
\textsc{FRESH} as described composes two $I$-projections: Stage~1 projects $\Qref$ onto $\mathcal{C}_Q$ to produce $\Qstar$, and Stage~2 projects $\Kref$ onto the outcome-constraint set at fixed $\Qstar$ to produce $\Kstar$. An evident alternative --- which we call \emph{joint \textsc{FRESH}} --- replaces this composition with a single $I$-projection of the generative joint $\Qref\otimes\Kref$ onto the lifted set of joint distributions of \S\ref{sec:two-vs-one},
\[
\mathcal{C}_J\;=\;\bigl\{P:\ P_X(\X_{\textsc{ie}})=1,\ \E_P[\phi_\ell(X)]=b_\ell\ \forall \ell\in\mathcal{H}_X,\ \E_P[f_j(X,Y)]=c_j\ \forall j\in\mathcal{H}_Y\bigr\},
\]
where $\mathcal{H}_X$ and $\mathcal{H}_Y$ are the hard baseline and outcome blocks, with the soft constraints of both stages added to the objective as penalties. Both procedures yield calibrated joint distributions that satisfy all hard constraints, and both are grounded in Csisz\'ar's $I$-projection theory~\cite{csiszar1975idivergence}. Neither is a priori correct or incorrect; they return different objects and satisfy different guarantees. We adopt the two-stage composition as the \textsc{FRESH} default for its simple transferability, while noting that joint \textsc{FRESH} followed by a downstream propensity adjustment is a viable alternative workflow.

\subsubsection*{The two procedures return different \texorpdfstring{$X$}{X}-marginals}
By the chain rule of relative entropy,
\begin{equation}
\KL(P\,\|\,\Qref\otimes\Kref)\;=\;\KL(P_X\,\|\,\Qref)\;+\;\E_{P_X}\!\bigl[\KL\bigl(P_{Y\mid X}\,\|\,\Kref(\cdot\mid X)\bigr)\bigr].
\label{eq:joint-decomp}
\end{equation}
The two-stage calibrated $X$-marginal is the Stage-1 output~\eqref{eq:stage1-iproj},
\[
\Qstar \;=\; \argmin_{q\in\mathcal{C}_Q}\,\KL(q\,\|\,\Qref)\;+\;(\text{Stage-1 soft penalty});
\]
the joint procedure's calibrated $X$-marginal is the marginal of the joint optimum,
\[
\Pstar_X(dx) \;=\; \Pstar(dx\times\Y),\qquad \Pstar \;=\; \argmin_{P\in\mathcal{C}_J}\,\KL(P\,\|\,\Qref\otimes\Kref)\;+\;(\text{soft penalties of both stages}).
\]
Both meet the hard baseline constraints; they differ in selection criterion. $\Qstar$ is the closest-to-$\Qref$ member of $\mathcal{C}_Q$. $\Pstar_X$ is the $X$-marginal of whichever joint trade-off in~\eqref{eq:joint-decomp} is cheapest in total $\KL$ --- in particular, joint \textsc{FRESH} may pay $\KL(P_X\|\Qref)$ in order to reduce the conditional term, when an alternative baseline distribution makes outcome targets cheaper to satisfy. Whenever $\mathcal{C}_Q$ leaves shape freedom in the $X$-marginal and shifting that shape can reduce conditional cost, the two selection criteria pick different members of $\mathcal{C}_Q$ and $\Qstar\neq \Pstar_X$; the exact condition follows.

\subsubsection*{When the two procedures coincide}
Take all constraints hard, bounded $f_j$, and targets interior to the attainable moments, so that both procedures have exponential-tilt solutions, and let $\lambda^\ast$ be the two-stage Stage-2 multiplier. The two-stage output is
\[
\Qstar\otimes\Kstar\;\propto\;\ind{x\in\X_{\textsc{ie}}}\,\exp\bigl(\nu^{\ast\top}\phi(x)-\log Z(\lambda^\ast\mid x)+\lambda^{\ast\top}f(x,y)\bigr)\;\Qref\otimes\Kref,
\]
while the joint optimum has the form $\Pstar\propto\ind{x\in\X_{\textsc{ie}}}\exp\bigl(\nu'^\top\phi(x)+\lambda'^\top f(x,y)\bigr)\,\Qref\otimes\Kref$, whose conditional is again the per-$x$ tilt $\Kref\,e^{\lambda'^\top f}/Z(\lambda'\mid x)$ and with $X$-marginal proportional to $\ind{x\in\X_{\textsc{ie}}}\,e^{\nu'^\top\phi(x)}\,Z(\lambda'\mid x)\,\Qref$. The procedures therefore coincide if and only if
\[
\log Z(\lambda^\ast\mid x)\;\in\;\operatorname{span}\{1,\ \phi_\ell:\ \ell\in\mathcal{H}_X\}\quad\text{on }\X_{\textsc{ie}}\ (\Qref\text{-a.e.}).
\]
If it does, $\Qstar\otimes\Kstar$ lies in the joint exponential family and meets every constraint in $\mathcal{C}_J$, so it is the joint $I$-projection. Conversely, if the two coincide, equality of the conditionals forces $\lambda'=\lambda^\ast$ (when no non-trivial combination of the $f_j$ is $\Kref(\cdot\mid x)$-a.s.\ constant), and equality of the $X$-marginals then forces $\log Z(\lambda^\ast\mid x)=(\nu^\ast-\nu')^\top\phi(x)+\text{const}$. Two special cases satisfy the condition automatically: a reference kernel $\Kref$ that does not depend on $x$, for which $Z$ is constant; and a $\mathcal{C}_Q$ that pins $P_X$ exactly, e.g.\ a single moment constraint on a binary $X$, since $\operatorname{span}\{1,x\}$ contains every function on $\{0,1\}$. Generic outcome targets under an $x$-dependent $\Kref$ satisfy neither, so in \textsc{FRESH} the two procedures generally differ.

\subsubsection*{Two-stage benefit: composability across analyses}
Many decision-relevant uses of \textsc{FRESH} involve calibrating to several aggregate sources and combining or comparing the resulting patient-level distributions: indirect treatment comparisons, biomarker-stratum analyses, scenario evaluation across regimens. A natural basis for such analyses is a common patient-level cohort with a shared $X$-marginal. In the two-stage composition $\Qstar$ depends only on $(\X_{\textsc{ie}},\{(\phi_\ell,b_\ell)\})$, so two sources reporting the same eligibility and baseline statistics yield the same $\Qstar$, and their calibrated distributions then differ only in $\Kstar$. When the sources' baselines differ, a common cohort must be built by a downstream balancing step. In the pancreatic case study the regimen is itself a hard Stage-1 coordinate (all \textsc{GN} in MPACT, all \textsc{FX} in PRODIGE-4), so the two $\Qstar$ differ, and the \textsc{FX}-versus-\textsc{GN} contrast rests on the model-based standardization of \S\ref{sec:case-study-eb}, which matches only PRODIGE-4's reported moments and needs its own balance and overlap diagnostics. What the two-stage design guarantees is that such steps are well posed: each $\Kstar_u(\cdot\mid x)$ is a per-$x$ tilt that can be evaluated under any $X$-marginal. The colorectal case study (\S\ref{sec:case-study-crc}) relies on this: three kernels $\Kstar_u$, each calibrated under its own source trial's $\Qstar_u$, are re-composed with a fourth baseline surrogate --- the target trial's --- without recalibration (\S\ref{sec:crc-composition}). The conditional of the joint $I$-projection is also a per-$x$ tilt $\Kref\,e^{\lambda^\top f}/Z(\lambda\mid x)$ and is equally transportable. The differences are that its multipliers are fitted under the $Z$-tilted marginal $\Pstar_X$ rather than under the published baseline, and that each source's $\Pstar_X$ depends on its own outcome targets, so two regimens calibrated independently will yield different $X$-marginals even when their reported baseline statistics agree exactly.

\paragraph{An alternative workflow with joint \textsc{FRESH}.}
The composability advantage has a viable alternative for users who prefer joint \textsc{FRESH}: simulate two trials independently as joint $I$-projections, then bring them onto a common reference cohort downstream by propensity-score methods on the simulated patient-level samples. Because $\Pstar_X$ matches the source's reported \emph{baseline summary statistics} but not necessarily the joint shape of all measured covariates, two such simulations can yield differently shaped $X$-marginals on covariates outside $\{\phi_\ell\}$ even when the constraint statistics agree exactly.  This suggests that in this setting a practitioner ought to apply propensity tooling (matching~\cite{rosenbaum1983propensity}, inverse-probability weighting~\cite{hirano2003efficient}, overlap weights~\cite{li2018balancing}, etc.) to ensure critical covariate balance. This workflow is in some respects more flexible than the two-stage composition, since it defers the choice of reference cohort to the analysis stage rather than committing to it at calibration time. On the other hand the two-stage approach delivers composability built in, because sources with the same eligibility and baseline statistics share $\Qstar$ by construction, while the joint-plus-propensity workflow requires a separately-justified propensity model and overlap diagnostics on each pair of simulations compared --- as does the two-stage workflow whenever its sources' baselines differ.

\subsection{Relation to existing methods}
\label{sec:maic-comparison}

The head-to-head simulation enabled by \textsc{FRESH} sits closest to population-adjusted indirect comparison, which combines IPD from some trials with aggregate summaries of others~\cite{phillippo2016tsd,phillippo2018maic}. MAIC reweights the IPD of one trial so that its covariate summaries match those published for the comparator trial, and compares outcomes in the reweighted IPD with the comparator's published outcomes~\cite{signorovitch2010maic}; its method-of-moments weights are entropy-balancing weights~\cite{phillippo2020equivalence}. STC is instead an outcome-regression method: it fits a regression on the IPD and uses it to predict outcomes at the comparator trial's reported covariate values~\cite{caro2010stc,phillippo2018maic}. ML-NMR fits a Bayesian regression that combines IPD from one or more trials with aggregate data from others, integrating the individual-level model over each aggregate trial's covariate distribution~\cite{phillippo2020mlnmr}, and has since been extended to general likelihoods, including flexible parametric survival models~\cite{phillippo2026mlnmr}; parametric G-computation likewise averages an outcome model fitted on the IPD over covariates simulated for the comparator population~\cite{remiroazocar2022gcomp}. With all-hard constraints, \textsc{FRESH}'s Stage~1 is entropy balancing, and hence MAIC weighting, applied to synthetic particles rather than to trial IPD. The case-study cross-trial step (\S\ref{sec:case-study-eb}) reweights a calibrated cohort whose outcomes are model draws onto PRODIGE-4's published baseline statistics, so it is a model-based standardization, closer in spirit to G-computation than to a MAIC on observed outcomes.

Pairing non-trial patient data, or a model of it, with published aggregates is not itself new, and three further families are close. In \emph{data integration}, internal individual-level data are reweighted or constrained to agree with external summaries: Hellerstein and Imbens impose population moments by empirical-likelihood weighting~\cite{hellerstein1999imposing}; constrained maximum-likelihood and empirical-likelihood methods calibrate regression models to external summary-level information, including from populations with different covariate distributions~\cite{chatterjee2016constrained,han2019empirical} and, for survival, published subgroup survival probabilities in a Cox model~\cite{huang2016auxiliary}; and Gu et al.\ generate synthetic data from an external model to integrate summary information from heterogeneous populations~\cite{gu2023synthetic}. \emph{External and synthetic control arms} emulate a trial's control arm from real-world cohorts by applying its eligibility criteria and adjusting for baseline differences~\cite{hernan2016target,tan2022ehr}; entropy balancing to a target population's summary statistics transports trial results~\cite{josey2021transport}; and Bayesian borrowing discounts historical control data through power priors or meta-analytic-predictive priors~\cite{ibrahim2000power,neuenschwander2010map}. Finally, \emph{simulator calibration} by minimum relative entropy is established in other fields: tilting model-based forecast distributions toward moment conditions~\cite{robertson2005relative}, weighted Monte Carlo calibration of asset-pricing models to market prices~\cite{avellaneda2001weighted,glasserman2005weighted}, reweighting molecular-simulation ensembles to experimental averages~\cite{hummer2015bayesian,cesari2018maxent}, and finding the KL-closest adjustment of a pretrained generative model under expectation constraints before fine-tuning it~\cite{khalifa2021distributional,smith2025calibrating}. Against this background \textsc{FRESH} differs in two ways at once: a generative model fit on a separate real-world-data source replaces the trial IPD, and each source's own outcome summaries recalibrate the model's outcome kernel --- a conditional maximum-entropy update at a fixed, separately calibrated baseline~\cite{berger1996maxent}, computed through the model's sampler alone. Patient-level information is still present --- it lives implicitly in the generative model --- but it is one step removed from any specific trial cohort: the model summarizes a separate, broader patient population, and is brought into alignment with each trial's reported aggregates by the two $I$-projections.

The assumptions shift accordingly. MAIC assumes that the reweighted IPD represents the comparator trial's population on every covariate whose imbalance would otherwise bias the comparison, and STC that its outcome regression is correctly specified and holds in the comparator population~\cite{signorovitch2010maic,phillippo2018maic}. For \textsc{FRESH}, sufficient conditions can be stated stage by stage, for recovery of the source pair $Q\otimes K$ from exact targets. Stage~1 with all-hard constraints recovers the source baseline, $\Qstar=Q$, when $Q$ is supported on $\X_{\textsc{ie}}$ and $\log(dQ/d\Qref)$ lies in $\operatorname{span}\{1,\phi_\ell\}$ there: $Q$ then lies in the exponential family~\eqref{eq:stage1-form} and meets the constraints, so it is the $I$-projection. A baseline functional $\E_Q[h(X)]$ is also recovered whenever $h$ itself lies in that span, which is the double robustness of entropy balancing~\cite{zhao2017entropy}. Stage~2 with all-hard outcome constraints recovers the source kernel, $\Kstar=K$ $Q$-a.e., when $\Qstar=Q$ and $K(\cdot\mid x)\propto\Kref(\cdot\mid x)\exp\bigl(\sum_j\lambda_j f_j(x,\cdot)\bigr)$ for some $\lambda$, by the same argument applied to Proposition~\ref{prop:cond-iproj}. Estimands that the constraints do not pin --- subgroup contrasts, distributions of unconstrained covariates, comparisons the published trials never tabulated --- are model-identified: they are only as good as the conditional structure that $\Qref$ and $\Kref$ encode. The trade-off has both directions. The generative model amortizes across many indirect comparisons in the same disease area (one model, many trial calibrations); the synthetic patient-level cohort supports downstream estimands at model-based patient-level granularity, including comparisons the published trials never tabulated.  Conversely, Stage~2's outcome calibration uses landmark-marginal targets, so the conditional kernel $\Kstar$ inherits whatever covariate-dependence structure $\Kref$ encodes --- STC, by fitting outcome regressions directly on real IPD, can identify richer conditional structure when IPD is available --- and misspecification of the conditional structure and of unconstrained covariates propagates through both stages and is not corrected by the calibration. Sensitivity to the generative-model architecture and training corpus is therefore a first-class concern (cf.\ the fold-model checkpoint selection and fine-tuning in \S\ref{sec:crc-model}).

Partially following the reporting recommendations of Technical Support Document~18 of the Decision Support Unit of the UK National Institute for Health and Care Excellence (NICE DSU TSD~18)~\cite{phillippo2016tsd}, adapted to the \textsc{FRESH} setting: weight diagnostics for the per-trial Stage-1 entropy-balancing steps and for the cross-trial entropy-balancing weights of the model-based standardization are in Appendix~\ref{app:weight-diagnostics}; the propagated-uncertainty CI on the cross-trial contrast is in Table~\ref{tab:mpanc-eb-results} with the bootstrap protocol in \S\ref{sec:case-study-error-quant}. The colorectal case study is not an indirect comparison but a single-arm forecast; its analogous reporting --- weight diagnostics for the four entropy-balancing steps and the variance-decomposed interval --- is in Appendix~\ref{app:crc-weight-diagnostics} and Table~\ref{tab:crc-forecast}.

\bibliographystyle{plainnat}
\bibliography{fresh_references}

\appendix
\section{Proofs}
\label{app:proofs}

This appendix proves Proposition~\ref{prop:cond-iproj} and Remark~\ref{rem:cancellation}.

\begin{proof}[Proof of Proposition~\ref{prop:cond-iproj}]
\textit{Hypotheses.} All constraints in this proof are outcome constraints, and we abbreviate $\mathcal{H}=\mathcal{H}_Y$ and $\mathcal{S}=\mathcal{S}_Y$. For a kernel $k$ write $s_j(k)=\E_{X\sim\Qstar,\,Y\sim k(\cdot\mid X)}[f_j(X,Y)]$, so that the attainable set of Proposition~\ref{prop:cond-iproj} is $\mathcal{M}_{\mathcal{H}}=\{(s_j(k))_{j\in\mathcal{H}}:\ \KL(k\,\|\,\Kref;\Qstar)<\infty\}$. We use the hypotheses of the proposition:
(H1) each $f_j$ is bounded and measurable, and we fix $\bar f$ with $|f_j|\le\bar f$ for all $j$;
(H2) $c_{\mathcal{H}}\in\operatorname{ri}\mathcal{M}_{\mathcal{H}}$;
and, for uniqueness of the hard multipliers only,
(H3) no nontrivial linear combination of $(f_j)_{j\in\mathcal{H}}$ is $\Qstar\otimes\Kref$-almost surely a function of $X$ alone.
(H2) excludes boundary targets --- e.g.\ an indicator landmark whose target is $0$ or $1$, or one beyond the support of $\Kref$ --- for which no finite multiplier exists.

\textit{Joint reformulation via the chain rule.} For a kernel $k$ let $\mu(dx,dy) = \Qstar(dx)\,k(dy\mid x)$ and $\mu_0(dx,dy) = \Qstar(dx)\,\Kref(dy\mid x)$. The chain rule of relative entropy~\cite{cover2006elements} gives
\[
\KL(\mu\|\mu_0) \;=\; \KL(\mu_X\|\Qstar) \;+\; \E_{\Qstar}\!\bigl[\KL\bigl(k(\cdot\mid X)\,\big\|\,\Kref(\cdot\mid X)\bigr)\bigr]\;=\;\KL(k\|\Kref;\Qstar),
\]
where the first term vanishes since $\mu_X=\Qstar$ by construction. The outcome moment $\E_\mu[f_j(X,Y)]$ equals $s_j(k)$, so problem~\eqref{eq:stage2-iproj} is equivalent to the constrained joint problem
\begin{equation}
\min_{\mu\in\mathcal{C}}\;\KL(\mu\|\mu_0) \;+\;\tfrac{1}{2}\!\sum_{j\in\mathcal{S}}\!\rho_j\bigl(\E_\mu[f_j(X,Y)]-c_j\bigr)^2,
\qquad
\mathcal{C} = \bigl\{\mu : \mu_X=\Qstar,\ \E_\mu[f_j(X,Y)]=c_j\ \forall j\in\mathcal{H}\bigr\},
\label{eq:stage2-joint}
\end{equation}
on the closed convex set $\mathcal{C}\subset\mathcal P(\X\times\Y)$ (an affine slice intersected with finitely many affine outcome hyperplanes). This is the conditional maximum-entropy problem of Berger et al.~\cite{berger1996maxent}, with their fixed empirical marginal replaced by a general $\Qstar$; entropy minimization over convex sets of this kind, with an infinite family of marginal constraints, is treated in generality by L\'eonard~\cite[Thm.~5.2 and \S6.2]{leonard2008minimization}; we give a direct argument under (H1)--(H3).

\textit{Candidate form.} The Lagrangian of~\eqref{eq:stage2-joint} carries a function-valued multiplier $\eta:\X\to\R$ for the marginal constraint $\mu_X=\Qstar$ and scalar multipliers $\lambda_j$ for the outcome moments. For the soft block, the convex-conjugate identity
\[
\tfrac{\rho_j}{2}(s-c_j)^2 \;=\; \sup_{\lambda_j\in\R}\bigl[\lambda_j(c_j-s) - \lambda_j^2/(2\rho_j)\bigr],\qquad s=\E_\mu[f_j(X,Y)],
\]
with optimizer $\lambda_j = -\rho_j(s-c_j)$ (the sign convention of \S\ref{sec:info-geom}), supplies the soft multipliers. Stationarity in $\mu$ suggests $d\mu^\ast/d\mu_0(x,y)\propto\exp\bigl(\sum_{j\in\mathcal{H}\cup\mathcal{S}}\lambda_j f_j(x,y)+\eta(x)\bigr)$. Integrating in $y$, $\mu^\ast_X(dx) \propto \Qstar(dx)\,e^{\eta(x)}\,Z(\lambda\mid x)$ with $Z(\lambda\mid x)$ as in~\eqref{eq:partition-fn}, and imposing $\mu^\ast_X=\Qstar$ pins $\eta(x) = -\log Z(\lambda\mid x)$. Substituting back, $\mu^\ast = \Qstar\otimes K_\lambda$, where $K_\lambda(dy\mid x)=\Kref(dy\mid x)\,e^{\lambda^\top f(x,y)}/Z(\lambda\mid x)$ is the per-$x$ tilt~\eqref{eq:expt-kernel} at $\lambda$ (the kernel $K_\lambda$ of \S\ref{sec:rm}). We now show that, for a suitable $\lambda^\ast$, $K_{\lambda^\ast}$ is the unique solution.

\textit{The dual and its minimizer.} Let $A(\lambda)=\E_{X\sim\Qstar}[\log Z(\lambda\mid X)]$, so that the dual~\eqref{eq:dual-stage2} reads $\Lambda(\lambda)=A(\lambda)-\sum_j\lambda_j c_j+\sum_{j\in\mathcal{S}}\lambda_j^2/(2\rho_j)$.
By (H1), $e^{-\bar f\|\lambda\|_1}\le Z(\lambda\mid x)\le e^{\bar f\|\lambda\|_1}$, so $\Lambda$ is finite, convex and differentiable, with $\partial\Lambda/\partial\lambda_j=s_j(K_\lambda)-c_j+\ind{j\in\mathcal{S}}\,\lambda_j/\rho_j$. Along a ray $\lambda+t\theta$, $t\to\infty$, the ridge term drives $\Lambda$ to $+\infty$ if $\theta_{\mathcal{S}}\neq0$. If $\theta_{\mathcal{S}}=0$, then $\Lambda(\lambda+t\theta)/t\to\E_{\Qstar}\bigl[\operatorname*{ess\,sup}_{\Kref(\cdot\mid X)}\theta_{\mathcal{H}}^\top f_{\mathcal{H}}(X,\cdot)\bigr]-\theta_{\mathcal{H}}^\top c_{\mathcal{H}}=\sup_{m\in\mathcal{M}_{\mathcal{H}}}\theta_{\mathcal{H}}^\top m-\theta_{\mathcal{H}}^\top c_{\mathcal{H}}$, which is positive by (H2) unless $\theta_{\mathcal{H}}^\top m$ is constant on $\mathcal{M}_{\mathcal{H}}$. In that case $\theta_{\mathcal{H}}^\top f_{\mathcal{H}}(X,Y)$ equals a function $h(X)$ almost surely, $\Lambda(\lambda+t\theta)=\Lambda(\lambda)+t\bigl(\E_{\Qstar}[h]-\theta_{\mathcal{H}}^\top c_{\mathcal{H}}\bigr)=\Lambda(\lambda)$, and $\Lambda$ is constant along the ray. Hence $\Lambda$ is coercive modulo a subspace on which it is constant, and attains its minimum at some $\lambda^\ast$. Its Hessian $\E_{\Qstar}\bigl[\operatorname{Cov}_{K_\lambda(\cdot\mid X)}(f)\bigr]+\operatorname{diag}\bigl(0_{\mathcal{H}},\rho_{\mathcal{S}}^{-1}\bigr)$ is positive definite under (H3), so $\lambda^\ast$ is then unique.

\textit{Verification and uniqueness of $\Kstar$.} Let $\Kstar=K_{\lambda^\ast}$. The stationarity condition $\nabla\Lambda(\lambda^\ast)=0$ reads $s_j(\Kstar)=c_j$ for $j\in\mathcal{H}$ and $\lambda_j^\ast=-\rho_j\bigl(s_j(\Kstar)-c_j\bigr)$ for $j\in\mathcal{S}$; these are the hard and Tikhonov conditions of the proposition, and $\Kstar$ is feasible. Let $\mathcal{J}(k)$ denote the objective of~\eqref{eq:stage2-iproj}, and let $k$ be any feasible kernel with $\mathcal{J}(k)<\infty$. Since $\log(d\Kstar/d\Kref)(y\mid x)=\lambda^{\ast\top}f(x,y)-\log Z(\lambda^\ast\mid x)$,
\[
\KL(k\|\Kref;\Qstar)\;=\;\KL(k\|\Kstar;\Qstar)\;+\;\sum_j\lambda_j^\ast s_j(k)\;-\;A(\lambda^\ast).
\]
Using $s_{\mathcal{H}}(k)=c_{\mathcal{H}}$ and, for $j\in\mathcal{S}$, the conjugate inequality $\lambda_j^\ast s+\tfrac{\rho_j}{2}(s-c_j)^2\ge\lambda_j^\ast c_j-(\lambda_j^\ast)^2/(2\rho_j)$, with equality if and only if $s=c_j-\lambda_j^\ast/\rho_j$, we obtain $\mathcal{J}(k)\ge\KL(k\|\Kstar;\Qstar)-\Lambda(\lambda^\ast)$. At $k=\Kstar$ both the KL term and the slack in the soft terms vanish, so $\mathcal{J}(\Kstar)=-\Lambda(\lambda^\ast)$ and
\[
\mathcal{J}(k)\;\ge\;\mathcal{J}(\Kstar)\;+\;\KL(k\|\Kstar;\Qstar)\qquad\text{for every feasible }k.
\]
Hence $\Kstar$ is a minimizer, and any other minimizer equals $\Kstar(\cdot\mid x)$ for $\Qstar$-a.e.\ $x$; in particular the kernel is unique even when $\lambda^\ast$ is not. The soft multipliers are unique in any case, since $\lambda_j^\ast=-\rho_j\bigl(s_j(\Kstar)-c_j\bigr)$ for $j\in\mathcal{S}$ and $\Kstar$ is unique.

\textit{Value function.} The optimal-value function of \S\ref{sec:multipliers} is
\[
V(c)\;=\;\mathcal{J}(\Kstar)\;=\;-\min_\lambda\Lambda(\lambda)\;=\;\sup_\lambda\Bigl[\sum_j\lambda_j c_j-A(\lambda)-\sum_{j\in\mathcal{S}}\frac{\lambda_j^2}{2\rho_j}\Bigr],
\]
a supremum of affine functions of $c$, hence convex. Where the maximizer is unique (e.g.\ under (H3)), convex conjugacy gives $\partial V/\partial c_j=\lambda_j^\ast$ for every $j$, hard or soft: the envelope relation used in \S\ref{sec:multipliers} and \S\ref{sec:fidelity}.
\end{proof}

\begin{proof}[Proof of Remark~\ref{rem:cancellation}]
The argument uses only that the target~\eqref{eq:target} is a product of per-particle factors; it therefore holds verbatim for any statistics $f_j$ that keep the target normalizable, including the smoothed $f_j^\varepsilon$ of \S\ref{sec:sigmoid}, and for any mask, including the single-particle case $\mathcal{I}=\{i\}$. Write $\kappa(\tilde{\mathbf{y}}\mid\mathbf{y},\mathcal{I})$ for the proposal density, which draws $\tilde y_i\sim\Kref(\cdot\mid x_i)$ independently for $i\in\mathcal{I}$ and sets $\tilde y_i=y_i$ otherwise. The mask $\mathcal{I}$ is drawn from a distribution depending only on the fixed particle weights $\{w_i\}$, not on $\mathbf{y}$ or $\tilde{\mathbf{y}}$, so its probability cancels between numerator and denominator of the MH acceptance ratio and we condition on $\mathcal{I}$ throughout. The remaining acceptance ratio is
\[
\frac{\pi(\tilde{\mathbf{y}}\mid\boldsymbol\lambda)\,\kappa(\mathbf{y}\mid\tilde{\mathbf{y}},\mathcal{I})}{\pi(\mathbf{y}\mid\boldsymbol\lambda)\,\kappa(\tilde{\mathbf{y}}\mid\mathbf{y},\mathcal{I})}.
\]
Using~\eqref{eq:target} and noting $\tilde y_i=y_i$ for $i\notin\mathcal{I}$, the ratio of targets reduces to
\[
\prod_{i\in\mathcal{I}}\frac{\Kref(\tilde y_i\mid x_i)}{\Kref(y_i\mid x_i)}\;\exp\Bigl(\sum_{i\in\mathcal{I}}\sum_j\lambda_j[f_j(x_i,\tilde y_i)-f_j(x_i,y_i)]\Bigr).
\]
The reverse-proposal ratio is
\[
\frac{\kappa(\mathbf{y}\mid\tilde{\mathbf{y}},\mathcal{I})}{\kappa(\tilde{\mathbf{y}}\mid\mathbf{y},\mathcal{I})}\;=\;\prod_{i\in\mathcal{I}}\frac{\Kref(y_i\mid x_i)}{\Kref(\tilde y_i\mid x_i)}.
\]
The product of these two cancels the $\Kref$ factors and leaves $\exp(\sum_{i\in\mathcal{I}}\Delta_i)$, with $\Delta_i$ as in~\eqref{eq:accept}; this is the joint update of the remark. With $\mathcal{I}=\{i\}$ the same computation gives the single-particle acceptance ratio $\exp(\Delta_i)$. Given $\mathcal{I}$, the target~\eqref{eq:target} is a product over particles, and these single-particle updates act on distinct coordinates; accepting or rejecting each $i\in\mathcal{I}$ independently therefore applies a product of MH kernels, each reversible with respect to its own factor of the target, and the product kernel is reversible with respect to~\eqref{eq:target}.
\end{proof}

\section{Case-study details}
\label{app:case-study-details}

\subsection{Trial inclusion / exclusion specifications (pancreatic case study)}
\label{app:ie-specs}

The two trial-specific eligibility sets $\X_{\textsc{ie}}^{\textsc{mpact}}$ and $\X_{\textsc{ie}}^{\textsc{prodige4}}$ that enter Stage~1 of the case study (\S\ref{sec:case-study}) are defined by the eligibility criteria reported in the published reports of Von Hoff~et~al.\ (NCT00844649)~\cite{vonhoff2013mpact} and Conroy~et~al.\ (NCT00112658)~\cite{conroy2011folfirinox}. Table~\ref{tab:mpanc-ie-specs} lists those criteria from each trial that are \emph{evaluable in our simulation} --- specifically, criteria expressible in terms of baseline covariates represented in the real-world clinico-genomic dataset on which the generative model is trained. The two specifications are nearly identical on this evaluable subset; their differences sit on a handful of performance-status and age bounds. Several criteria from each protocol --- e.g.\ RECIST measurability, lab-derived organ function, surgical history, pregnancy status, prior-malignancy history --- are absent from the dataset's covariate schema and are therefore not enforceable by the eligibility filter applied to baseline draws from $\Qref$; they are omitted from the table. Only the experimental-arm regimen constraints are used (\textsc{GN} for MPACT, \textsc{FX} for PRODIGE-4); Table~\ref{tab:mpanc-baseline-targets} lists them.

\begin{table}[H]
\centering\small
\setlength{\tabcolsep}{6pt}
\renewcommand{\arraystretch}{1.15}
\begin{tabular}{@{}l l@{}}
\toprule
Criterion & Threshold / scope \\
\midrule
\multicolumn{2}{l}{\textbf{Inclusion --- common to both trials} (AND-combined)} \\
\quad Indication                                  & \texttt{ONCOCODE}~$=$~PAAD (pancreatic adenocarcinoma) \\
\quad Stage                                       & TNM stage~4 (metastatic) \\
\quad Prior systemic therapy                      & first-line only \\
\quad Metastatic location known                   & not null \\
\midrule
\multicolumn{2}{l}{\textbf{Inclusion --- trial-specific}} \\
\quad ECOG performance status, MPACT$^{\dag}$     & 0--2 \\
\quad ECOG performance status, PRODIGE-4          & 0--1 \\
\quad Age (years), MPACT                          & $\ge 18$ \\
\quad Age (years), PRODIGE-4                      & $18$--$75$ \\
\midrule
\multicolumn{2}{l}{\textbf{Exclusion --- common to both trials} (OR-combined)} \\
\quad Brain metastases                            & excluded \\
\quad Non-adenocarcinoma pancreatic histology     & PANET, PANEC (also PAAC for PRODIGE-4) \\
\bottomrule
\end{tabular}
\caption{Inclusion / exclusion criteria from each trial's published eligibility criteria that are evaluable against the simulator's baseline covariate schema; together they define the eligibility sets $\X_{\textsc{ie}}^{\textsc{mpact}}$ and $\X_{\textsc{ie}}^{\textsc{prodige4}}$ supplied to Stage~1 of \textsc{FRESH}. $^{\dag}$MPACT requires KPS~$\ge 70$~\cite{vonhoff2013mpact}, which maps to ECOG~$\le 2$ under the mapping of Ma et al.~\cite{ma2010interconversion} (ECOG~0~$=$~KPS~100, ECOG~1~$=$~KPS~80--90, ECOG~2~$=$~KPS~60--70) and also under that of Prasad et al.~\cite{prasad2018kps}, which agrees for ECOG~0--1 but assigns KPS~60 to ECOG~3; two MPACT \textsc{GN} patients had a baseline KPS of 60. The published Stage~1 ECOG proportions in Table~\ref{tab:mpanc-baseline-targets} correspond to the mapped trial-reported KPS distribution.}
\label{tab:mpanc-ie-specs}
\end{table}

\subsection{MH sampler hyperparameters (pancreatic case study)}
\label{app:mh-hyperparameters}

Table~\ref{tab:mh-hyperparameters} lists the numerical hyperparameter settings supplied to Stage~2 of Algorithm~\ref{alg:fresh} for the case-study runs of \S\ref{sec:case-study}. The same values are used for both trial arms, and are inherited unchanged when Stage~2 is re-invoked under the warm-started fan-out of \S\ref{sec:case-study-error-quant}.

\begin{table}[H]
\centering\small
\setlength{\tabcolsep}{8pt}
\renewcommand{\arraystretch}{1.15}
\begin{tabular}{@{}l l l@{}}
\toprule
Symbol & Meaning & Value \\
\midrule
$\alpha$         & proposal rate (Bernoulli per-particle inclusion)                          & $10^{-3}$  \\
$\varepsilon$    & sigmoid scale on quantile constraints (Eq.~\eqref{eq:sigmoid-surrogate}) & $10$~days  \\
$\gamma_0$       & Robbins--Monro schedule scale, $\gamma_t=\gamma_0/(t_0+t)^\beta$ (Eq.~\eqref{eq:rm}) & $1.0$ \\
$\beta$          & Robbins--Monro polynomial decay exponent                                  & $0.6$      \\
$t_0$            & Robbins--Monro polynomial schedule offset                                 & $1$        \\
$\delta_{\max}$  & per-iteration clip on each multiplier update                              & $10^{-3}$  \\
$T_b$            & burn-in                                                                   & $1{,}000$  \\
$s$              & post-burn-in thinning interval                                            & $1{,}000$  \\
$n_{\mathrm{keep}}$ & retained post-burn-in draws per particle                               & $10$       \\
$n_{\mathrm{part}}$ & \makecell[l]{data-parallel cohort partitions, each with\\ its own multipliers (across-partition $\widehat R$)} & $400$      \\
---              & stopping rule                                                             & \makecell[l]{$\widehat R<1.05$ on every constraint,\\ $\max_j|\widehat g_j-c_j|\le 0.02$, and at least\\ $30$ post-burn-in checkpoints (every $1{,}000$\\ iterations); $T_{\max}=10^{6}$} \\
---              & chain initialization                                                      & \makecell[l]{cold start (reference chains);\\ warm start from reference tail (fan-out)} \\
\bottomrule
\end{tabular}
\caption{Stage~2 MH sampler hyperparameters used in the case study (\S\ref{sec:case-study}); values shared between the two trial arms and across all bootstrap-replicate fan-out chains.}
\label{tab:mh-hyperparameters}
\end{table}

\subsection{Weight diagnostics for the entropy-balancing steps (pancreatic case study)}
\label{app:weight-diagnostics}

Table~\ref{tab:weight-diagnostics} reports the diagnostics for the three entropy-balancing steps in the pancreatic case study (\S\ref{sec:case-study}) --- the two per-trial Stage-1 entropy-balancing steps and the cross-trial entropy balancing (EB) of the model-based standardization --- and Table~\ref{tab:weight-quantiles} reports the quantile distribution of $w/\bar w$. The effective sample size (ESS) is defined as $(\sum_i w_i)^2/\sum_i w_i^2$~\cite{kish1965survey}.

\begin{table}[H]
\centering\small
\setlength{\tabcolsep}{4pt}
\renewcommand{\arraystretch}{1.2}
\begin{tabular}{@{}l r r r r r r@{}}
\toprule
Step & $N$ & ESS & ESS/$N$ & $w_{\max}/\bar w$ & top-$5\%$ & top-$10\%$ \\
\midrule
Stage~1 MPACT (per-trial baseline)            & $282{,}092$ & $127{,}742$ & $0.45$ & $3.61$ & $17.1\%$ & $32.2\%$ \\
Stage~1 PRODIGE-4 (per-trial baseline)        & $175{,}782$ & $151{,}513$ & $0.86$ & $1.73$ & $\phantom{0}8.6\%$  & $17.1\%$ \\
Cross-trial EB, MPACT$\to$PRODIGE-4             & $171{,}562$ & $144{,}727$ & $0.84$ & $1.91$ & $\phantom{0}9.4\%$  & $17.8\%$ \\
\bottomrule
\end{tabular}
\caption{Weight diagnostics for the three entropy-balancing steps in the pancreatic case study. The right two columns give the share of total weight carried by the top-$5\%$ and top-$10\%$ of particles. Uniform weights would give $\mathrm{ESS}/N=1$, $w_{\max}/\bar w=1$, and a top-$5\%$ ($10\%$) share of $5\%$ ($10\%$). The cross-trial row is computed from an earlier run of the cross-trial step with the same balancing specification and the same $N$, not from the run behind Table~\ref{tab:mpanc-eb-results}; the two weight vectors were not compared.}
\label{tab:weight-diagnostics}
\end{table}

\begin{table}[H]
\centering\small
\setlength{\tabcolsep}{6pt}
\renewcommand{\arraystretch}{1.2}
\begin{tabular}{@{}l r r r r r r r@{}}
\toprule
Step & $q_{01}$ & $q_{05}$ & $q_{25}$ & $q_{50}$ & $q_{75}$ & $q_{95}$ & $q_{99}$ \\
\midrule
Stage~1 MPACT             & $0.001$ & $0.001$ & $0.019$ & $0.643$ & $1.463$ & $3.109$ & $3.610$ \\
Stage~1 PRODIGE-4         & $0.471$ & $0.471$ & $0.570$ & $1.062$ & $1.222$ & $1.715$ & $1.732$ \\
Cross-trial EB            & $0.302$ & $0.302$ & $0.785$ & $0.939$ & $1.275$ & $1.765$ & $1.913$ \\
\bottomrule
\end{tabular}
\caption{Quantiles of the normalized weight $w/\bar w$ for each entropy-balancing step. Uniform weights would have all quantiles at $1.0$.}
\label{tab:weight-quantiles}
\end{table}

The cross-trial EB step --- the one most directly affecting the CIs in Table~\ref{tab:mpanc-eb-results} --- has $\mathrm{ESS}/N=0.84$ and $w_{\max}/\bar w=1.91$; Stage-1 PRODIGE-4 is similarly well-behaved. Because the particles are synthetic draws, $\mathrm{ESS}/N$ here is the relative efficiency of the weighted particle cohort --- a Monte Carlo efficiency --- not a count of patients' worth of information retained. Stage-1 MPACT's lower nominal $\mathrm{ESS}/N=0.45$ is driven by a long \emph{left} tail of near-zero weights ($q_{01}/q_{50}\approx 0.002$) on particles that fail eligibility criteria, not by concentration on a small fraction of high-influence particles ($q_{99}/q_{50}\approx 5.6$, top-$5\%$ share $17\%$).

\subsection{MH sampler convergence diagnostics (pancreatic case study)}
\label{app:mh-diagnostics}

We report adaptive-MCMC diagnostics for the two reference Stage-2 MH chains used in the pancreatic case study. Each chain runs on $n_{\mathrm{part}}=400$ data-parallel partitions of the cohort, each partition with its own multipliers $\lambda$ and running constraint estimates $\widehat g_j$. Acceptance rates, the across-partition Gelman--Rubin $\widehat R$ statistic~\cite{gelman1992inference} on the running constraint estimates $\widehat g_j$, the partition mean and across-partition standard deviation of the multipliers at stop, and achieved-vs-target deviations are summarized in Table~\ref{tab:mh-diagnostics}. Both chains stopped at $T_{\text{stop}}=31{,}000$, the earliest iteration the stopping rule allows ($1{,}000$ burn-in iterations plus $30$ checkpoints every $1{,}000$ iterations).

\begin{table}[H]
\centering\small
\setlength{\tabcolsep}{6pt}
\renewcommand{\arraystretch}{1.2}
\begin{tabular}{@{}l r r@{}}
\toprule
Diagnostic & MPACT (\textsc{GN}) & PRODIGE-4 (\textsc{FX}) \\
\midrule
Acceptance rate at stop                          & $0.70$    & $0.57$    \\
Acceptance rate, across-partition min/max at stop & $0.67/0.74$ & $0.53/0.60$ \\
Max across-partition $\widehat R$ on $\widehat g_j$ across constraints & $0.99$    & $0.99$    \\
Min / max across constraints of partition-mean $\lambda_j$ at stop & $-0.46\,/\,+0.81$ & $-0.32\,/\,+0.98$ \\
Min / max across constraints of across-partition SD of $\lambda_j$ at stop & $0.069\,/\,0.075$ & $0.095\,/\,0.109$ \\
Constraint residual at stop $\max_j|\widehat g_j-c_j|$ & $0.0016$  & $0.0020$  \\
Max achieved-vs-target landmark deviation (days) & $5.27$    & $5.62$    \\
Stopping iteration $T_{\text{stop}}$             & $31{,}000$ & $31{,}000$ \\
\bottomrule
\end{tabular}
\caption{MH-sampler diagnostics for the Stage-2 chains in each trial arm of the pancreatic case study. Acceptance rates are cumulative over the run. Gelman--Rubin $\widehat R$ is computed across the $n_{\mathrm{part}}=400$ data-parallel partitions of the cohort, on the running estimate $\widehat g_j^{(t)}$ of each constraint statistic. The constraint residual is the largest $|\widehat g_j-c_j|$ over constraints at the end of the chain, on the scale of the sigmoid approximation; landmark deviations are residuals between the achieved KM landmark time and the published target. The stopping rule requires $\widehat R<1.05$, a constraint residual of at most $0.02$ and at least $30$ post-burn-in checkpoints (Table~\ref{tab:mh-hyperparameters}); both case-study chains stopped at the earliest iteration this allows, well within $T_{\max}$.}
\label{tab:mh-diagnostics}
\end{table}

The downstream bootstrap uncertainty in Table~\ref{tab:mpanc-eb-results} is computed by re-running the full pipeline across $B_{\textsc{gn}}\times B_{\textsc{fx}}=10{,}000$ replicate pairs (\S\ref{sec:case-study-error-quant}), which captures end-to-end variability.

\subsection{Trial inclusion / exclusion specifications (colorectal case study)}
\label{app:crc-ie-specs}

Table~\ref{tab:crc-ie-specs} lists the eligibility criteria of the two source trials and of BREAKWATER (\S\ref{sec:case-study-crc}) that are evaluable against the simulator's baseline covariate schema, as for the pancreatic case study (Appendix~\ref{app:ie-specs}). The source-trial eligibility sets $\X_{\textsc{ie}}^{\textsc{xel}}=\X_{\textsc{ie}}^{\textsc{fox}}$ (both NO16966) and $\X_{\textsc{ie}}^{\textsc{tri}}$ (TRIBE) are applied within their cohorts before Stage~1; BREAKWATER's set $\X_{\textsc{ie}}^{\textsc{bw}}$ is applied to the pooled calibrated particles before the target-trial entropy-balancing step (\S\ref{sec:crc-composition}). The regimen sets defining each cohort and the BREAKWATER control arm are those of Table~\ref{tab:crc-baseline-targets}. As before, criteria not representable in the covariate schema (measurable disease, organ-function laboratories, prior malignancy, NO16966's exclusion of prior oxaliplatin or bevacizumab~\cite{cassidy2011no16966}, and, for BREAKWATER, the requirement that adjuvant therapy be $\ge 6$ months before relapse) are omitted. The brain-metastasis filter excludes any recorded brain metastasis, which is stricter than BREAKWATER, whose protocol excludes only symptomatic brain metastases~\cite{pfizer2024breakwaterprotocol}.

\begin{table}[H]
\centering\small
\setlength{\tabcolsep}{6pt}
\renewcommand{\arraystretch}{1.15}
\begin{tabular}{@{}l l@{}}
\toprule
Criterion & Threshold / scope \\
\midrule
\multicolumn{2}{l}{\textbf{Inclusion --- common to all three trials} (AND-combined)} \\
\quad Indication                                  & \texttt{ONCOCODE}~$\ni$~COADREAD (colorectal adenocarcinoma) \\
\quad Stage                                       & TNM stage~4 (metastatic) \\
\quad Prior systemic therapy                      & none (first line) \\
\quad Metastatic location known                   & non-empty, not \texttt{NONE} \\
\midrule
\multicolumn{2}{l}{\textbf{Inclusion --- trial-specific}} \\
\quad Age and ECOG, NO16966                       & age $\ge 18$; ECOG 0--1 \\
\quad Age and ECOG, TRIBE                         & (age 18--70 and ECOG 0--2) or (age 71--75 and ECOG 0) \\
\quad Age and ECOG, BREAKWATER                    & age $\ge 16$; ECOG 0--1 \\
\quad \textit{BRAF} genotype, BREAKWATER          & \textit{BRAF}~V600E present \\
\midrule
\multicolumn{2}{l}{\textbf{Exclusion --- common to all three trials} (OR-combined)} \\
\quad Brain metastases                            & any recorded brain metastasis excluded \\
\multicolumn{2}{l}{\textbf{Exclusion --- BREAKWATER only}} \\
\quad Microsatellite status                       & microsatellite-instability-high (MSI-high) excluded \\
\quad \textit{KRAS} / \textit{NRAS}               & any recorded \textit{KRAS} or \textit{NRAS} alteration excluded \\
\bottomrule
\end{tabular}
\caption{Inclusion / exclusion criteria of NO16966~\cite{cassidy2008xelox}, TRIBE~\cite{loupakis2014tribe}, and BREAKWATER~\cite{elez2025breakwater} evaluable against the simulator's baseline covariate schema; together they define the eligibility sets supplied to the colorectal case study of \S\ref{sec:case-study-crc}. TRIBE's ECOG~2 admission below age~70 is folded into ECOG~1 by its own Table-1 constraint (Table~\ref{tab:crc-baseline-targets}). BREAKWATER excludes \textit{RAS}-mutant tumors~\cite{pfizer2024breakwaterprotocol,kopetz2025breakwater,elez2025breakwater}, although the registry eligibility text does not list the criterion; it is implemented here as exclusion of any recorded \textit{KRAS} or \textit{NRAS} alteration.}
\label{tab:crc-ie-specs}
\end{table}

\subsection{MH sampler hyperparameters (colorectal case study)}
\label{app:crc-mh-hyperparameters}
Table~\ref{tab:crc-mh-hyperparameters} lists the Stage-2 settings used in every chain of the colorectal case study --- the three source cohorts, the \crcNumFolds{} fold-models, and the \crcNumReps{} bootstrap replicates per cohort (\crcNumChains{} chains in all), every one cold-started. They differ from the pancreatic settings (Table~\ref{tab:mh-hyperparameters}) in the proposal rate, in the sigmoid scale (with a narrower scale for PFS than for OS), and in running a single chain that stops early once every constraint is met rather than a multi-chain convergence-triggered stop; each chain retains a single post-burn-in draw per particle ($n_{\mathrm{keep}}=1$), so the calibrated outcome of a particle in a given $(r,b)$ cell is one draw from $K^{\ast,(r,b)}_{u}(\cdot\mid x_i)$ and the across-cell spread is what is reported.

\begin{table}[H]
\centering\small
\setlength{\tabcolsep}{8pt}
\renewcommand{\arraystretch}{1.15}
\begin{tabular}{@{}l l l@{}}
\toprule
Symbol & Meaning & Value \\
\midrule
$\alpha$         & proposal rate (Bernoulli per-particle inclusion)                          & $\crcMhAlpha$      \\
$\varepsilon$    & sigmoid scale on quantile / landmark constraints (Eq.~\eqref{eq:sigmoid-surrogate}) & $\crcMhEpsilonOs$~days (OS), $\crcMhEpsilonPfs$~days (PFS) \\
$\gamma_0$       & Robbins--Monro schedule scale, $\gamma_t=\gamma_0/(t_0+t)^\beta$ (Eq.~\eqref{eq:rm}) & $\crcMhGammaZero$ \\
$\beta$          & Robbins--Monro polynomial decay exponent                                  & $\crcMhRmA$      \\
$t_0$            & Robbins--Monro polynomial schedule offset                                 & $\crcMhTZero$      \\
$\delta_{\max}$  & per-iteration clip on each multiplier update                              & $\crcMhDeltaMax$      \\
$T_b$            & burn-in                                                                   & $\crcMhBurnIn$  \\
$s$              & post-burn-in thinning interval                                            & $\crcMhSpacing$  \\
$n_{\mathrm{keep}}$ & retained post-burn-in draws per particle                               & $\crcMhTailSamples$        \\
$n_{\mathrm{part}}$ & data-parallel cohort partitions (multi-chain $\widehat R$)             & n/a (single chain) \\
$\tau_{\mathrm{stop}}$ & early stop: tolerance on every hard deviation $|\widehat P_j-c_j|$ & $\crcMhStopTol$ \\
$n_{\mathrm{stop}}$ & early stop: consecutive checks below $\tau_{\mathrm{stop}}$, one per $\crcMhLogInterval$ iterations & $\crcMhStopChecks$ \\
$T$              & iteration cap (if the early stop does not fire)                           & $\crcMhIters$ \\
---              & chain initialization                                                      & cold start (every chain) \\
\bottomrule
\end{tabular}
\caption{Stage~2 MH sampler hyperparameters used in the colorectal case study (\S\ref{sec:case-study-crc}); values shared across the three source cohorts, the five fold-models, and all bootstrap-replicate chains. The stopping rule watches the unsmoothed deviation $|\widehat P_j-c_j|$ of each hard constraint from the chain's own bootstrap-perturbed target, where $\widehat P_j$ is the weighted fraction of particles with an outcome at or below the constraint's time, rather than the sigmoid-smoothed statistic the multipliers are fitted to.}
\label{tab:crc-mh-hyperparameters}
\end{table}

\subsection{Weight diagnostics for the entropy-balancing steps (colorectal case study)}
\label{app:crc-weight-diagnostics}
Table~\ref{tab:crc-weight-diagnostics} reports the size and effective sample size of the four entropy-balancing steps of the colorectal case study: the three per-cohort source-trial entropy-balancing steps (Stage~1 of Algorithm~\ref{alg:fresh} in each cohort) and the target-trial entropy-balancing step of the pooled, BREAKWATER-eligible particles onto BREAKWATER's Table~1; Table~\ref{tab:crc-weight-quantiles} gives the quantiles of the normalized weights. The source-trial steps cost little for the two NO16966 cohorts (ESS/$N$ of $\crcWXeloxEff$ and $\crcWFolfoxEff$), because the candidate sample was drawn on the strata the source trials share, and more for TRIBE ($\crcWTribeEff$). The target-trial step is the costly one (ESS/$N=\crcWBwEff$). It carries more constraints than any source-trial step, among them race, with BREAKWATER's $38.7\%$ Asian share (Table~\ref{tab:crc-baseline-targets}), and the CEA quantile; we did not measure how much of the ESS loss each constraint accounts for.

\begin{table}[H]
\centering\small
\setlength{\tabcolsep}{6pt}
\renewcommand{\arraystretch}{1.2}
\setlength{\tabcolsep}{4pt}
\begin{tabular}{@{}l r r r r r r@{}}
\toprule
Step & $N$ & ESS & ESS/$N$ & $w_{\max}/\bar w$ & top-$5\%$ & top-$10\%$ \\
\midrule
Stage~1, cohort \textsc{XEL} (NO16966 XELOX pooled)     & $\crcWXeloxN$ & $\crcWXeloxEss$ & $\crcWXeloxEff$ & $\crcWXeloxMax$ & $\crcWXeloxTopFive\%$ & $\crcWXeloxTopTen\%$ \\
Stage~1, cohort \textsc{FOX} (NO16966 FOLFOX4 pooled)   & $\crcWFolfoxN$ & $\crcWFolfoxEss$ & $\crcWFolfoxEff$ & $\crcWFolfoxMax$ & $\crcWFolfoxTopFive\%$ & $\crcWFolfoxTopTen\%$ \\
Stage~1, cohort \textsc{TRI} (TRIBE FOLFOXIRI+bev)      & $\crcWTribeN$ & $\crcWTribeEss$ & $\crcWTribeEff$ & $\crcWTribeMax$ & $\crcWTribeTopFive\%$ & $\crcWTribeTopTen\%$ \\
Target-trial EB, pooled $\to$ BREAKWATER (per cell) & $\crcWBwN$ & $\crcWBwEss$ & $\crcWBwEff$ & $\crcWBwMax$ & $\crcWBwTopFive\%$ & $\crcWBwTopTen\%$ \\
\bottomrule
\end{tabular}
\caption{Weight diagnostics for the entropy-balancing steps of the colorectal case study, in the format of Table~\ref{tab:weight-diagnostics}. $N$ counts particles carrying nonzero weight after the step's eligibility filter and regimen constraint; the target-trial EB row is per $(r,b)$ cell, and its weight vector is identical across cells (\S\ref{sec:crc-composition}). ESS is $(\sum_i w_i)^2/\sum_i w_i^2$; the right two columns give the share of total weight carried by the top-$5\%$ and top-$10\%$ of particles.}
\label{tab:crc-weight-diagnostics}
\end{table}

\begin{table}[H]
\centering\small
\setlength{\tabcolsep}{6pt}
\renewcommand{\arraystretch}{1.2}
\begin{tabular}{@{}l r r r r r r r@{}}
\toprule
Step & $q_{01}$ & $q_{05}$ & $q_{25}$ & $q_{50}$ & $q_{75}$ & $q_{95}$ & $q_{99}$ \\
\midrule
Stage~1, cohort \textsc{XEL}      & $\crcWXeloxQone$ & $\crcWXeloxQfive$ & $\crcWXeloxQtwentyfive$ & $\crcWXeloxQfifty$ & $\crcWXeloxQseventyfive$ & $\crcWXeloxQninetyfive$ & $\crcWXeloxQninetynine$ \\
Stage~1, cohort \textsc{FOX}      & $\crcWFolfoxQone$ & $\crcWFolfoxQfive$ & $\crcWFolfoxQtwentyfive$ & $\crcWFolfoxQfifty$ & $\crcWFolfoxQseventyfive$ & $\crcWFolfoxQninetyfive$ & $\crcWFolfoxQninetynine$ \\
Stage~1, cohort \textsc{TRI}      & $\crcWTribeQone$ & $\crcWTribeQfive$ & $\crcWTribeQtwentyfive$ & $\crcWTribeQfifty$ & $\crcWTribeQseventyfive$ & $\crcWTribeQninetyfive$ & $\crcWTribeQninetynine$ \\
Target-trial EB               & $\crcWBwQone$ & $\crcWBwQfive$ & $\crcWBwQtwentyfive$ & $\crcWBwQfifty$ & $\crcWBwQseventyfive$ & $\crcWBwQninetyfive$ & $\crcWBwQninetynine$ \\
\bottomrule
\end{tabular}
\caption{Quantiles of the normalized weight $w/\bar w$ for each entropy-balancing step of the colorectal case study (cf.\ Table~\ref{tab:weight-quantiles}). Uniform weights would have all quantiles at $1.0$.}
\label{tab:crc-weight-quantiles}
\end{table}

The three source-trial entropy-balancing steps are well-behaved ($w_{\max}/\bar w\le \crcWScaffMaxOverMean$, top-$5\%$ share $\le \crcWScaffTopFive\%$). The target-trial entropy-balancing step is the costly one, and its weights are concentrated. The top $5\%$ of particles carry $\crcWBwTopFive\%$ of the weight and the largest weight is $\crcWBwMax$ times the mean, while a left tail of near-zero weights remains (Table~\ref{tab:crc-weight-quantiles}: the $5$th percentile of $w/\bar w$ is $\crcWBwQfive$, the lower quartile $\crcWBwQtwentyfive$). The forecast therefore rests on an effective $\crcWBwEss$ particles per cell. Entropy weights are positive wherever the constraints allow; the particles that a hard constraint excludes get exactly zero weight and are dropped before the forecast is written, and are not counted in $N$. We did not tabulate which covariate combinations carry the largest or the near-zero weights.

\subsection{MH sampler convergence diagnostics (colorectal case study)}
\label{app:crc-mh-diagnostics}
Table~\ref{tab:crc-mh-diagnostics} reports, for each source cohort, the diagnostics of Table~\ref{tab:mh-diagnostics} that the colorectal chains record, summarized over the $\crcNumCells$ fold-model~$\times$~replicate chains of the cohort. The Stage-2 runner used for this run records no per-chain acceptance rate, so those rows are marked not logged. It does log, for each chain, the iteration at which it stopped and the largest hard deviation $\max_j|\widehat P_j-c_j|$ from the chain's own bootstrap-perturbed targets at that point (Table~\ref{tab:crc-mh-hyperparameters}). The achieved-vs-target deviations are computed from the calibrated outputs under the cohort's Stage-1 weights, as the mean across the $\crcNumCells$ chains of the achieved landmark survival (or median time) minus the unperturbed published target, so that the bootstrap perturbation averages out. Two of the pancreatic rows have no counterpart: the multi-chain Gelman--Rubin $\widehat R$ is not defined for the single-chain implementation these runs used, and the final multipliers were not logged for the bootstrap chains. The evidence that the colorectal chains converged therefore rests on the stopping rule and on the deviations of Table~\ref{tab:crc-mh-diagnostics}.

\begin{table}[H]
\centering\small
\setlength{\tabcolsep}{5pt}
\renewcommand{\arraystretch}{1.2}
\resizebox{\linewidth}{!}{%
\begin{tabular}{@{}l c c c@{}}
\toprule
Diagnostic & \textsc{XEL}: XELOX pooled & \textsc{FOX}: FOLFOX4 pooled & \textsc{TRI}: FOLFOXIRI+bev \\
\midrule
Acceptance rate, median across chains            & not logged & not logged & not logged \\
Acceptance rate, min/max across chains           & not logged & not logged & not logged \\
Max $\widehat R$ across constraints              & n/a (single chain) & n/a & n/a \\
Min / max of $\lambda_j^\ast$ at stop            & not logged & not logged & not logged \\
Hard deviation at stop $\max_j|\widehat P_j-c_j|$, median / max across chains & $\crcMhXeloxHardMed$ / $\crcMhXeloxHardMax$ & $\crcMhFolfoxHardMed$ / $\crcMhFolfoxHardMax$ & $\crcMhTribeHardMed$ / $\crcMhTribeHardMax$ \\
Maximum absolute achieved-vs-target landmark deviation (survival units) & $\crcXeloxLmDev$ & $\crcFolfoxLmDev$ & $\crcTribeLmDev$ \\
Achieved-vs-target median deviation, OS / PFS (days) & $\crcXeloxOsMedDev\,/\,{\crcXeloxPfsMedDev}$ & $\crcFolfoxOsMedDev\,/\,{\crcFolfoxPfsMedDev}$ & $\crcTribeOsMedDev\,/\,{\crcTribePfsMedDev}$ \\
Iterations at stop, min / median / max           & $\crcMhXeloxItersMin$ / $\crcMhXeloxItersMed$ / $\crcMhXeloxItersMax$ & $\crcMhFolfoxItersMin$ / $\crcMhFolfoxItersMed$ / $\crcMhFolfoxItersMax$ & $\crcMhTribeItersMin$ / $\crcMhTribeItersMed$ / $\crcMhTribeItersMax$ \\
Chains stopped before $T=\crcMhIters$            & $\crcMhXeloxNEarly$ of $\crcMhXeloxNChains$ & $\crcMhFolfoxNEarly$ of $\crcMhFolfoxNChains$ & $\crcMhTribeNEarly$ of $\crcMhTribeNChains$ \\
\bottomrule
\end{tabular}}
\caption{MH-sampler diagnostics for the Stage-2 chains of the colorectal case study, one column per source cohort, over that cohort's $\crcNumCells$ chains. Landmark deviations are the mean over chains of the achieved weighted survival at the landmark time minus the published level, maximized in absolute value over landmark targets; median deviations are the mean over chains of the achieved weighted median minus the published median. The hard deviation at stop and the iteration counts are read from the chains' log lines, against each chain's own bootstrap-perturbed targets, under the stopping rule of Table~\ref{tab:crc-mh-hyperparameters}.}
\label{tab:crc-mh-diagnostics}
\end{table}

\end{document}